\documentclass[aps,pra,superscriptaddress,twocolumn,longbibliography]{revtex4-2}

\usepackage[english]{babel}
\usepackage{amssymb}
\usepackage{amsmath, amsthm}
\usepackage{algorithmicx}
\usepackage{algpseudocode}
\usepackage{graphicx}
\usepackage{color}
\usepackage[colorlinks=true,
			urlcolor=blue,
			citecolor=blue,
			linkcolor=blue,
			citebordercolor={1 0 0},
			linkbordercolor={0 0 1}]{hyperref}
\usepackage{quantikz}
\usepackage{braket}
\usepackage{comment}
\usepackage{stmaryrd}
\usepackage{appendix}
\usepackage{orcidlink}

\usepackage{nameref}
\newcommand{\figref}[1]{Fig.~\ref{#1}}

\newcommand{\tr}{{\rm Tr}}

\newcommand{\cptp}{{\rm CPTP}}

\algrenewcommand\algorithmicrequire{\textbf{Input:}}
\algrenewcommand\algorithmicensure{\textbf{Output:}}
\newcounter{algorithm}
\renewcommand{\thealgorithm}{\arabic{algorithm}}
\newenvironment{algorithm}[1][]{%
    \refstepcounter{algorithm}%
    \par\medskip
    \hrule
    \smallskip
    \def\caption##1{\textbf{Algorithm \thealgorithm} ##1\par\smallskip\hrule\smallskip}
    \begin{minipage}{\linewidth}
}{%
    \end{minipage}
    \par\smallskip
    \hrule
    \par\medskip
}

\newcommand{\bea}{\begin{eqnarray}}
\newcommand{\eea}{\end{eqnarray}}
\newcommand{\be}{\begin{equation}}
\newcommand{\ee}{\end{equation}}
\newcommand{\ba}{\begin{equation}\begin{aligned}}
\newcommand{\ea}{\end{aligned}\end{equation}}
\newcommand{\eqs}[1]{\begin{align}#1\end{align}}

\newtheorem{theorem}{\textbf{Theorem}}
\newtheorem{lemma}{\textbf{Lemma}}
\newtheorem{corollary}{Corollary}
\newtheorem{proposition}{Proposition}

\newtheorem{remark}{Remark}
\newtheorem{definition}{Definition}
\newtheorem{task}{Task}

\def\1{\mathds{1}}

\def\id{\mathsf{id}}

\def\mC{\mathcal{C}}

\def\mF{\mathcal{F}}

\def\mH{\mathcal{H}}

\def\mN{\mathcal{N}}
\def\mO{\mathcal{O}}
\def\mP{\mathcal{P}}

\def\mU{\mathcal{U}}

\def\mW{\mathcal{W}}
\def\mZ{\mathcal{Z}}

\def\md{\mathfrak{D}}

\def\mf{\mathfrak{F}}

\def\mbR{\mathbb{R}}

\newcommand{\ketbra}[2]{\vert #1 \rangle \langle #2 \vert}

\makeatletter

\usepackage{circuitikz}
\usepackage{ulem}

\makeatletter
\newif\ifsupplementtoc
\newcommand{\supplementtableofcontents}{%
    \section*{Contents}
    \@starttoc{stoc}
}
\let\oldsubsection\subsection
\renewcommand{\subsection}{\@ifstar{\supp@subsectionstar}{\@dblarg\supp@subsection}}
\def\supp@subsection[#1]#2{%
    \oldsubsection[{#1}]{#2}%
    \ifsupplementtoc
        \addcontentsline{stoc}{subsection}{\protect\numberline{\thesubsection}#1}%
    \fi
}
\def\supp@subsectionstar#1{%
    \oldsubsection*{#1}%
    \ifsupplementtoc
        \addcontentsline{stoc}{subsection}{#1}%
    \fi
}
\let\oldsubsubsection\subsubsection
\renewcommand{\subsubsection}{\@ifstar{\supp@subsubsectionstar}{\@dblarg\supp@subsubsection}}
\def\supp@subsubsection[#1]#2{%
    \oldsubsubsection[{#1}]{#2}%
    \ifsupplementtoc
        \addcontentsline{stoc}{subsubsection}{\protect\numberline{\thesubsubsection}#1}%
    \fi
}
\def\supp@subsubsectionstar#1{%
    \oldsubsubsection*{#1}%
    \ifsupplementtoc
        \addcontentsline{stoc}{subsubsection}{#1}%
    \fi
}

\makeatother

\begin{document}

\title{Stabilizer Statistical Mechanics: A Framework for Efficient Quantification and Classification of Magic States}

\author{William E. Salazar}
\altaffiliation{Equal contribution}
\email{william\_esteban@u.nus.edu}
\affiliation{LG Electronics Toronto AI Lab, Toronto, Ontario M5V 1M3, Canada} 
\affiliation{Centre for Quantum Technologies, National University of Singapore, 3 Science Drive 2, Singapore 117543, Singapore}

\author{Gaurav Saxena \orcidlink{0000-0001-6551-1782}}
\altaffiliation{Equal contribution}
\email{gaurav.saxena@lge.com}
\affiliation{LG Electronics Toronto AI Lab, Toronto, Ontario M5V 1M3, Canada} 

\author{Jack S. Baker \orcidlink{0000-0001-6635-1397}}
\email{jack.baker@lge.com}
\affiliation{LG Electronics Toronto AI Lab, Toronto, Ontario M5V 1M3, Canada} 

\author{Leong Chuan Kwek}
\email{kwekleongchuan@nus.edu.sg}
\affiliation{Centre for Quantum Technologies, National University of Singapore, 3 Science Drive 2, Singapore 117543, Singapore}
\affiliation{
MajuLab, CNRS-UNS-NUS-NTU International Joint Research Unit, UMI 3654, Singapore 117543, Singapore}
\affiliation{National Institute of Education, Nanyang Technological University, 1 Nanyang Walk, Singapore 637616, Singapore}

\author{Thi Ha Kyaw \orcidlink{0000-0002-3557-2709}}
\email{thiha.kyaw@lge.com}
\affiliation{LG Electronics Toronto AI Lab, Toronto, Ontario M5V 1M3, Canada}

\date{\today}

\begin{abstract} 
The partition function is statistical mechanics' answer to an exponentially large spectrum, distilling it into a single analytic object whose temperature dependence resolves the full structure of the underlying ensemble. We show that magic, the resource separating universal quantum computation from classically simulable stabilizer dynamics, admits precisely such a description. Mapping the Pauli spectrum of a quantum state onto the energy levels of a fictitious many-body system, the Pauli gas, we construct its canonical partition function, the stabilizer partition function, and from its associated free energy a magic monotone that we call the stabilizer work. Both these objects are analytic functions of an inverse-temperature-like parameter and are efficiently estimable via Bell sampling. The framework is analytically tractable. We derive exact ensemble-averaged partition functions for Haar-random, $\nu$-compressible, and pseudomagic states, together with concentration guarantees. We show that for every value of its parameter, the stabilizer work is a faithful, subadditive magic monotone, while remaining efficiently accessible on quantum hardware and admitting an operational interpretation. Unlike measures that probe a single moment of the Pauli distribution, the stabilizer work is intrinsically moment-generating. As the temperature is tuned from high to low, it interpolates continuously between the stabilizer 2-R\'enyi entropy and the stabilizer nullity, revealing two previously disconnected monotones as limiting cases of a single object. We demonstrate the framework on low-rank stabilizer simulation, resource interconversion, molecular ground states, and quantum many-body systems. A thermodynamics of magic is therefore not merely an analogy but a working toolkit, opening a statistical-mechanical route to magic properties of quantum systems that no single measure can access.
\end{abstract}

\maketitle

\section{Introduction}\label{sec:introduction}

\begin{figure*}[t]
    \centering
    \includegraphics[width=0.9\linewidth]{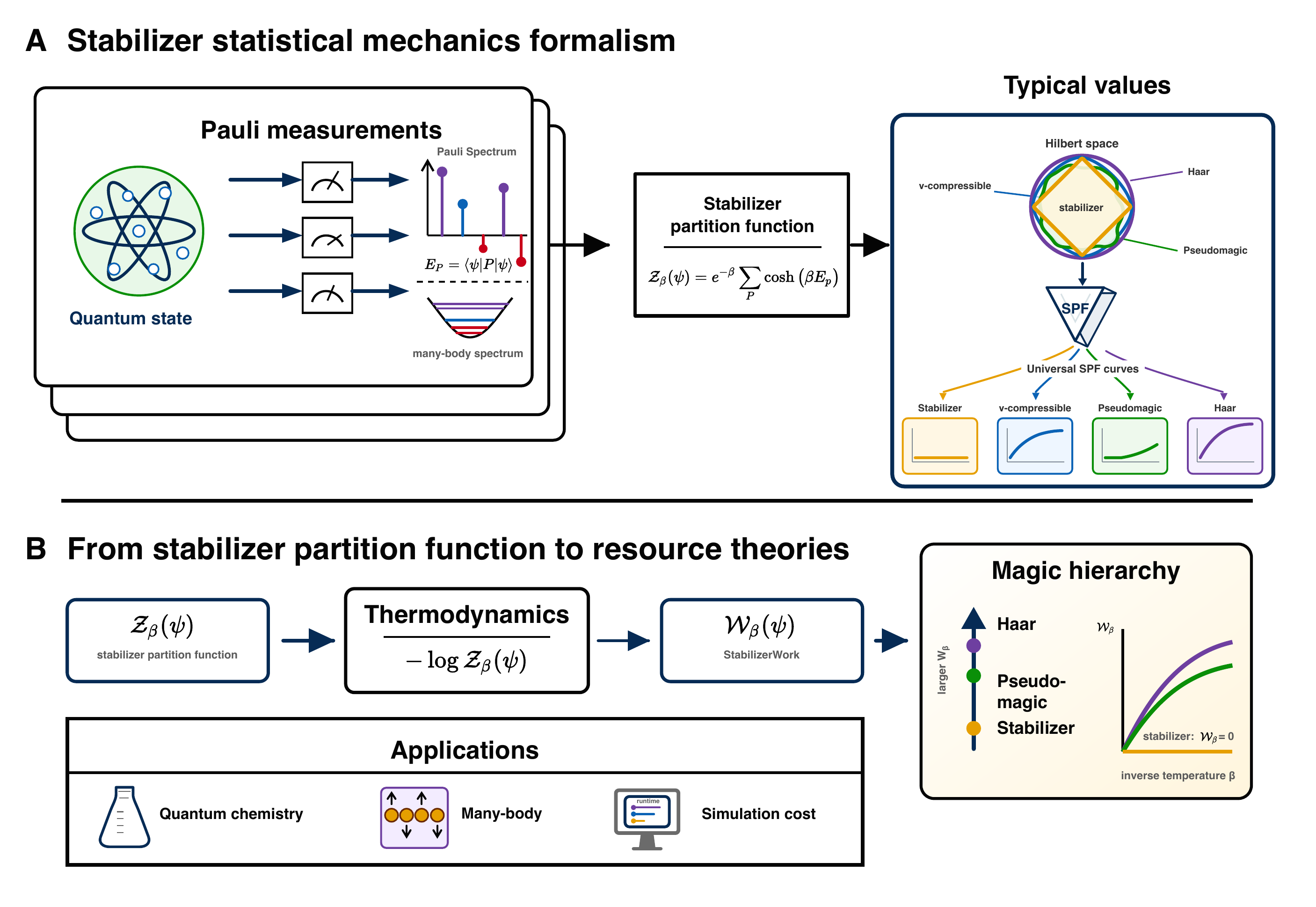}
    \caption{Overview of the stabilizer statistical mechanics framework.
    (A)~A quantum state $\psi$ is characterized by its Pauli spectrum, the collection of expectation values $\tr[\psi P]$ over all Pauli strings $P$. We reinterpret this exponentially large collection as the energy spectrum of a complex many-body system, the Pauli gas. The associated canonical partition function is the stabilizer partition function (SPF) $\mathcal{Z}_{\beta}$ (Definition~\ref{def:Pauli_transform}), which compresses the full statistics of the Pauli spectrum into a single analytic function of an inverse-temperature-like parameter $\beta$, in the same way that a thermodynamic partition function encodes the statistics of microstates. The SPF is efficiently estimable by Bell sampling (Theorem~\ref{theorem:SPF_Complexity}) and takes a universal form for representative families across Hilbert space: stabilizer states (Proposition~\ref{Prop:Uniqueness_stabilizer_states}), separable states (Theorem~\ref{theorem: separable_states_P_transform}), $\nu$-compressible states (Theorem~\ref{theorem: v_compressible_Pauli_T}), Haar-random states (Theorem~\ref{theorem: Pauli_T_Haar_random}), and pseudomagic states (Theorem~\ref{Theorem:Pseudorandom_SPF}), making the SPF a classifier of magic states.
    (B)~Following the thermodynamic dictionary, the free energy of the Pauli gas, referenced against its unique stabilizer value, defines a magic monotone that we call the stabilizer work $\mathcal{W}_{\beta}$ (Section~\ref{sec:SW}). We show that it has desirable resource-theoretic properties.  It orders states into a magic hierarchy from stabilizer states ($\mathcal{W}_{\beta}=0$) through pseudomagic to Haar-random states. We apply the framework to low-rank stabilizer simulation, resource interconversion, molecular ground states in quantum chemistry, and many-body physics (Section~\ref{sec:Applications}).}
    \label{fig:hero}
\end{figure*}

Universal fault-tolerant quantum computation rests on a striking resource asymmetry. Stabilizer states, Clifford operations, and Pauli measurements form a rich computational subtheory that nevertheless admits efficient classical simulation, irrespective of qubit count or circuit depth~\cite{Gottesman/Stabilizer_Dynamics/1998,Aronsson_gottesman/Tableau/2004}. 
Universality requires access to nonstabilizer resources, or quantum magic, which must be injected through non-Clifford gates, through magic states consumed by gate teleportation~\cite{Bravyi/magic_universality/2005,Gottesman/UniversalQC_teleportation/1999}, or through non-Pauli measurements~\cite{Li/magic_injection_measurements/2024}. In leading fault-tolerant architectures, the preparation of these resources via magic-state distillation dominates the physical-qubit and runtime budgets~\cite{Bravyi/magic_universality/2005,Bravyi2012Nov,Fowler2013Jun,Campbell/FTQC_MSD/2017,Gidney2019Apr,Litinski2019Mar,Litinski2019Dec,Zheng2025Sep,Wills2025Nov,Ruiz2026Feb}, and counts of non-Clifford gates or magic states have consequently become the de facto accounting unit for both experimental feasibility and classical-simulation hardness~\cite{Bravyi/Clifford_plus_T/2016,Bravyi/stab_fidelity/2019,KissingerWetering2020,Monaco2025Sep,Kook2026Jun}. This accounting is not academic. Fault-tolerant resource estimates across quantum chemistry and materials science~\cite{ReiherEtAl2017,Babbush/electronic_spectra/2018,BerryEtAl2019Qubitization,VonBurgEtAl2021Catalysis,SuEtAl2021FirstQuant,LeeEtAl2021THC,KimEtAl2022Battery,SteudtnerEtAl2023Observables,KivlichanEtAl2020Condensed,RubinEtAl2023Bloch,KanSymons2025Hubbard,BaySmidtEtAl2025GeneralizedHubbard,BakerEtAl2024SpinDefects,CasaresEtAl2026ODMR,Naranjo2026}, cryptanalysis~\cite{RoettelerEtAl2017ECDLP,GidneyEkera2021,GouzienEtAl2023ECC,GarnKan2025BinaryECC}, and differential equations~\cite{ChildsLiuOstrander2021PDE,PenuelEtAl2024CFD,DemirdjianEtAl2026Burgers} are routinely dominated by their non-Clifford component and are regularly used to decide whether a computation is physically realizable at all.

Gate counts, however, are inherently representation and implementation dependent. A nonstabilizer circuit can be recast as a stabilizer circuit acting on magic-laden inputs, or, as in measurement-based quantum computation, the magic can be shifted into measurements~\cite{Raussendorf/One-wayQC/2001,Briegel2009Jan}. In other words, the same computational resource can be redistributed among state preparation, gates, ancillas, and readout without changing the underlying task. Optimizing this resource therefore requires translating circuit-level counts into a circuit-independent currency, namely the intrinsic magic content of states and operations. The resource theory of magic provides the natural language for this translation~\cite{Veitch_2014,VeitchEtAl2012Negativity,RevModPhys.91.025001} by specifying the free states (stabilizer states), the free operations~\cite{Veitch_2014,PhysRevA.97.062332,Seddon2019Jul,PhysRevA.106.042422}, and monotones that cannot increase under them. Quantifying magic is thus not a formal exercise but a prerequisite for resource optimization in fault-tolerant quantum computing.

Alongside quantification sits the equally consequential task of classification, i.e., deciding, often from limited measurement data, to which resource class a given state belongs. Far from bookkeeping, this classification governs what can be done with a state. Classical simulation algorithms have drastically different costs across magic classes, so the class of a state decides which simulator, if any, can verify it.
Moreover, benchmarking error-corrected hardware requires certifying that a nominally magic-bearing logical state indeed carries the intended nonstabilizerness. 
How consequential such classification can be was recently made explicit by Gu, Oliviero, and Leone~\cite{Gu/Magic_vs_Ent/2024}, who showed that the balance between a state's magic and its entanglement splits Hilbert space into two \textit{computational} phases.
While entanglement estimation and manipulation admit efficient algorithms throughout the entanglement-dominated phase, the same tasks are provably intractable in the magic-dominated phase (even though the states there are also entangled).
Alongside these developments, sharp constraints on classification capabilities have emerged. 
Pseudomagic ensembles~\cite{Gu/pseudomagic/2024,Aaronson/pseudoentanglement/2022} exhibit vastly different magic content yet remain indistinguishable to any computationally bounded observer, showing that magic can be computationally hidden. Together, these developments turn a practical requirement into a principled one, i.e., a useful magic measure should come with its classification power, and its limits, exactly characterized.

Measured against these requirements, the existing landscape of magic measures exhibits a persistent trade-off. Operationally meaningful measures such as the robustness of magic~\cite{Howard/RoM/2017,Heinrich/RoM_Fast/2019,Hamaguchi/RoM_practical/2024}, stabilizer extent and rank~\cite{PhysRevX.6.021043,Bravyi/stab_fidelity/2019}, and relative entropies of magic~\cite{Veitch_2014,PRXQuantum.2.010345,BGJ_PNAS2023_QuantumEntropy,Rubboli2024mixedstate,Tarabunga_2025_QST} bound classical-simulation costs and distillation rates, but their evaluation requires optimizations over stabilizer decompositions whose size grows superexponentially with qubit number, rendering them intractable beyond a few qubits and inaccessible to direct measurement. 
The computable measures, most prominently the $\alpha$-stabilizer R\'enyi entropies (SREs)~\cite{Leone/Stab_Renyi/2022}, are closed-form functionals of the Pauli spectrum which is a collection of expectation values of all Pauli strings.
These SREs, for integer index $\alpha\geq2$, are measurable via Bell sampling~\cite{Haug/Scalable_measures/2023,Tobias/Pauli_sampling/2024}. This has made them the workhorse of a rapidly growing many-body literature~\cite{Oliviero/IsingModelMagic/2022,Tarabunga/magic_via_PauliMarkovChains/2023,Turkeshi/MagicSpreading/2025,Sierant/FermionicMagic/2026}. Yet the SREs come with structural caveats.
A fixed R{\'e}nyi index probes a particular moment of the Pauli spectrum, whereas the Pauli spectrum contains exponentially many expectation values whose detailed organization can differ even when low-order moments agree.
The monotonicity under stabilizer protocols has been established only for $\alpha\geq2$~\cite{Leone/Stab_Renyi_monotone/2024} and remains unresolved for $\alpha<2$. 
Furthermore, for highly magical states the underlying stabilizer purities become exponentially small, so that estimating an SRE to fixed additive accuracy becomes resource-intensive exactly in the regime where the SRE is largest. At the coarsest extreme sits the stabilizer nullity~\cite{Beverland/Pauli_Spectrum/2020}, a faithful, additive, integer-valued monotone that lower-bounds $T$ counts~\cite{Jiang/T_cont_vs_nnullity/2023} but cannot distinguish any two states with trivial stabilizer group, which is to say, almost all states. 
This creates a double tension in the current landscape.
The quantities that can be computed and measured either carry structural caveats as monotones or are too coarse to resolve the states of interest, while the quantities with impeccable resource-theoretic credentials cannot be evaluated at scale.

A subtler limitation is shared by all of these measures. Any measure, by definition, projects a state onto a single number, but the nonstabilizerness of a many-body state resides in an entire distribution, namely the exponentially many entries of its Pauli spectrum. Two states can agree on the second-order SRE yet organize their magic in structurally different ways, either concentrated in a few large Pauli weights or diffused across an exponentially flat bulk, and no fixed R\'enyi index separates such structures. 
Entanglement theory resolved an analogous limitation by moving beyond the entanglement entropy to the full entanglement spectrum~\cite{Li_Haldane/Entanglement_Spectrum/2008}. 
Magic calls for a similar refinement, not another scalar, but a principled family whose full profile serves as the diagnostic.
To this purpose, statistical mechanics offers a natural framework, as it was built precisely to extract macroscopic, operationally meaningful quantities from exponentially large spectra, and the Pauli spectrum is exactly such an object.
It is this observation that we develop in the present work.

This is not a coincidence. 
The historical relationship between information theory and thermodynamics suggests where such a family should come from. Since Shannon~\cite{Shannon1948}, complex microscopic structures have repeatedly become operationally transparent when organized through entropies, spectra, and thermodynamic potentials. R\'enyi generalized entropy into a continuous family whose order tunes the sensitivity to different regions of a distribution~\cite{Renyi1961}, entanglement entropy acquired its asymptotic meaning through concentration and dilution~\cite{BennettEtAl1996Entanglement}, and quantum thermodynamics extended the free energy into families of generalized free energies and R\'enyi-divergence conditions governing state conversion~\cite{HorodeckiOppenheim2013,BrandaoEtAl2015SecondLaws}, with deep connections between convertibility, relative entropy, and thermodynamic structure holding across resource theories~\cite{BrandaoGour2015ResourceTheories}. 
Statistical mechanics embodies this lesson in its sharpest form. A partition function does not select a moment of the spectrum in advance.
Instead, its dependence on temperature continuously changes which parts of the microscopic spectrum dominate, thereby generating a hierarchy of spectral information from a single object. 
What remains comparatively underdeveloped for magic is precisely such a statistical-mechanical treatment, a generating object that treats the Pauli spectrum as an ensemble, resolves its structure across a continuous range of scales, and recovers established resource measures in controlled limits.

Developing this object is the central aim of the present work. We introduce a statistical-mechanical framework that studies quantum magic directly through the structure of the Pauli spectrum. By associating the spectrum with an effective statistical ensemble, we define a \textit{stabilizer partition function} and use it to build a thermodynamic description of nonstabilizerness. 
This perspective provides a unified way to probe magic across different spectral scales, connecting resource-theoretic questions to the statistical structure of quantum states, and it yields both new tools for quantifying magic and new insight into how magic is distributed across Hilbert space.
The magic measure that we propose, namely the \textit{stabilizer work}, is the free-energy difference between the state's ensemble and the unique stabilizer reference, and comes equipped with a parameter that plays the role of an inverse temperature and acts as a genuine resolution knob.
We show that the stabilizer work occupies the corner of the trade-off left vacant by existing measures. 
The thermodynamic dictionary is not decorative. It dictates the correct form of the measure and renders its resource-theoretic properties transparent, including faithfulness, Clifford invariance, subadditivity, and a third-law interpretation of the zero-temperature limit.
As the temperature is tuned from high to low, the stabilizer work interpolates continuously between the stabilizer 2-R\'enyi entropy and the stabilizer nullity.
The nullity is known to upper-bound the SREs, but that is an inequality between separately defined quantities.
Here, the two emerge as the high- and low-temperature limits of a single thermodynamic object, unifying two previously disconnected monotones and elevating the full temperature profile of the stabilizer work to a \textit{magic spectrum} carrying strictly more information than either endpoint.
We illustrate the utility of the framework in settings ranging from pseudomagic and classical simulation to quantum chemistry and many-body physics.

The remainder of this paper is summarized in Fig.~\ref{fig:hero} and organized as follows. 
Section~\ref{sec:preliminaries} fixes the notation and reviews the Pauli spectrum description of quantum states together with stabilizer nullity and stabilizer R\'enyi entropies. 
Section~\ref{sec:SPF} introduces the stabilizer partition function and develops the Pauli-gas correspondence, including its relation to Pauli moments, its computational properties, and its principal bounds and continuity results. 
In Section~\ref{sec:SPF_across_Hilbert_space}, we use this framework to characterize representative regions of Hilbert space and derive analytic results for product states, Haar-random states, $\nu$-compressible states, and pseudomagic ensembles. 
Section~\ref{sec:SW} turns from characterization to resource quantification by introducing the stabilizer work and establishing its resource-theoretic properties, its high- and low-temperature limits, continuity, and its ability to distinguish different structures of magic
across the temperature-dependent spectrum.
We then develop concrete applications in Section~\ref{sec:Applications}, connecting the framework to low-rank stabilizer simulation, resource interconversion, and physically motivated states in quantum chemistry and many-body physics. 
Finally, Section~\ref{sec:discussion} summarizes the work and the broader implications of the statistical-mechanical perspective, discusses extensions to mixed states and quantum channels, and identifies directions for future investigation.

\section{Preliminaries}\label{sec:preliminaries}

In this section we fix the notation and basic concepts used throughout the paper.
Unless stated otherwise, the main development concerns pure states of $n$-qubit systems.
We denote by $\mathcal{H}_{n}\cong (\mathbb{C}^{2})^{\otimes n}$ the Hilbert space of $n$ qubits and $d_n=2^n$ for its dimension.
The set of density operators on $\mH_n$ is denoted $\md(\mathcal{H}_{n})$.
The mixed states will be denoted by $\rho,\sigma,\tau \in \md(\mH_n)$, whereas $\ket{\psi},\ket{\phi},\ket{\chi}$ will be used to denote pure states.
We freely identify a pure state with its rank-one projector, for instance $\psi:=\ket{\psi}\bra{\psi}$.

For a composite system, subsystems are labeled by capital letters $A,B,C$, etc..
If $A$ contains $n$ qubits, then $\mathcal{H}_{A}\cong (\mathbb{C}^{2})^{\otimes n}$ and $d_{A}=|A|=2^n$.
Given a bipartite state $\psi_{AB}\in\mH_n$, its reduction to A is $\rho_A = \tr_B[\psi_{AB}]$
Unless stated otherwise, logarithms entering information-theoretic quantities are taken to base
two.

\subsection{Pauli and Clifford conventions}

Let $\mathcal{P}_{1}=\{\mathbb{I},X,Y,Z\}$ denote the single-qubit Pauli labels, and define
\begin{equation}
\mathcal{P}_{n}:=\mathcal{P}_{1}^{\otimes n}
\end{equation}
as the phase-free set of $n$-qubit Pauli strings.
When phases matter, we use the full Pauli group $\widetilde{\mathcal{P}}_{n}:=\{\pm 1,\pm i\}\mathcal{P}_{n}$.
Throughout the paper, global phases are irrelevant unless explicitly stated, and Pauli operators are understood modulo phase whenever they are used only as labels.

Given an abelian subgroup $\mathcal{S}\subset \widetilde{\mathcal{P}}_{n}$ that does not contain $-\mathbb{I}$, a pure state $\ket{\psi}$ is said to be stabilized by $\mathcal{S}$ if
\begin{equation}
P\ket{\psi}=\ket{\psi}, \qquad \forall\, P\in \mathcal{S}.
\end{equation}
We say that $\ket{\psi}$ is a stabilizer state if it is the unique common $+1$ eigenstate of an abelian subgroup of $\widetilde{\mathcal{P}}_{n}$ generated by $n$ independent commuting Pauli operators.
The set of all $n$-qubit stabilizer states is denoted by $\mathrm{Stab}_{n}$, and the stabilizer group of $\ket{\psi}$ is denoted by $\mathrm{Stab}(\psi)$.
Although the choice of generators is not unique, stabilizer states admit an efficient description in terms of any generating set, for example through the tableau formalism~\cite{Aronsson_gottesman/Tableau/2004}.

Closely related to the Pauli group is the $n$-qubit Clifford group, defined as
\begin{equation}
\mathcal{C}l_{n}:=\{U\in U(2^{n})\,:\, U\mathcal{P}_{n}U^{\dagger}=\mathcal{P}_{n}\}.
\end{equation}
That is, Clifford unitaries map Pauli operators to Pauli operators under conjugation, up to phase.
A canonical generating set for $\mathcal{C}l_{n}$ is given by the Hadamard ($H$), phase ($S$), and controlled-NOT (CNOT) gates.

Although the Clifford group is exponentially large, operations restricted to Clifford unitaries acting on stabilizer states aren't enough for universal quantum computation. Achieving universality requires the inclusion of resources beyond the Clifford framework, commonly referred to as nonstabilizer or magic.

\subsection{The Pauli Spectrum}

With these conventions in place, we can introduce one of the main objects of this work, namely the Pauli spectrum.
\begin{definition}[Pauli Spectrum]
For an $n$-qubit pure state $\ket{\psi}$, we define its Pauli spectrum as the collection of $4^{n}$ real numbers indexed by $P \in \mathcal{P}_{n}$,
\begin{equation*}
    {\rm Spec}(\psi):=\{\, \tr(\psi P) \;\; \text{for} \;\; P \in \mathcal{P}_{n}\}\,.
\end{equation*}
\end{definition}

Compared with the convention introduced in \cite{Beverland/Pauli_Spectrum/2020}, we retain the sign of each expectation value and work with $\tr(\psi P)$ rather than $|\tr(\psi P)|$.
This signed convention is the natural one for the framework developed in the next sections when we discus the Pauli spectrum for typical sets of states, while the absolute-value version can be recovered whenever needed~\footnote{Strictly speaking, ${\rm Spec}(\psi)$ should be viewed as a multiset, since multiplicities matter}.

The usefulness of the Pauli spectrum is already visible at the single-qubit level.
Any pure one-qubit state can be written as $\psi=(I+r_xX+r_yY+r_zZ)/2$ with $r_x^2+r_y^2+r_z^2=1$, so its Pauli spectrum is simply $\{1,r_x,r_y,r_z\}$.
Pure stabilizer states are exactly the six axis-states on the Bloch sphere, for which the spectrum takes the extremal form $\{1,\pm 1,0,0\}$ (up to permutations).
By contrast, a magic state such as $\ket{T}$ already exhibits non-integer entries, with $\text{Spec}(\ket{T})=\{1,\frac{1}{\sqrt{2}},\frac{1}{\sqrt{2}},0\}$.
This simple example already shows how the Pauli spectrum separates stabilizer structure from genuinely nonstabilizer behavior, establishing it as a magic witness.

Notice that the number of $\pm1$ entries in the Pauli spectrum of a given state gives the size of its stabilizer group, i.e., $|\text{Stab}(\psi)|$. For an $n$-qubit stabilizer state, its spectrum can be characterized using the following proposition.
\begin{proposition}[Stabilizer Pauli Spectrum]\label{proposition:Stab_Pauli_spectrum}
The Pauli spectrum of a pure $n$-qubit stabilizer state of a system $A$ consists of $d_{A}^{2}-d_{A}$ zeros and $d_{A}$ plus/minus ones.
\end{proposition}
Notice that the Pauli spectrum, as defined above, is not invariant under general Clifford unitaries, that is, $\text{Spec}(\psi) \neq \text{Spec}(\mathcal{C}(\psi))$ for any arbitrary Clifford unitary $\mathcal{C}$.
This follows directly from the definition of the Clifford unitaries and the Pauli spectrum. 
Since the Clifford unitaries transform one Pauli into another, this may result in the permutation and/or sign-flip of the Pauli spectrum elements of the state.
For stabilizer states, Clifford action results in permutation and possible sign changes of plus/minus ones in the spectrum, while leaving the size of the stabilizer group unchanged. 
Consequently, the stabilizer group size can be readily inferred from the spectrum by counting all the entries with $|x|=1 \in \text{Spec}(\psi)$.

We will also use two complementary partitions of the Pauli data. At the level of Pauli labels, define
\begin{align}
    \mathrm{Supp}(\psi)
    &:={\left\{P\in\mathcal{P}_n:\tr[\psi P]\neq0\right\}},\\
    \mathrm{Ker}_{\mathcal P}(\psi)
    &:={\left\{P\in\mathcal{P}_n:\tr[\psi P]=0\right\}}.
\end{align}

\begin{proposition}[Spectrum Splitting]
\label{Prop:rank_nullity_construction}
Every pure state $\psi$ induces the disjoint decomposition
\begin{equation}
    \psi \to \mathcal{P}_{n} = {\rm Supp}(\psi) \cup {\rm Ker}_{\mathcal{P}}(\psi)
\end{equation}
\end{proposition}

Clearly ${\rm Stab}(\psi) \subseteq {\rm Supp}(\psi)$, where by Proposition~\ref{proposition:Stab_Pauli_spectrum} the equality is obtained for stabilizer states. For arbitrary pure states, the cardinality of both the support and the kernel, are related via the nullity-rank theorem, ${\rm i.e.}$ $|{\rm Supp}(\psi)|+|{\rm Ker}_{\mathcal{P}}(\psi)|=d_{n}^{2}$, and in general the Pauli spectrum will consist of $|{\rm Ker}_{\mathcal{P}}(\psi)|$-zeros.

\subsection{Measures of nonstabilizerness}

Quantifying the magic of a given state is a central task in the resource theory of magic.
It characterizes the state's capacity to enable quantum computational advantage, optimizes the use of such resources in computation, and elucidates the quantum complexity of phases of matter, etc.
To this purpose, several magic measures have been introduced in the literature \cite{Beverland/Pauli_Spectrum/2020, Dutta/Wigner_qubit/2026,Howard/RoM/2017,Leone/Stab_Renyi/2022,Tarabunga/spectrum_flatness/2024, Heinrich/RoM_Fast/2019,Hamaguchi/RoM_practical/2024, Bravyi/stab_fidelity/2019, Veitch_2014, PRXQuantum.2.010345, PhysRevA.106.042422, BGJ_PNAS2023_QuantumEntropy, Rubboli2024mixedstate, Tarabunga_2025_QST, Jiang/T_cont_vs_nnullity/2023,Bravyi/Clifford_plus_T/2016}. 
These can be broadly classified into three categories, which we will refer here as 
distance-based measures, complexity-based measures, and statistics-based methods. 
Distance-based magic measures quantify non-stabilizerness as a notion of separation from the set of stabilizer states by employing a suitable metric to determine the \textit{minimum distance} between a given magic state and the convex hull of stabilizer states.
These measures include the robustness of magic (RoM) \cite{Howard/RoM/2017,Heinrich/RoM_Fast/2019,Hamaguchi/RoM_practical/2024}, the stabilizer extent \cite{Bravyi/stab_fidelity/2019}, and other relative entropies of magic such as the ones introduced in~\cite{Veitch_2014, PRXQuantum.2.010345, PhysRevA.106.042422, BGJ_PNAS2023_QuantumEntropy, Rubboli2024mixedstate, Tarabunga_2025_QST}.
Complexity-based ones quantify magic by counting the minimum number of non-stabilizer resources needed to prepare the state; these measures carry a straightforward operational meaning by directly measuring the amount of required magic resources.
In the Clifford+T model, an example would be given by the T-count \cite{Jiang/T_cont_vs_nnullity/2023,Bravyi/Clifford_plus_T/2016}.
Statistical ones on the other hand use the Pauli spectrum as the building block to quantify non-stabilizerness, two of which we present below.

\subsubsection{Stabilizer Nullity}

Given an $n$-qubit state $\ket{\psi}$ stabilized by $2^{n-\nu}$ Pauli strings, we say that the state has $\nu$-nullity~\cite{Beverland/Pauli_Spectrum/2020}, denoted by $\nu(\psi)=\nu$.
The nullity quantifies the size of the generators of the stabilizer group according to $\log_2(\rm |Stab(\psi)|)=n-\nu(\psi)$. 
Equivalently, it is obtained from the Pauli spectrum by counting the number of entries equal to $\pm 1$. 
Furthermore, the $\nu$-nullity provides bounds on the size of the kernel.
\begin{proposition}
    The Pauli Kernel of a state $\psi$ has the following bounds
    \begin{equation}
    2^{n}(2^n-2^{\nu(\psi)}) \leq |{\rm Ker}_{\mathcal{P}}(\psi)|\leq 4^{n}-2^{n}\,.
\end{equation}
\end{proposition}
\begin{proof}
    The lower bound follows from the inequality $|{\rm Supp}(\psi)|= d 2^{M_0(\psi)}\leq d2^{\nu(\psi)}$~\cite{Leone/Stab_Renyi_monotone/2024}, where $M_0$ denotes the zero-order stabilizer R\'enyi entropy.
    From the observation that the number of zeroes in the Pauli spectrum is at most the total size of the spectrum minus the size of the stabilizer group of the state, we can get an upper bound in terms of nullity as follows
    \begin{equation}
    2^{n}(2^n-2^{\nu(\psi)}) \leq |{\rm Ker}_{\mathcal{P}}(\psi)|\leq 4^{n}-2^{n-\nu(\psi)}\,.
\end{equation}
    However, a tighter more fundamental bound exists which is achieved for the stabilizer states case (as $\nu=0$ for stabilizer states) which is given by
    \eqs{
    |{\rm Ker_{\mathcal{P}}(\psi)}|\leq 4^{n}-2^{n}\,,
    }
    hence completing the proof.
\end{proof}

\subsubsection{Stabilizer R\'enyi entropy }

Given the Pauli spectrum of a quantum state, one can naturally associate a probability distribution with it. This distribution allows the definition of the R\'enyi entropy of order $\alpha$, known as stabilizer R\'enyi entropy~\cite{Leone/Stab_Renyi/2022}. For completeness, we recall its definition below.
In this way, the Pauli spectrum plays a central role, with its moments being captured by the stabilizer R\'enyi entropies.

\begin{definition}[Stabilizer R\'enyi entropy \cite{Leone/Stab_Renyi/2022}]

Given an $n$-qubit pure state $\ket{\psi}\in A$, the probability distribution associated with its Pauli spectrum is given by $\Xi(\psi)=\{ \frac{x^2}{d_A}, \; \text{for}\; x \in \text{Spec}(\psi_A)\}$. The stabilizer $\alpha$-R\'enyi entropy of $|\psi_A\rangle$ is defined as
\begin{equation}
    M_{\alpha}(\psi_A):=\frac{1}{1-\alpha}\log \left(\sum_{x \,\in \,\text{Spec}(\psi_A)} \frac{x^{2\alpha}}{d_A}  \right)\,.
\end{equation}
%
\end{definition}

The stabilizer R\'enyi entropy constitutes a well-behaved magic measure from a resource theoretic perspective. Specifically, it is (i) faithful, in the sense that $M_{\alpha}(\psi)=0$ if and only if $\psi$ is a stabilizer state; (ii) stable under Cliffords, i.e., $M_{\alpha}(C\ket{\psi})=M_{\alpha}(\psi)$ for all $C \in \mathcal{C}l_{n}$, and (iii) additive, i.e.,
$M_{\alpha}(\psi\otimes \phi)= M_{\alpha}(\psi)+M_{\alpha}(\phi)$. Furthermore, for $\alpha \geq 2$ the stabilizer R\'enyi entropy is a magic monotone~\cite{Leone/Stab_Renyi_monotone/2024}.

While the $\nu$-nullity captures the size of the stabilizer group, the stabilizer R\'enyi entropies capture finitely many moments,
and it provides an upper bound to all the other R\'enyis via the R\'enyi hierarchy, i.e., $M_{\alpha}(\psi)\leq M_0 = \log(|{\rm Supp}(\psi)|/d_n) \leq \nu(\psi)$.

In a similar fashion, $\text{RoM}$ and $T$-count provide upper bounds to Stabilizer R\'enyi entropy, that is, $M_{\alpha}(\psi) \leq 2\log(\text{RoM}(\psi))$ for all $\alpha \geq \frac{1}{2}$ and $M_{\alpha}(\psi) \leq T(\psi)$ for all $\alpha \geq \frac{1}{2}$ for the T-count. 

As pointed above, the stabilizer R\'enyi entropies only capture finitely many moments. 
In the next section, we introduce a function, the Stabilizer partition function, that captures all the moments of the Pauli spectrum.

\section{The Stabilizer Partition Function}\label{sec:SPF}

The Pauli spectrum provides a clear distinction between stabilizer and nonstabilizer states, making it appealing as a tool for characterizing, classifying, and quantifying nonstabilizer states.
Despite its intuitive appeal, the intrinsic theoretical and practical complexity of the Pauli spectrum presents a significant obstacle.
From a theoretical perspective, the Pauli spectrum comprises an exponentially large list of seemingly random numbers, $x\in [-1,1]$ that, as pointed out earlier, lacks invariance under Clifford unitaries.
This means that this exponentially large list of numbers will get scrambled by operations that do not increase nonstabilizerness.
From the practical standpoint, the situation is no better, as determining all $4^{n}$ entries of the Pauli spectrum for a given state necessitates an exponential number of measurements, leading to exponentially large sample and shot complexity.

The above challenges can nevertheless be circumvented for specific tasks, such as quantification of nonstabilizerness, by recognizing that the relevant figure of merit is not the full Pauli spectrum itself but rather its statistical properties, specifically its moments.
For instance, the $\alpha$-stabilizer Rényi entropies (for $\alpha \in \mathbb{N}$) quantify magic by encoding the even moments of the Pauli spectrum. 
However, knowledge of only a few moments provides only partial information about the underlying distribution. 
For instance, the knowledge of mean and variance of a random variable does not determine finer features of the distribution, such as its skewness.
Similarly, a finite set of moments (equivalently, a finite set of stabilizer Rényi entropies) only partially unveils the Pauli spectrum and other inherent properties. 
Motivated by this analogy, the following natural question arises.
Assuming that the statistical structure of the Pauli spectrum encodes most or all of the information relevant to nonstabilizerness of a state, how can we probe this inherently complex statistical structure in a unified manner?

To address this question, we adopt a perspective inspired from thermodynamically large systems and reinterpret the Pauli spectrum through the lens of many-body statistical mechanics.
Below, we introduce two complementary perspectives that map the magic of a quantum system to statistical properties of many-body systems.

\subsection{Many-body system perspectives of Pauli Spectrum}

\paragraph{Perspective 1.}

The exponentially large number of entries of the Pauli spectrum can be interpreted as the energy levels of $2\times 4^{n}$ many-body system. To set grounds consider the POVM set $\{\Pi_{P}^{\pm}=\frac{1}{2d_{n}^{2}}(\mathbb{I}\pm P) : P \in \mathcal{P}_{n}\}$. The Pauli spectrum is obtained from the empirical distribution of this POVM in the large number of shots limit. To each of of the elements of the POVM set we can associate a state in
$\mathcal{H}_{n}^{\otimes2} \otimes \mathbb{C}^{2}$
\begin{equation}
    \Pi_{P}^{s} \leftrightarrow \left|P,s \right):=(\mathbb{I}\otimes P)\ket{\Phi^+} \otimes \ket{s}, \quad \forall \quad P \in \mathcal{P}_{n} 
\end{equation}
where $\ket{\Phi^+}= d_{n}^{-\frac{1}{2}}\sum_{l} \ket{l,l}$ is the maximally entangled state in the double space $\mathcal{H}_{n}^{\otimes 2}$, and $s$ is an auxiliary state tracking the ``polarization'' degree of freedom $s$ of the POVM. The states are orthogonal $(\tilde{s},\tilde{P}\left|P,s \right)=\delta_{s \tilde{s}} \delta_{P \tilde{P}}$ and form a basis for $\mathcal{H}_{n}^{\otimes2} \otimes \mathbb{C}^{2}$. Given an $n$-qubit state $\psi$ we define the Stabilizer Hamiltonian as
\begin{equation}
    \hat{H}(\psi):=\sum_{P \in \mathcal{P}_{n}, s \in \{-1,1\}}(1-s\,\tr (\psi P))\left| P, s\right)\left( P,s\right|
\end{equation}
The Stabilizer Hamiltonian preserve the parity of the $[\hat{H}(\psi), \mathbb{I}^{\otimes2} \otimes \hat{n}]=0$ and is block diagonal in the double space, ${\rm e.g.}$ $\hat{H}(\psi)=\hat{H}_{+}(\psi)\oplus \hat{H}_{-}(\psi)$. (We leave for the supplemental materials the details of the construction.)
In this ``Pauli gas" all the Pauli strings that stabilize the state $\psi$ correspond to ground states, i.e., $\hat{H}(\psi)$ has a $|\rm STAB (\psi)|$-degenerate ground state. Excited states correspond to Pauli strings that do not stabilize $\psi$.

\paragraph{Perspective 2.}

We now present another complementary, thermodynamic perspective in which the Pauli spectrum of a state is reinterpreted as the energy spectrum of a fictitious many-body system.
The construction is as follows.
Interpret the (rescaled) state as a Hamiltonian $\hat{H} = d_{_A}\psi_{_A}$, and regard each Pauli operator $P\in \mP_n$ as labeling a particle that can occupy one of two levels, with energies
$$E_{P,\pm} = 1 \pm \tr[\psi_A P]\,.$$
The resulting collection of $d^2$ two-level particles (with 2$d^2$ levels in total), confined to the band $[0,2]$ is what we refer to as the ``Pauli gas" whose energy spectrum is associated to the Pauli spectrum of $\psi_A$, where the element $p_i$ in the Pauli spectrum yields the conjugate pair $1\pm p_i$ in the energy spectrum of the Pauli gas.
Moreover, these energies are not merely formal. For every non-identity Pauli, $E_{P,\pm}$ is the energy of a genuine physical state $(I\pm P)/d$. The identity Pauli contributes the pinned band edges $E=0$ and $E=2$, carrying no information and playing the role of vacuum mode of gas.

The utility of this construction is that it helps us use the standard thermodynamic quantities to construct
the central objects of our work, the stabilizer partition function and stabilizer work, allowing us to utilize thermodynamics' tools in characterizing and quantifying the magic of a quantum state.
Recall that the Clifford operations permute the order and/or change the sign of the elements of the Pauli spectrum without altering the amount of magic in the system.
Thus, if the Hamiltonian is modified by conjugation with Cliffords, the energies of this ``Pauli gas" remain invariant.

The dictionary between magic and spectral structure is now sharp. When $\hat{H}$ is proportional to a pure stabilizer state, the gas is maximally rigid, i.e., only the integer levels $\{0,1,2\}$ occur, with ground-level degeneracy equal to $d$.
Magic breaks this structure, populating fractional levels in the band interior and depleting the degenerate ground level.
Non-stabilizerness, a microscopic property of $\psi_A$, is thereby recast as a macroscopic, spectral property of a thermodynamically large ensemble ($d^2 = 4^n$ particles, growing exponentially with system size).
In particular, the magic measure that we introduce below (referred to as the stabilizer work) is precisely the free-energy difference of the Pauli gas, evaluated against the stabilizer reference.
Its zero-temperature limit then acquires a transparent meaning as the free-energy difference reduces to a difference of residual entropies, which equals the stabilizer nullity in the case of stabilizer work. The nullity thus emerges as a third-law statement for the Pauli gas.


\subsection{Stabilizer Partition Function}

Guided by the above analogies, we introduce the Stabilizer Partition Function ({SPF}) as the canonical partition function associated with the ``Pauli Gas".
\begin{definition}[Stabilizer Partition Function]\label{def:Pauli_transform}
The Stabilizer Partition Function of a state $\ket{\psi} \in \mathcal{H}$ is a real function $\mathcal{Z}_{\beta}(\psi): \mathbb{\mathbb{R}} \times \mathcal{H} \to \mathbb{\mathbb{R}}$ defined as
\begin{equation*}
    \mathcal{Z}_{\beta}(\psi) := \frac{1}{2}\tr\left[e^{-\beta \hat{H}(\psi)}\right] =e^{-\beta}\sum_{x \in {\rm Spec}(\psi)}\cosh( \beta x)\,.
\end{equation*}
where $\beta\in \mathbb{R}$.
\end{definition}
We would like to emphasize here that, from a probability theory perspective, the SPF as defined above is a moment-generating function.
As with any generating function, it captures the full statistical properties of the Pauli spectrum into an analytic function, rather than isolating a specific $n$-th moment.
In this sense, SPF captures the global statistical properties of the spectrum, resolving the question raised earlier.
Moreover, the definition admits another analogy with statistical mechanics, where the parameter $\beta$ plays a role analogous to that of the inverse of the physical temperature (without really having any connection to it).
For instance, for small $\beta$ (suggestive of high temperatures) all micro-states contribute and we have $\lim_{\beta \to 0}\mathcal{Z}_{\beta}(\psi)=d_{n}^{2}$, whereas for large $\beta$ (low temperatures) only low-energy states contribute to the partition function, and we have the asymptotic behavior
\begin{equation}
    \lim_{\beta \to \infty}\mathcal{Z}_{\beta}(\psi) = \frac{|\rm STAB (\psi)|}{2}.
\end{equation}
Furthermore, for every $\beta$ we can rewrite the SPF as
\begin{equation}
    \mathcal{Z}_{\beta}(\psi)=e^{-\beta}\left(|{\rm Ker}(\psi)|+\sum_{P\, \in \,\rm Supp(\psi)}\cosh{(\tr{P\psi})}\right)\,.
\end{equation}

We also introduce here a filtered variant of SPF, referred to as the core stabilizer partition function and denoted by $\mZ^c$.
This function is obtained by dropping the leading-order constant term in the SPF, which corresponds to the number of Paulis.
We define it as follows:
\begin{definition}[Core Stabilizer Partition Function]
    Let $\ket{\psi}\in \mH$. The core Stabilizer Partition Function (cSPF) is a real-valued map  $\mZ^{c}_{\beta}(\psi):\mbR\times\mH\to \mbR$ given by
    \eqs{
    \mZ^c_{\beta}(\psi):= e^{-\beta}\left(\sum_{x\in {\rm Spec(\psi)}}(\cosh(\beta x)-1)\right) = \mZ_{\beta}(\psi) - \frac{d^2}{e^{\beta}}\,.
    }\,.
\end{definition}
As we demonstrate in the subsequent section, this core SPF enables the construction of a magic measure with desirable resource-theoretic properties.

\subsection{Properties and Features of the SPF}

The SPF does not retain the full labelled Pauli spectrum.
Instead, it packages its Clifford-invariant even statistics into an analytic object, thereby capturing the stabilizer-relevant information accessible through Pauli moments.
We now outline the main features of this construction, which will be used throughout the remainder of this work.

\subsubsection{Resource Theoretic Properties}

\begin{proposition}[Invariance]\label{Prop:Clifford_invariance}
The SPF is invariant under Clifford unitaries, i.e.,
\begin{equation}
    \mathcal{Z}_\beta({\mC(\psi)})=\mathcal{Z}_{\beta}({\psi}), \quad \text{for all} \quad \mC \in \mathcal{C}l_{n}
\end{equation}    
\end{proposition}
\begin{proof}
    Clifford unitaries act transitively on Paulis. This implies that Cliffords can add a phase to the elements of the Pauli spectrum and/or can permute them. 
    For the former, the hyperbolic cosine in the SPF, takes care of the $1\to -1 $ symmetry/phase change.
    For the latter, we have $\text{Spec}(\psi)^{\downarrow}=\text{Spec}(\mC(\psi))^{\downarrow}$ for every $\mC\in \mathcal{C}l_{n}$, where the elements in the Pauli spectrum are arranged in decreasing order, and hence, summation over spectrum keeps the SPF invariant.
\end{proof}
\begin{proposition}[Uniqueness for stabilizer states]\label{Prop:Uniqueness_stabilizer_states}
For any $n$-qubit state $\ket{\phi}\in {\rm Stab}_n$, 
its SPF is uniquely given by $\mathcal{Z}_{\beta}(\phi)=\mathcal{Z}_{\beta}( {\rm STAB}_n)$, with 
\begin{equation*}
    \mathcal{Z}_{\beta}({\rm STAB}_n)=e^{-\beta}\left[(d_{n}^{2}-d_{n})+d_{n}\cosh(\beta)\right],
    \quad d_{n}= 2^n.
\end{equation*}
\end{proposition} 
\begin{proof}
This follows by the direct application of proposition~\ref{proposition:Stab_Pauli_spectrum} and proposition~\ref{Prop:Clifford_invariance} in the SPF definition.
\end{proof}
In the resource theoretic sense, this uniqueness for stabilizer states corresponds to saying the SPF is faithful under free operations (Clifford unitaries).
It is also straightforward to verify that the above two properties extend to the core SPF, with the core SPF admitting a more succinct and a simpler closed form expression for stabilizer states:
\eqs{\mZ^c_{\beta}({\rm STAB_n}) = e^{-\beta}d_n (\cosh(\beta) - 1).}

\subsubsection{Alternative representations}

Although the SPF was defined directly from the Pauli spectrum in Def.~\ref{def:Pauli_transform}, that form is not always the most transparent one. In particular, it obscures both how the SPF can be accessed operationally and how it relates to existing stabilizer measures. We now give two equivalent representations, based on the orbit-stabilizer theorem and on stabilizer R\'enyi entropies, which make these two aspects explicit.

\begin{proposition}[Clifford-orbit representation]

By the orbit-stabilizer theorem, the sum over all the $4^{n}-1$ non-identity Pauli strings in the definition of the SPF (Def. \ref{def:Pauli_transform}) can be replaced by the Clifford group action over an arbitrary non-identity Pauli $P$ (say $Z^{\otimes n}$), yielding the following expression 
\begin{align}
\label{eq:orbit_stabilizer_expansion}
    \mathcal{Z}_{\beta}(\psi)=e^{-\beta}\left[\cosh{(\beta)}
    +\frac{d_{n}^{2}-1}{|\mathcal{C}l_{n}|}\sum_{C \in \mathcal{C}l_{n}}\cosh{\left(\beta\,\tr\left[\psi P_{C}\right]\right)}\right],
\end{align}
with $P_{C}=CPC^{\dagger}$. When the previous expansion is set in terms of the Clifford $k$th-fold channels $\Phi^{k}_{\mathcal{C}l_{n}}(\cdot)=|\mathcal{C}l_{n}|^{-1}\sum_{C\in \mathcal{C}l_{n}}C^{\otimes k}(\cdot) C^{\dagger \otimes k}$, we will refer to Eq.~\eqref{eq:orbit_stabilizer_expansion} as the Clifford-orbit representation of the SPF, and is expressed as
\begin{align}
\label{eq:Clifford_representation_P_transform}
    \mathcal{Z}_{\beta}(\psi)&=e^{-\beta}\left[\cosh{(\beta)}\right.\\ \nonumber
    &\left.+(d_{n}^{2}-1)\sum_{k \geq 0}\frac{\beta^{2k}}{2k!}\tr\left[P^{\otimes 2k}\, \Phi^{2k}_{\mathcal{C}l_{n}}(\psi^{\otimes 2k})\right]\right], \,\, \forall \, P \neq \mathbb{I}.
\end{align}

\end{proposition}
The Clifford-orbit representation has two useful consequences. First, it rewrites the SPF as an average over the Clifford orbit of a single non-identity Pauli, rather than as a direct sum over the full Pauli spectrum. Second, after expanding the hyperbolic cosine, it expresses the SPF in terms of Clifford $2k$-fold channels acting on $2k$ copies of the state. Since Clifford unitaries form a $2$-design, the first genuinely state-dependent contribution appears beyond the universal second-moment sector. This representation is therefore the natural starting point for sampling-based access to the SPF.

\begin{proposition}[Stabilizer R\'enyi representation]\label{prop.Stabilizer_representation_PT}
The SPF for an $n$-qubit state $\ket{\psi}$ admits an expansion in terms of the integer-index stabilizer entropies $M_{k}(\psi)$ 
\begin{equation*}
    \mathcal{Z}_{\beta}(\psi)=e^{-\beta}\left[d_{n}^{2}+\frac{d_{n}}{2}\beta^{2}+\sum_{k\geq 2} \frac{\beta^{2k}}{(2k)!}2^{n-(k-1)M_{k}(\psi)}\right]\,.
\end{equation*}
\end{proposition}
\begin{proof}
    This follows directly by expressing the moments of the Pauli spectrum in terms of the Stabilizer R\'enyi entropies.
\end{proof}

The above representation shows that the SPF is not an unrelated object, but a generating function for the same moment hierarchy probed by the integer stabilizer R\'enyi entropies. A fixed R\'enyi entropy isolates one moment of the Pauli spectrum, while the SPF packages the whole hierarchy into a single tunable function of $\beta$. In this sense, changing $\beta$ changes which part of the moment hierarchy is emphasized.

\subsubsection{Robustness}

\begin{theorem}[Lipschitz property]\label{Theorem:Lipshitzs_property}
Given two $n$-qubit pure states $\ket{\phi},\ket{\psi}$, the difference between their SPFs is upper bounded as
\begin{equation}
    |\mathcal{Z}_{\beta}(\psi)-\mathcal{Z}_{\beta}(\phi)|\leq L_{n}(\beta)\|\psi-\phi\|_{1}, 
\end{equation}
with $L_{n}(\beta)=e^{-\beta}d_{n}\beta\left(\sinh{\beta}\right)$.
\end{theorem}
\begin{proof}
    The proof is left to the supplemental material and is based on the Clifford orbit representation given in Eq.~\eqref{eq:Clifford_representation_P_transform}.
\end{proof}
Notice that the term $\beta\sinh(\beta)$ inherits the $\mathbb{Z}_{2}$ symmetry $\beta\to -\beta$, while the global prefactor reflects the energy shift in the SPF definition. Additionally, by upper bounding 
$||\cdots||_{1}$ in terms of $||\cdots||_{2}$, an immediate corollary of this theorem is that the SPF is a Lipschitz function with Lipschitz constant $L_n(\beta)$.

We want to stress that, at the physical level, the Lipschitz property is not just a convenient upper bound but a fundamental requirement that guarantees that the SPF, and therefore our magic monotone, behaves as a robust quantity.
Concretely, given a state $\psi$, any $\epsilon$-close state $\psi_{\epsilon}$ (in trace distance) will have an SPF $\mathcal{Z}_{\beta}(\psi_{\epsilon})$ that remains close to that of $\psi$. This robustness under small perturbations will play a crucial role in the next section, where we examine the SPF across different classes of non-stabilizer states.

\subsubsection{Computability}

As for other magic witnesses based on the Pauli spectrum, such as the stabilizer R\'enyi entropies, the computational task behind the SPF is the estimation of Pauli-spectrum moments. In \cite{Tobias/Pauli_sampling/2024}, the authors introduced a Bell-sampling procedure to estimate these moments efficiently. Here we use the same primitive, with a small amount of additional classical memory, to estimate finite-order truncations of the SPF. The full SPF is then approached by controlling the corresponding truncation error.

\begin{theorem}[SPF complexity (informal)]\label{theorem:SPF_Complexity}
For fixed $k$ and $\beta$, the truncated SPF $\mathcal{Z}^{(k)}_{\beta}(\psi)$, i.e., the SPF obtained by a $kth$ order approximation to the truly SPF, can be estimated, up to additive error $\epsilon$ and failure probability $\delta$, using the same Bell-sampling primitive that estimates the Pauli-spectrum moments. Approximating the normalized full SPF $\mathcal{Z}_{\beta}(\psi)$ further requires choosing $k$ so that the truncation remainder is below the target accuracy.
\end{theorem}

The full proof of the above theorem is given in the supplemental materials, but the idea is simple. Bell sampling first produces a Pauli string $P$ distributed according to the Pauli-spectrum weights. Repeated measurements of this same $P$ on copies of $\psi$ give access to the even Pauli moments, which we store as a vector $\widehat{\mathbf m}_{k}$. The truncated SPF is then obtained only at the final classical post-processing step, by contracting $\widehat{\mathbf m}_{k}$ with the coefficient vector $\mathbf c^{(k)}(\beta)$. Thus the sampling primitive itself is independent of $\beta$, although the final accuracy and sample complexity depend on $\beta$, $\epsilon$, $\delta$, and the truncation error.

\subsubsection{Bounds}

Below we provide both state-independent and dependent bounds.
The state-independent bounds characterize the fundamental limitations of the SPF irrespective of the input state, whereas the state-dependent bounds are more useful in computational settings.

\begin{theorem}[SPF state-independent bounds]\label{theorem:SPF_state_independent_bounds}

Given $\beta \in \mathbb{R}$, the SPF for any $n$-qubit pure state $\ket{\psi}$ is bounded as
\eqs{
\begin{split}
    & e^{-\beta}\left[(d_{n}^{2}-d_{n})+d_{n}\cosh{\beta}\right] \geq  \mathcal{Z}_{\beta}(\psi) \\
    &\:\qquad\geq e^{-\beta}\left[\cosh(\beta)+ (d_n^2-1)\cosh\left(\frac{\beta}{\sqrt{d+1}}\right)\right]
\end{split}
}
\end{theorem}
\begin{proof}
The above bound can be obtained by expanding the SPF as $\mZ_{\beta}(\psi) = d^2 + d\beta^2/{2!} +\sum_{k>1}(\beta^{2k}/(2k)!)(\sum_P \tr[\psi P]^{2k})$.
The upper bound follows directly from the SPF of the $n$-qubit stabilizer states.
The lower bound is obtained by considering a state whose Pauli spectrum is uniform over all non-trivial Pauli operators. 
\end{proof}

We next bound the SPF of a given state $\psi$ in terms of its stabilizer R\'enyi entropy.

\begin{theorem}[SPF state-dependent bounds]\label{theorem:SPF_state_dependent_bounds}

Given $\beta \in \mathbb{R}$, the SPF of a given $n$-qubit pure state $\ket{\psi}$ satisfies
\begin{align*}
& e^{-\beta}\left[(d_{n}^{2}-d_{n})+d_{n}\cosh{\beta}\right] \geq  \mathcal{Z}_{\beta}(\psi) \\ \nonumber &\qquad\geq
e^{-\beta}\left[d_{n}(d_{n}-2^{M_{2}(\psi)})+ d_{n}2^{M_{2}(\psi)}\cosh{(\beta\, 2^{-\frac{M_{2}(\psi)}{2}})}\right]
\end{align*}
where $M_{2}(\psi)$ is the order-2 Stabilizer R\'enyi entropy.
\end{theorem}
\begin{proof}
The proof relies on the Stabilizer R\'enyi representation of Proposition~\ref{prop.Stabilizer_representation_PT}, and the R\'enyi hierarchy. 
Concretely, using $M_{2}(\psi) \geq M_{k>2}(\psi)\geq M_{\infty}(\psi)=0$, we obtain both the lower and the upper bounds. 
A variational refinement of the lower bound is provided in the Supplemental material.
\end{proof}
The above theorem shows that the SPF is governed by two extremal configurations of the Pauli spectrum. 
On the one hand, the lower bound is achieved by a Pauli spectrum consisting of $d_{n}^{2}-d_{n}2^{M_{2}(\psi)}$ zeros and $d_{n}2^{M_{2}(\psi)}$ nonzero entries taking possible values in $\{\pm 2^{-\frac{M_{2}(\psi)}{2}}\}$. 
On the other hand, the upper bound is obtained by a spectrum containing $(d_{n}^{2}-d_{n})$ zeros and $d_{n}$ entries equal to $\pm 1$, corresponding exactly to the Pauli spectrum of an $n$-qubit stabilizer state, thus converging to the upper bound as obtained for the state-independent case.

\section{Stabilizer Partition Function across the Hilbert space}\label{sec:SPF_across_Hilbert_space}

In the previous section, we have defined the SPF and presented its main properties. 
In this section, we turn our attention to characterizing the form of the SPF for different classes of states across the whole ``Hilbert space Zoo'' (see Fig~\ref{fig:Hilbert_zoo}). Concretely, and motivated by the recently introduced computational separability of the Hilbert space \cite{Gu/Magic_vs_Ent/2024}, we focus on both highly magical and highly entangled states.

\begin{figure}
    \centering
    \includegraphics[width=0.7\linewidth]{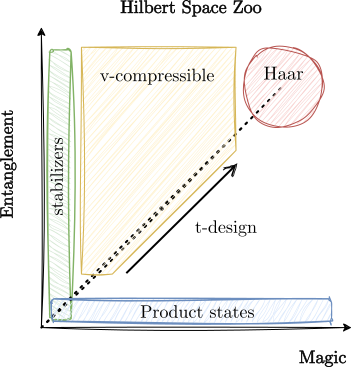}\label{fig:Hilbert_zoo}
    \caption{Stabilizer partition function across the Hilbert space zoo. For pure stabilizer states the SPF is uniquely valued and its given by Proposition~\ref{Prop:Uniqueness_stabilizer_states}.
    For separable states across every cut, their SPF is given by Proposition~\ref{Prop:T_ncopy-Pauli}, and Theorem~\ref{theorem: separable_states_P_transform}. Under generic many-body dynamics such as those present in chaotic systems or random quantum circuits, the states become Haar distributed with SPF given by Theorem~\ref{theorem: Pauli_T_Haar_random}. For the broad class of $\nu$-compressible states, the SPF is given by Theorem~\ref{theorem: v_compressible_Pauli_T}. Finally, for psedudomagic states the SPF is given by Theorem~\ref{Theorem:Pseudorandom_SPF}.}
\end{figure}

\subsection{SPF for product states}

We begin our exploration of the SPF across Hilbert space with states that carry magic but no entanglement across any cut. This is the cleanest setting in which to separate the contribution of magic from that of entanglement. Concretely, we consider product states of the form $\ket{\psi}= \otimes_{l=1}^{n}\ket{\psi_{l}}$, with $\{\ket{\psi_{l}}\}_{l=1}^{n}$ arbitrary single-qubit states~\footnote{Magic injection can also be implemented by adding local random Haar unitaries. Such states, however, are not separable across every cut.}.

\begin{proposition}[SPF of highly magical states]\label{Prop:T_ncopy-Pauli}

Given $n$ copies of the magical $T$ state $\ket{T}=T\ket{+}$, the SPF is 
\begin{equation*}
\mathcal{Z}_{\beta}(\ket{T}^{\otimes n})=e^{-\beta}\left[(4^{n}-3^{n})+\sum_{k=0}^{n}\binom{n}{k}2^{k}\cosh{(\beta\,2^{-k/2}) }\right].
\end{equation*}
\begin{proof}
The Pauli spectrum of $n$ copies of the $T$ state consists of all products $\prod_{l=1}^{n}\braket{T|\sigma_{l}|T}$, with $\sigma_{l} \in P_{1}$ a single-qubit Pauli. Therefore, we have either one of the following possibilities:

\vspace{2mm}

i) For all the $4^{n}-3^{n}$ Pauli strings containing at least one single Pauli $Z$. The contribution to the spectrum is zero.

\vspace{2mm}

ii) As $\braket{T|X|T}=\braket{T|Y|T}=\frac{1}{\sqrt{2}}$, all weight-$k$~\footnote{Recall that the weight of a Pauli string is the number of non-identity Paulis in the string.} Pauli strings contribute equally with $2^{-k/2}$ to the spectrum. There are $2^{k}\binom{n}{k}$ such strings.
    
\end{proof}
\end{proposition}

The SPF of the $\ket{T}^{\otimes n}$ illustrates one of the key features of the SPF, which is that even in the absence of entanglement, it is capable in distinguishing magic. 
At $\beta=0$ the SPF is insensitive to magic and only depends on the system size, {e.g.} $\mathcal{Z}_{0}(\ket{T}^{\otimes n})=4^{n}$, whereas for $\beta \neq 0$ the SPF already distinguishes magic states from the stabilizer states, i.e., $\mathcal{Z}_{\beta}(\ket{T}^{\otimes n})<\mathcal{Z}_{\beta}(\mathrm{STAB}_n)$ for every $\beta \neq 0$.

\begin{figure}
\includegraphics[width=\linewidth]{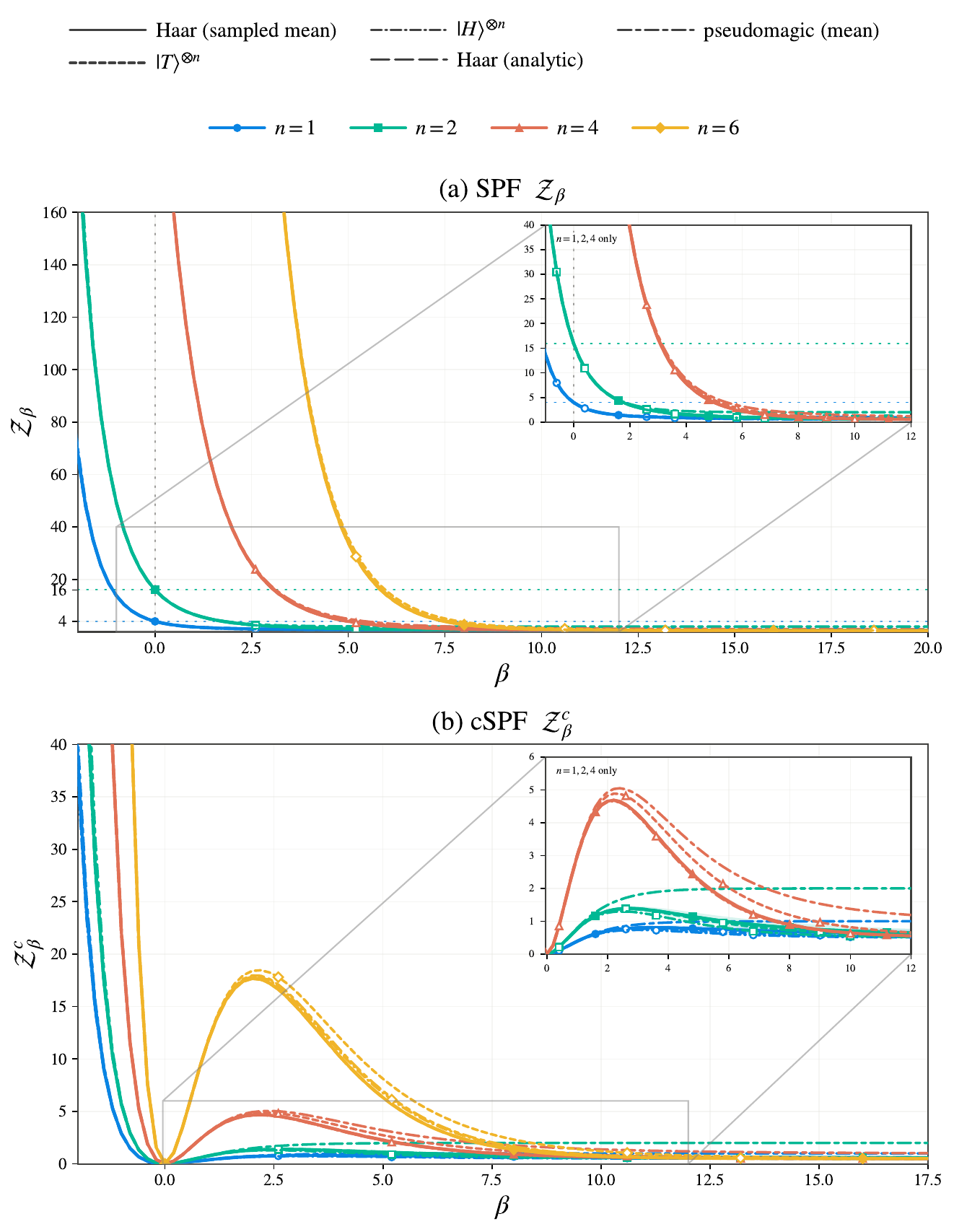}
\caption{Plots of (a) Stabilizer Partition Function (SPF) and (b) core SPF (cSPF), for different quantum states with varying number of qubits, with respect to the real parameter $\beta$. At $\beta = 0$, the SPF's value for any quantum state is equal to $4^n$. One may notice that as the number of qubits increase, the value of SPF/cSPF gets smaller than those for the lesser qubits for larger $\beta$. This is not a numerical artifact. It might be counterintuitive to see such a crossover, but since these are partition functions, it carries no meaning. When we normalize the cSPF and find the Stabilizer work, it becomes aligned with the intuition as shown later in~\figref{fig:SW}.}
\label{fig:SPF_cSPF}
\end{figure}

The previous result can be further generalized to $n$ copies of arbitrary single qubit states. 

\begin{theorem}[$n$-copy SPF]
Given a single-qubit pure state $\ket{\psi}$, 
\begin{equation*}
    \mathcal{Z}_{\beta}(\ket{\psi}^{\otimes n})=e^{-\beta}\sum_{w=0}^{n}\binom{n}{w}\sum_{n_{x},n_{y},n_{z}}\binom{w}{n_{x} \, n_{y} \, n_{z}}\cosh(\beta\,V_{\psi}[\vec{n}])
\end{equation*}
with $V_{\psi}[\vec{n}]=\prod_{\sigma \in \{X,Y,Z\}}\braket{\sigma}^{n_{\sigma}}$.

\vspace{2mm}

\begin{proof}
Similarly to Proposition~\ref{Prop:T_ncopy-Pauli}, each weight-$k$ Pauli string containing $n_{x}$ Pauli-$X$, $n_{y}$ Pauli-$Y$, and $n_{z}$ Pauli-$Z$ contributes $\braket{X}^{n_{x}}\braket{Y}^{n_{y}}\braket{Z}^{n_{z}}$ to the Pauli spectrum. Counting all such configurations gives the stated expression.
\end{proof}
    
\end{theorem}

\subsection{SPF for separable bipartite states}

We now extend the analysis from tensor powers of a single local state to arbitrary bipartite product states.
The key structural feature that emerges in this setting is that while the Pauli spectrum factorizes under tensor product of states, i.e., $\text{Spec}(\ket{\psi_{A}}\otimes \ket{\phi_{B}})=\text{Spec}(\ket{\psi_{A}})\times \text{Spec}(\ket{\phi_{B}})$, the corresponding SPF is not multiplicative:
\[
\mathcal{Z}_{\beta}(\psi_{A}\otimes \phi_{B})\neq \mathcal{Z}_{\beta}(\psi_{A})\,\mathcal{Z}_{\beta}(\phi_{B}).
\]
Nevertheless, this failure of multiplicativity is highly structured.
The following theorem shows that the SPF of the joint state admits a decomposition in terms of the subsystem SPFs together with a residual term arising from the non-stabilizer sectors of the two sub-systems.

\begin{theorem}[SPF for separable states]\label{theorem: separable_states_P_transform}

Let $\ket{\psi_{AB}}=\ket{\psi_A}\otimes\ket{\psi_B}$ be an $n$-qubit separable state across an $A|B$ bipartition.
Then its SPF admits the decomposition
%
\begin{align} \mathcal{Z}_{\beta}(\psi_{AB}) &=  2^{|B|-\nu_{B}}\mathcal{Z}_{\beta}(\psi_{A})+2^{|A|-\nu_{A}}\mathcal{Z}_{\beta}(\psi_{B}) \nonumber \\
&\qquad-e^{-\beta}2^{n-\nu_{A}-\nu_{B}}\cosh{(\beta)}
+\mathcal{Z}^{(R)}_{\beta}(\psi_{AB}).
\end{align}
with $\nu_{A(B)}$ denoting the $\nu$-nullity of sub-system $A(B)$, and 

\begin{equation}
\mathcal{Z}^{(R)}_{\beta}(\psi_{AB})=e^{-\beta}\sum_{x \in \mathrm{Spec}_{R}(\psi_{A})\, \times \, \mathrm{Spec}_{R}(\psi_{B})} \cosh(\beta\,x),
\end{equation}
is the residual non-stabilizer contribution to the SPF of the whole state. This term is obtained from the Cartesian product of the residual Pauli spectra $\mathrm{Spec}_{R}(\psi_{A})\, \times \, \mathrm{Spec}_{R}(\psi_{B})$. 
\end{theorem}

A particularly relevant special case arises when one subsystem is a stabilizer state. 
From Theorem~\ref{theorem: separable_states_P_transform}, one can get the SPF of such states by setting the stabilizer nullity to be zero for that subsystem. However, for such systems, the SPF expression admits a simpler closed form, as stated below.

\begin{lemma}[SPF for Stabilizer + Pure]\label{lemma:Stab+pure}
Let $\ket{\psi_{A}}\otimes \ket{{\rm Stab}_{B}}$ be a separable state with system B being a stabilizer state and $\ket{\psi_{A}}$ an arbitrary (non-necessarily stabilizer) state. Then
\begin{align*}
 \mathcal{Z}_{\beta}( \psi_{A}\otimes {\rm Stab}_{B})
= d_{B}\mathcal{Z}_{\beta}(\psi_{A})+e^{-\beta}(d_B^2-d_B)d_A^2
\end{align*}
with $\mathcal{Z}_{\beta}(\psi_{A})$ denoting the SPF of the reduced state on subsystem $A$.

\begin{proof}
    The Pauli spectrum of $\ket{{\rm Stab}_{B}}$ consists of $d_{B}$ ones and $d_{B}^{2}-d_{B}$ zeros. Using multiplicativity of the Pauli spectrum under tensor products, the above expression follows directly from the definition of the SPF.
\end{proof}
    
\end{lemma}
Furthermore, for the core SPF, the Pure+Stabilizer case simplifies considerably, i.e., \eqs{\mZ^c_{\beta}(\psi_A\otimes {\rm Stab}_B) = d_B \mZ^c_{\beta}(\psi_A)\,.}

We next establish a supermultiplicativity-type property of the core SPF, which will play a useful role in our discussion of stabilizer work in the subsequent section.
\begin{lemma}[cSPF supermultiplicativity]\label{lem:cSPF_supermultiplicativity}
    Let $\ket{\psi_{AB}} = \ket{\psi_A} \otimes \ket{\phi_B}$ be any bipartite $n$-qubit separable state across an $A|B$ bipartition. Then
    \eqs{
    \mZ_{\beta}^c\left(\psi_A\otimes \phi_B\right)\geq \frac{\mZ_{\beta}^c\left(\psi_A\right)\mZ_{\beta}^c\left( \phi_B\right)}{e^{-\beta}\left(\cosh(\beta)-1\right)}
    }
\end{lemma}
\begin{proof}
If we let $w_k := \frac{\beta^{2k}}{(2k)!}, a_k := \sum_i(\tr[\psi_A P^A_i]^{2k}),$ and $b_k:= \sum_j(\tr[\phi_B P^B_j]^{2k})$, then to prove the above theorem, it suffices to show that 
$$
\frac{\sum_{k\geq 1}w_k a_kb_k}{(\sum_{k\geq 1}w_k a_k)(\sum_{k\geq 1}w_k b_k)}\geq \frac{1}{\cosh(\beta)-1}\,.
$$
Now let $p_k = \frac{w_k a_k}{\sum_{k\geq 1} w_k a_k}$ be a probability distribution.
Then the numerator in the LHS in the above equation becomes $\sum_{k\geq 1}p_k b_k$.
Let us now consider a probability distribution $p^{\star}$ whose $k$-th term $p_k^{\star}= \frac{w_k}{\sum_{k\geq1}w_k}$. It is easy to identify that $p^{\star}$ can be derived from a stabilizer state in system $A$.
Now, it holds that $\frac{p_k}{P^{\star}_k}\geq \frac{p_{k+1}}{p^{\star}_{k+1}}$ and using Lemma~\ref{lem:Chebyshev_lemma} (in Supplementary Information), we get
\eqs{
\sum_{k\geq 1}p_kb_k\geq \sum_{k\geq 1}p_k^{\star}b_k = \frac{\sum_{k\geq 1}w_kb_k}{\sum_{k\geq 1}w_k}
}
which, by rearranging the terms, completes the proof since $\sum_{k\geq 1}w_k = \cosh(\beta)-1$.
\end{proof}

\subsection{SPF for Haar-random states}

We now turn to typical states displaying both maximal magic and maximal entanglement, that is, to Haar-random states. 
Indeed, for any bipartition $A|B$, Haar-random states exhibit volume-law entanglement, with entanglement entropy scaling as $S_{E}=O(n_{A})$ while simultaneously keeping the number of stabilizer generators equal to zero almost surely.  

\begin{theorem}[SPF for Haar-random states]\label{theorem: Pauli_T_Haar_random}

Let $\ket{\psi}=U\ket{0}^{\otimes n}$ be an $n$-qubit state, with $U$ a Haar-random unitary sampled from the $U(2^{n})$ Haar ensemble. Then, for $\beta\geq0$, its average SPF $\mathcal{Z}_{\beta}({\rm Haar}_{n})=\mathbb{E}_{\ket{\psi}\sim \rm{Haar}_{n}}[\mathcal{Z}_{\beta}(\psi)]$ is given by

\begin{align}
    \mathcal{Z}_{\beta}({\rm{Haar}_{n}}) &= e^{-\beta}\Bigg[\cosh{(\beta)}\nonumber \\ 
    &\quad+(d_{n}^{2}-1)\Gamma\left(\frac{d_{n}+1}{2}\right)\left(\frac{2}{\beta}\right)^{\frac{d_{n}-1}{2}}I_{\frac{d_{n}-1}{2}}(\beta)\Bigg]
\end{align}
where $I_{\nu}(x)$ denotes the modified Bessel function of the first kind, and $\Gamma(x)$ the Gamma function.
\end{theorem}
\begin{proof}
    The proof is left to the supplemental material and is based on the construction of the SPF from the Pauli-spectrum of Haar random states in \cite{Turkeshi/Pauli_spectrum/2025}. 
\end{proof}

As a sanity check of the previous theorem, notice that the normalization condition for the SPF implies $\lim_{\beta \to 0}\mathcal{Z}_{\beta}(\psi)=d_{n}^{2}$ for every state $\ket{\psi}$. This normalization is easily recovered from the result in Theorem~\ref{theorem: Pauli_T_Haar_random} by using the small-argument expansion of the Bessel functions~\footnote{We recall the reader that $\lim_{x\to 0}I_{\alpha}(x)=\frac{1}{\Gamma(\alpha+1)}\left(\frac{x}{2}\right)^{\alpha}$.}.

Now, and as previously pointed out, all but a measure-zero set of Haar states have no stabilizers. This implies that, for Haar-random states, the only $\cosh{(\beta)}$ contribution to the SPF should come from the identity Pauli, while the remaining contribution from the $4^{n}-1$ non-identity Paulis should be approximately flat. This observation becomes evident in the asymptotic limit as described by the following lemma.

\begin{lemma}[Asymptotic SPF for Haar states]\label{lemma:Haar_PT_asymptotics}
For fixed $\beta$ and large $d_{n}$, more generally whenever $\beta^{2}=o(d_{n})$, the SPF for Haar random states satisfies 
\begin{equation}   
\mathcal{Z}_{\beta}({\rm{Haar}_{n}})= e^{-\beta}\left[\cosh{(\beta)}+ (d_{n}^{2}-1)(1+o(1))\right].
\end{equation}
\end{lemma}
\begin{proof}
By substituting the large-order expansion of the modified Bessel function~\footnote{We recall $\Gamma(\nu+1)\left(\frac{2}{\beta}\right)^\nu I_\nu(\beta)= 1+\frac{\beta^2}{4(\nu+1)} +O\left(\frac{\beta^4}{\nu^2}\right).$} in Theorem~\ref{theorem: Pauli_T_Haar_random} gives the claim. 
\end{proof} 
A simple argument to understand why the complete result in Theorem~\ref{theorem: Pauli_T_Haar_random} drastically simplifies in the limit of large Hilbert spaces is the following: For $\ket{\psi} \sim \text{Haar}$, all (non-identity) Pauli measurements typically vanish, e.g. $\tr[P\psi]_{\psi\sim \text{Haar}}=0$ for $P \in \mathcal{P}_{n} \backslash \mathbb{I}$.  As the size of the Hilbert space increases, this implies that the Pauli spectrum of Haar random states narrows by approaching to a string of $(d_{n}^{2}-1)$ zeros and a single one, yielding the particular form presented in Lemma~\ref{lemma:Haar_PT_asymptotics}. In terms of the Support/Kernel we have that $|{\rm Ker}_{\mathcal{P}}({\rm Haar_{n})}|=d_{n}^{2}-1$, whereas $|{\rm Stab(Haar}_{n})|=|{\rm Supp} {\rm(Haar}_{n})|=1$ almost surely. 

Motivated by the remarkably simple form of the average SPF for Haar-random states, we are interested in seeking the robustness of this result. Concretely, given a Haar-random state $\ket{\psi}$ we would like to know how far its SPF deviates from the typical value presented in Lemma~\ref{lemma:Haar_PT_asymptotics}. The upshot is that those fluctuations are highly suppressed, as captured by the following concentration statement.

\begin{lemma}[Levy's Lemma for the SPF]\label{lemma:Haar_concentration}
Consider an $n$-qubit Haar state $\ket{\psi}$. The probability that its (properly normalized) SPF $\mathcal{Z}_{\beta}(\psi)$ deviates from the typical value $\mathcal{Z}_{\beta}({\rm Haar}_{n})$ in Theorem~\ref{theorem: Pauli_T_Haar_random}
by a factor $\epsilon \geq 0$ is doubly exponential suppressed in the number of qubits, namely

\begin{align*}
    {\rm Prob}_{\ket{\psi}\sim \rm{Haar}_{n}}\left(\,\frac{|\mathcal{Z}_{\beta}(\psi)-\mathcal{Z}_{\beta}({{\rm Haar}_{n}}) |}{L_{n}(\beta)}\geq \epsilon\, \right)& \\ \nonumber
    \leq 4\exp{\left( -\frac{2}{9\pi^{3}}2^{n}\epsilon^{2}\right)}.
\end{align*}
where $L_{n}(\beta)$ is the Lipschitz constant from Theorem~\ref{Theorem:Lipshitzs_property}.
\end{lemma}

In simple terms, the previous lemma shows that the SPF of Haar-random states concentrates sharply around its ensemble average.
More precisely, deviations from the ensemble average are doubly exponentially suppresed in the number of qubits.
Consequently, almost all Haar-random states share essentially the same SPF, and therefore, the same stabilizer work (to be introduced later).
This concentration phenomenon is illustrated in~\figref{fig:SPF_cSPF},~\figref{fig:SW}.
As the number of qubits increase, both the SPF (and its filtered version cSPF) as well as the corresponding stabilizer work become confined to an increasingly narrow band around their Haar-average values (for reference see~\figref{fig:SW} where as $n$ increases, the magic of random states converge towards its mean).

A direct operational application of this lemma lies in the task of Haar-state certification. Specifically, given an unknown state $\psi$ with the promise that is either sampled from the Haar ensemble or generated by some alternative process, the goal is to determine which of the two cases holds up to a certain tolerance. Here we provide a simple algorithm to perform this task using  a single oracle access to the SPF of the state.
\vspace{8pt}
\begin{algorithm}
    \caption{Haar state certification via the SPF}
    \label{algorithm:Haar_discrimination}
    \begin{algorithmic}[1]
        \Procedure{HaarCertify}{}
        \Require SPF $\mathcal{Z}_{\beta}(\psi)$ of the unknown $n$-qubit state $\psi$ for small $\beta$, equivalently $\beta=o(\sqrt{d_{n}})$
        
        \Ensure Void()

        \State $\delta_{\beta}(\psi)\gets \mathcal{Z}_{\beta}(\psi)-e^{-\beta}(\cosh{(\beta)}+(d_{n}^{2}-1))$

        \State $\delta_{\beta}(\psi)\gets \frac{e^{\beta}}{(d_{n}^{2}-1)L_{n}(\beta)}\delta_{\beta}(\psi)$ \Comment{Normalize}

        \vspace{2pt}

        \If{$\delta_{\beta}(\psi) = o(1)$ for $\beta=o(\sqrt{d_{n}})$}
        \vspace{1pt}
            \State Print( $\psi$ is Haar random )
        \Else
            \State Print( $\psi$ isn't Haar random )
        \EndIf
        \EndProcedure
    \end{algorithmic}
\end{algorithm}
\vspace{8pt}
The effectiveness of the SPF as a Haar-randomness witness can be traced back to its stabilizer-entropy representation (Proposition~\ref{prop.Stabilizer_representation_PT}).
Individual $\alpha$-stabilizer R\'enyi entropies are generally insufficient to certify Haar randomness, since many state ensembles, such as sufficiently high-order unitary designs, reproduce their Haar values. In particular, distinguishing a Haar-random state from a $k$-design using a fixed $\alpha$-stabilizer R\'enyi entropy becomes impossible whenever the corresponding moments agree, typically requiring $k < 2\alpha$ for discrimination.
By being a fine-tuned series of the $\alpha$-stabilizer entropies, the SPF becomes sensitive to true Haar randomness. This observation will become relevant in a subsequent subsection where we analyze the SPF for states that spoof the magic of Haar random ones.

\subsection{SPF for $\nu$-compressible states}

Once the SPF has been characterized for both highly magical but unentangled states and highly entangled highly magical states, the next step is to ask for states lying somewhere in between, ${\rm i.e.}$ states for which the amount of magic is fixed while the entanglement can vary freely. These are precisely the set of $\nu$-compressible states.

\begin{definition}[$\nu$-compressible states \cite{Gu/Magic_vs_Ent/2024}]

Let $\ket{\psi}$ be an $n$-qubit pure quantum state. We say $\ket{\psi}$ is $\nu$-compressible if it can
be prepared by the application of a unitary $U$ acting on $\nu$-qubits followed by a Clifford circuit $C$ acting over $n$-qubits. 
\end{definition}

The power of looking at $\nu$-compressible states rely in that they provide a broad parametrization of non-stabilizer states with the following properties: First, any state $\ket{\psi}$
is $\nu$-compressible if and only if it has stabilizer nullity $\nu$. The equivalence between stabilizer nullity $\nu$ states and $\nu$-compressible states was established in \cite{Gu/Magic_vs_Ent/2024}. Second, $\nu$-nullity states $\ket{\psi}$ with $\nu=o(n)$ belong to the so-called Entanglement Dominated (ED) phase of the Hilbert space, ${\rm i.e.}$ opposite to Magic Dominated (MD) states, ED states have entanglement that significantly surpasses their magic, allowing the construction of sample- and time-efficient algorithms for entanglement tasks over them \cite{Gu/Magic_vs_Ent/2024}.

To characterize the SPF for $\nu$-compressible states, we sample the states according to the uniform measure $\mu_{\nu}$ on the set of $\nu$-compressible states constructed via the convolution of the $\nu$-qubit Haar measure and the uniform measure over $n$-qubit Cliffords, ${\rm i.e.}$ $\mu_{n,\nu}=\mu_{\mathcal{C}l_{n}}\ast \mu_{\text{Haar}_{\nu}}$.

\begin{theorem}[SPF for $\nu$-compressible states]\label{theorem: v_compressible_Pauli_T}
Let $\ket{\psi}$ be a $\nu$-compressible state sampled uniformly from $\mu_{n,\nu}$, i.e. $\ket{\psi}=U \ket{0}^{\otimes n}$ with $U \sim \mu_{n,\nu}$. Its average SPF $\mathcal{Z}_{\beta}({\rm Comp}_{\nu})=\mathbb{E}_{\ket{\psi}\sim \mu_{\nu}}[\mathcal{Z}_{\beta}(\psi)]$ is given by 
\begin{align*}
   \mathcal{Z}_{\beta}({\rm Comp}_{\nu})&= 2^{n-\nu}\mZ_{\beta}({\rm Haar_{\nu}})+e^{-\beta}4^{\nu}\left(4^{n-\nu}-2^{n-\nu} \right)
\end{align*}
where $\mathcal{Z}_{\beta}(\rm{Haar}_{\nu})$ is the SPF of the reduced $\nu$-qubit Haar state.    
\end{theorem}

\begin{proof}
    Due to Clifford invariance (Proposition~\ref{Prop:Clifford_invariance}), the SPF for every $\nu$-compressible state $\ket{\psi}$ can be brought to the product form $\mathcal{Z}_{\beta}(\psi)=\mathcal{Z}_{\beta}\left(U\ket{0}^{\otimes \nu}\otimes \ket{0}^{\otimes (n-\nu)}\right)$, with $U$ a $\nu$-qubit non-Clifford unitary. Therefore, by the application of Lemma~\ref{lemma:Stab+pure}, and Theorem~\ref{theorem: Pauli_T_Haar_random}, one arrives at the above expression.
\end{proof}

As a consistency check, observe that Theorem~\ref{theorem: Pauli_T_Haar_random} reduces to the Haar case in the limit $\nu \to n$, and to the Stabilizer case (Proposition~\ref{Prop:Uniqueness_stabilizer_states}) when $\nu \to 0$. 
For the cSPF, the expression simplifies to $$\mZ^c_{\beta}({\rm Comp}_{\nu}) = 2^{n-\nu}\mZ^c_{\beta}({\rm Haar}_{\nu})\,.$$
Analogous to Haar-random states, $\nu$-compressible states also concentrate around their typical value, namely 
\begin{lemma}[Concentration for $\nu$-compressible states]\label{lemma:compressible_concentration}
Given a $\nu$-compressible state $\psi$, its SPF concentrates around its mean $\mathcal{Z}_{\beta}({\rm Comp}_{\nu})$ given by Theorem~\ref{theorem: v_compressible_Pauli_T} with a Lipschitz constant $L_{\rm Comp_{\nu}}(\beta)=2^{n-\nu}L_{n}(\beta)$.

\begin{proof}
    By Lemma~\ref{lemma:Stab+pure}, the difference between the SPFs of two $\nu$-compressible states $\ket{\psi},\ket{\phi}$ is given by $\lvert\mathcal{Z}_{\beta}(\psi)-\mathcal{Z}_{\beta}(\phi)\rvert=2^{n-\nu}\lvert\mathcal{Z}_{\beta}(U_{\psi}\ket{0}^{\otimes \nu})-\mathcal{Z}_{\beta}(U_{\phi}\ket{0}^{\otimes \nu})\rvert$ with $U_{\psi},U_{\phi}$ $\nu$-qubit non-Clifford unitaries. Therefore, by Lemma~\ref{lemma:Haar_concentration}, the Lipschitz constant for $\nu$-compressible states is $L_{\rm Comp_{\nu}}(\beta)=2^{n-\nu}L_{n}(\beta)$.
\end{proof}
    
\end{lemma}

Before continuing further, we desire to highlight that the previous concentration statement for both, Haar and $\nu$-compressible states play a fundamental role. Namely, Lemmas~\ref{lemma:Haar_concentration} and \ref{lemma:compressible_concentration} guarantee the use of the SPF to characterize these families of states. Without them, we wouldn't be able to guarantee that SPF will remain (exponentially close) to its average between different samples \footnote{This is the same as saying that the SPF is self-averaging.}. 

\begin{lemma}[Asymptotic SPF for $\nu$-compressible states]

For fixed $\beta$ and sufficiently large $2^{\nu}$, specifically when $\beta^{2}=o(2^{\nu})$, the SPF for $\nu$-compressible states satisfies
\begin{align*}
    \mathcal{Z}_{\beta}({\rm Comp}_\nu)&=
    e^{-\beta}\left[ 4^n -2^{n-\nu} +2^{n-\nu}\cosh\beta \right] \\
    &\quad + e^{-\beta}2^{n-\nu}(4^\nu-1)o(1).
\end{align*}
\begin{proof}
The result follows directly from Lemma~\ref{lemma:Haar_PT_asymptotics}.
\end{proof} 
\end{lemma}
Recall that $|{\rm STAB (Comp_{\nu})}|=2^{n-\nu}$ and $|{\rm Ker}_{\mathcal{P}}{\rm  (Comp_{\nu})}|=4^{n}-2^{n-\nu}$.
Thus, the asymptotic expression above implies that the SPF of a $\nu$-compressible state depends only on the size of its stabilizer sector up to vanishing corrections.
Analogous to the Haar-certification problem discussed previously, the SPF allows us to perform the task of $\nu$-nullity certification. 
\begin{task}
Given an unknown $n$-qubit state $\psi$, we desire to certify whether it is $\nu$-compressible or not.
\end{task}
The following algorithm solves the problem given oracle access to the SPF.
\vspace{8pt}
\begin{algorithm}
    \caption{$\nu$-Compressible state Verification}
    \label{alg:nu_compressible_verification}
    \begin{algorithmic}[1]
     \Procedure{$\nu$Certify}{}
        \Require Integer $\nu \leq n$ (expected $\nu$-nullity of the state), oracle access to the state SPS $\mathcal{Z}_{\beta}(\psi)$ for $\beta=o(\sqrt{2^{\nu}})$
        \Ensure Void()
        
        \State $\delta_{\beta}(\psi) \gets \mathcal{Z}_{\beta}(\psi)-e^{-\beta}[4^{n}-2^{n-\nu}+2^{n-\nu}\cosh{(\beta)}]$

        \State $\delta_{\beta}(\psi) \gets \frac{e^{\beta}}{2^{n-\nu}(4^{\nu}-1)L_{\rm Comp_{\nu}}(\beta)} \delta_{\beta}(\psi)$ \Comment{Normalize}

        \If{$\delta_{\beta}(\psi) = o(1)$ for $\beta=o(\sqrt{2^{\nu}})$}
            \State Print( $\psi$ is $\nu$-compressible)
        \Else
            \State Print( $\psi$ isn't $\nu$-compressible)
        \EndIf
    \EndProcedure
    \end{algorithmic}
\end{algorithm}
\vspace{8pt}
Compared to algorithm introduced in \cite{Gu/Magic_vs_Ent/2024}, the previous algorithm does it at the cost of having oracle access. 

An immediate consequence of Theorem~\ref{theorem: v_compressible_Pauli_T}, together with Lemma~\ref{lemma:compressible_concentration}, is a characterization of the SPF for $t$-doped states. These are a subset of
$\nu$-compressible states obtained by starting from an $n$-qubit stabilizer state and applying at most $t$ non-Clifford gates, each of size
$l = O(1)$. Such states are automatically $\nu$ compressible with $\nu \le 2lt$. The relevance of $t$-doped states is twofold. First, they constitute
$\varepsilon$-approximate $k$-designs for $t = O\!\left(k^{4}\log k^{2}\,\log \frac{1}{\varepsilon}\right)$ as shown in \cite{Haferkamp/arbitrary_k_designs/2022}.
Second, they can be efficiently learned using Clifford decoders whenever $ t = O(\log n)$ as demonstrated in \cite{Leone/Clifford_decoders_PRA/2024,Oliviero/Clifford_Decoders_letter/2024}. This makes $t$-doped states an analytically tractable yet expressive family
that interpolates between stabilizer states and generic (highly magical) states.

\begin{corollary}[SPF of $t$-doped states]
Let $\ket{\psi}$ be a $t$-doped state. Then its SPF is given by Theorem~\ref{theorem: v_compressible_Pauli_T} for some $\nu \leq 2lt$.
\end{corollary}
\begin{proof}
    As $t$-doped $\subseteq$ $\nu$-compressible, the claim follows from Theorem~\ref{theorem: v_compressible_Pauli_T}.
\end{proof}

Analogously to the cases of Haar-random and $\nu$-compressible states, the SPF for $t$-doped states concentrates around its ensemble-average value. This concentration property naturally gives rise to the following inference problem. 
\begin{task}
Suppose we are given an unknown $n$-qubit state with the promise that is $t$-doped with known $t$. The task is to find $t$.
\end{task}
The following algorithm shows that, given oracle access to the SPF, one can incur a confidence interval such that $t \in [t_{\rm min}, t_{\rm max}]$.
\vspace{8pt}
\begin{algorithm}
\caption{$t$-extraction}\label{algorithm:t_doped_discrimination}
\begin{algorithmic}
\Procedure{tExtract}{}
\Require SPF $\mathcal{Z}_{\beta}(\psi)$ of the unknown $n$-qubit state $\ket{\psi}$
\Ensure Lower bound $t_{\rm min} \leq t$

\For{$\tilde{t} \gets 1,\dots, t_{max}$} \Comment{loop over $t$}

\For{$\tilde{\nu} \gets 1,\dots, 2l\tilde{t}$} 
\Comment{loop over possible $\nu$}

\State Bool $\gets$ $\nu$CERTIFY( $\tilde{\nu}$)
\If{$\text{Bool} =\text{YES} $}
\State \Return $t_{\rm min} \gets \tilde{t}$

\EndIf
\EndFor
\EndFor

\EndProcedure
\end{algorithmic}
\end{algorithm}
\vspace{8pt}
Finally, we emphasize that the oracle-access assumption is primarily a conceptual simplification. By Theorem~\ref{theorem:SPF_Complexity}, the SPF itself can be estimated efficiently in polynomial time. 
Consequently, the certification and inference procedures presented above are efficiently implementable and do not require direct oracle access in practice.

\subsection{SPF for Pseudomagic States}

We now turn our attention to pseudo-magic states, ${\rm i.e.}$ states displaying low magic, while at the same time being efficiently preparable and computationally indistinguishable from high magical states. 
Computational indistinguishability here means that a polynomial bounded algorithm cannot distinguish a pseudo-magic state from a high magical state.
The central idea is the following: given two ensembles $\mathcal{E}_{\text{low}}$ and $\mathcal{E}_{\text{high}}$ of $n$-qubit states with the property that they are indistinguishable (given polynomial number of copies), i.e.,
\begin{align*}
    \left \| \mathbb{E}_{\psi\sim\mathcal{E}_{\text{low}}}\left(\ket{\psi}\bra{\psi}^{\otimes t}\right)- \mathbb{E}_{\psi \sim \mathcal{E}_{\text{high}}}\left(\ket{\psi}\bra{\psi}^{\otimes t}\right)  \right \|_{1} \leq  \frac{1}{\text{Poly}(n)},
\end{align*}
one says that $\mathcal{E}_{\text{low}}$ and $\mathcal{E}_{\text{high}}$ form a pseudomagic pair of gap $f(n)$ vs $g(n)$ according to some magic measure $\mathcal{M}$ if with overwhelming probability $\mathcal{M}(\psi \sim \mathcal{E}_{\text{high}}) = f(n)$ and $\mathcal{M}(\psi \sim \mathcal{E}_{\text{low}}) = g(n)$ such that $f(n) > g(n)$. Intuitively speaking, states sampled from $\mathcal{E}_{\text{low}}$ will mimic the magic properties of $\mathcal{E}_{\text{high}}$ (according to the measure $\mathcal{M}$) even though having much less magic~\cite{Gu/pseudomagic/2024}. 
The notion of pseudo-magic is similar to the notion of pseudo-entanglement where one chooses an entanglement monotone instead \cite{Aaronson/pseudoentanglement/2022}. 
In \cite{Aaronson/pseudoentanglement/2022,Gu/pseudomagic/2024}, it was shown that Subset phase states (SPS) defined via 
\begin{equation}
\label{eq:SPS_def}
    \ket{\psi_{f,S}}=\frac{1}{\sqrt{|S|}}\sum_{x \in S}(-1)^{f(x)}\ket{x},
\end{equation}
with $f:\{0,1\}^{n}\to \{0,1\}$ a pseudorandom function, and $S \subseteq \{0,1\}^{n}$ a pseudorandom subset of fixed length  $|S|=2^{k}$ ($\log{n}\leq k\leq n$) form a pseudomagic pair of gap $\omega(k)$ vs $n$ with Haar random states (maximal magical) for $\mathcal{M}$ the stabilizer R\'enyi entropies.

In this section we show that the SPF is also able to detect pseudo-magic. 
We do so by providing an exact closed-form expression for the Stabilizer partition function of the pseudo-magic states in \eqref{eq:SPS_def}. By simplicity we are going to focus our analysis on the $|S|=2^{n}$ case, where there is only a single subset $S=\{0,1\}^{n}$, and the only average one has to care about is the one over the pseudo-random functions.

\begin{theorem}[SPF for pseudomagic states]\label{Theorem:Pseudorandom_SPF}

Given an $n$-qubit SPS, sampled uniformly from SPS with  $|S|=2^{n}$, its average SPF is given by  
\begin{align*}
    \mathcal{Z}_{\beta}({\rm SPS})&=e^{-\beta}
\left[\frac{(d_n-1)(d_n+2)}{2}+\cosh{\beta} \right. \\
& + \left. \frac{d_n(d_n-1)}{2}
\left(\cosh\left(\frac{2\beta}{d_n}\right)\right)^{d_n/2}
\right].
\end{align*}
\end{theorem}

\begin{proof}
The proof of this result is left to the supplemental material and is based on the exact evaluation of the Pauli spectrum moments for SPS.
\end{proof}
The family of pseudo-magic states in \eqref{eq:SPS_def} also display pseudo-entanglement. However one can easily tune the degree of entanglement by applying Cliffords, namely $\ket{\rm SPS} \to C\ket{\rm SPS}, \;\; {\rm for} \; C\in Cl_{n}$, while keeping the magic content of the SPS unchanged.

Notice that this result makes explicit that the stabilizer nullity isn't a good measure to quantify pseudo-magic as pointed out in \cite{Gu/pseudomagic/2024}. Indeed by looking only at nullity, SPS will be as magical as Haar random states since both have $\nu=\mathcal{O}(n)$. 
However, the algorithm of $\nu$-extraction still is able to efficiently detect a pseudo-magic state, but this efficieny doesn't translate into efficiently distinguishing these states from the Haar random (the $\nu \to n$ case) case, else it will contradict how the pseudomagic states are defined. We see that the SPF curves get exponentially close and so distinguishing them becomes challenging.

\section{Stabilizer Work
}\label{sec:SW}

In the previous section, we introduced the stabilizer partition function (SPF), a novel quantity that encapsulates both the stabilizer and non-stabilizer (magic) structure of a quantum state.
We further demonstrated that the SPF distinguishes different universality classes of magic states throughout Hilbert space.
We now use the SPF to quantify magic and show that it obeys some desirable resource-theoretic properties.
In particular, we show that interpreting the Pauli spectrum as the statistical description of an effective many-body system naturally leads to a well-defined magic monotone. 
Motivated by the role of free energy in thermodynamics, we introduce \textit{stabilizer work}, which quantifies the amount of useful magic contained in a quantum state.

As in any quantum resource theory~\cite{RevModPhys.91.025001}, one must first specify the free states and free operations.
Since the SPF is the central object of our framework, it provides a natural route to defining the corresponding free operations.
Considering the convex hull of pure stabilizer states as the set of free states, we define the maximal set of free operations, denoted by $\mf_{\rm all}$, as those quantum channels that preserve the SPF of stabilizer states. 
Concretely, a quantum channel $\mO\in \mf_{\rm all}(A\to B)$ if $\mZ_{\beta}(\mO(\phi_A)) = \mZ_{\beta}({\rm STAB}_{B})$ for all $\phi\in {\rm STAB}_A$, where the notation $\mf_{\rm all}(A\to B)$ denotes the set of free quantum channels that take a state in Hilbert space $\mH_A$ to state in Hilbert space $\mH_B$. 
By the uniqueness property of the SPF for stabilizer states (Proposition~\ref{Prop:Uniqueness_stabilizer_states}), this definition is equivalent to requiring that stabilizer states are mapped to stabilizer states. (Note that it also works for mixed states, as we use the convex roof extension to generalize the SPF definition for all quantum states, see Definition~\ref{def:mixed_state_SPF}.)
Consequently, the class $\mf_{\rm all}$ coincides with the stabilizer-preserving operations (SPOs) introduced in Ref.~\cite{PhysRevA.97.062332}. This class contains the stabilizer operations as well as the completely stabilizer-preserving operations (CSPOs) studied in Refs.~\cite{Seddon2019Jul, PhysRevA.106.042422}.

Throughout this work, we restrict our attention to pure-state transformations. 
Accordingly, the free states are pure stabilizer states, the resource states are pure non-stabilizer (magic) states, and the set of free operations, denoted by $\mf$, consists of deterministic pure-state stabilizer protocols. 
These include Clifford unitaries, preparation of stabilizer states, measurements in the computational basis, and partial trace operations, provided that the resulting state remains pure.
In the Discussion section (Section~\ref{sec:discussion}), we also provide a characterization of these operations in terms of Pauli Transfer matrices, as it is very natural to cast the conditions of the Pauli spectrum in terms of PTM.

To construct a magic monotone, we employ the core stabilizer partition function (cSPF). 
The reason being that SPF contains a leading-order contribution that scales with the number of Pauli operators and therefore, does not combine multiplicatively under tensor products in a manner consistent with the higher-order moment structure.
The cSPF removes this contribution, yielding a quantity with more natural compositional properties. This motivates the following definition.
\begin{definition}[Stabilizer Work]
Given an $n$-qubit state, the stabilizer work is defined as
\begin{equation}
    \mathcal{W}_{\beta}(\psi):=-\log \mathcal{Z}^c_{\beta}(\psi)+\log \mathcal{Z}^c_{\beta}({\rm STAB_{n}})\,,
\end{equation}
where the logarithm is taken with respect to base 2.
\end{definition}

The stabilizer work displays the following properties:
\begin{enumerate}
    \item Faithfulness. $\mathcal{W}_{\beta}(\psi_{A})=0 \iff \psi_A\in {\rm STAB}_A$. 
    \item Clifford invariance. $\mW_{\beta}(\mU(\psi))=\mW(\psi)\quad \forall \, \mU\in \mC l_n $.
    \item Sub-additivity. $\mathcal{W}_{\beta}(\psi_A \otimes \chi_B)\leq \mathcal{W}_{\beta}(\psi_{A}) + \mathcal{W}_{\beta}(\chi_{B})$. 
    \item Monotonicity under discarding systems. $\mW_{\beta}(\psi_A\otimes \chi_B)\geq \mW_{\beta}(\psi_A)$.
    \item Invariant under appending a stabilizer state. $\mW_{\beta}(\psi_A)=\mW_{\beta}(\psi\otimes \phi_B)$ for any $\phi_B\in {\rm STAB_B}$. 
\end{enumerate}
Figure~\ref{fig:SW} illustrates the dependence of stabilizer work on the parameter $\beta$ for several families of quantum states.

\begin{figure}
\includegraphics[width=\linewidth]{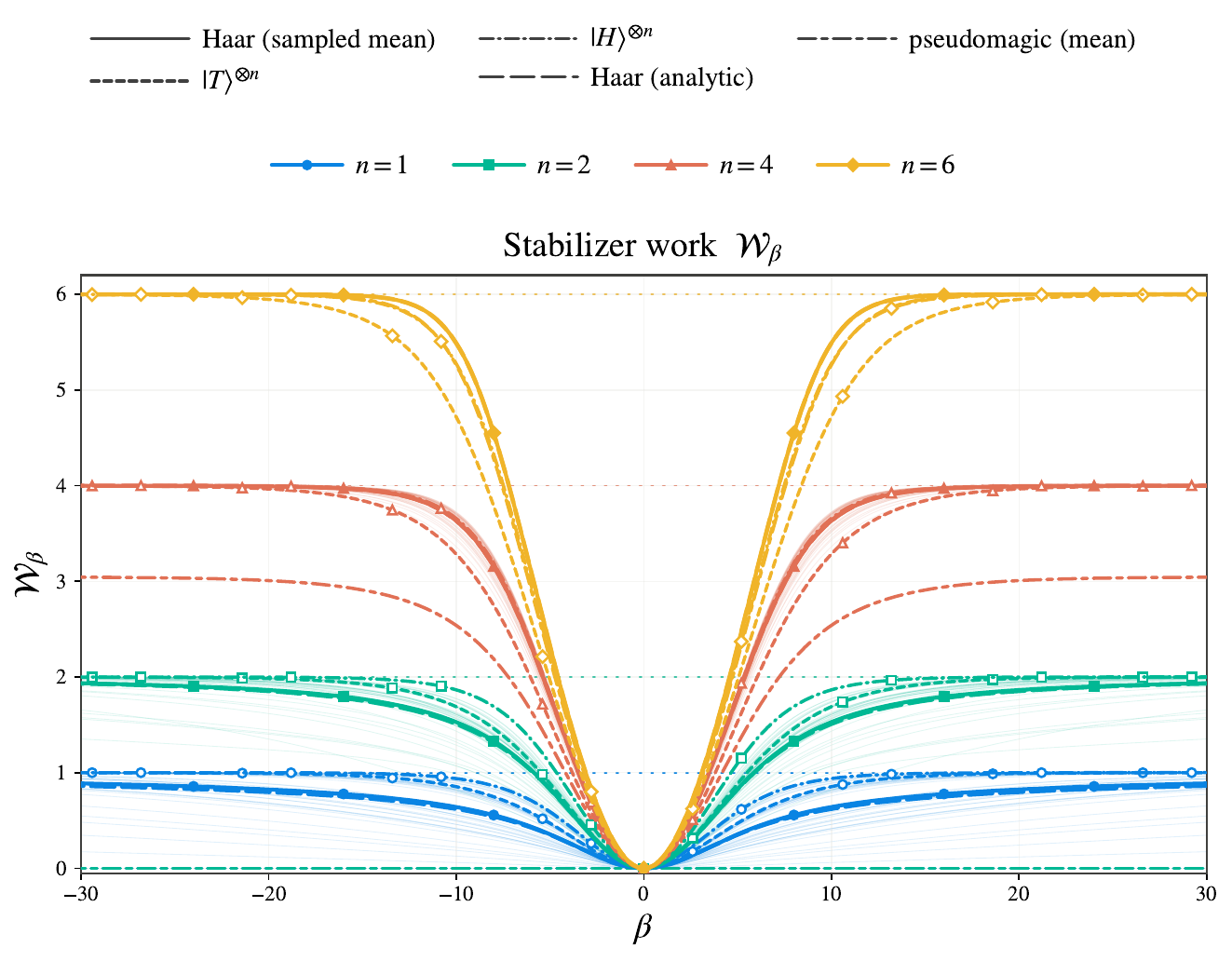}
\caption{Stabilizer work as a function of the parameter $\beta$. Haar-random states form a narrowing concentration band as the number of qubits increases, reflecting the concentration of Pauli spectra in high-dimensional Hilbert spaces. Stabilizer work for pseudomagic states progressively approach that of the typical Haar states with increasing system size, making them increasingly difficult to distinguish from generic random states using coarse-grained observables.
We also compare the tensor product families of $\ket{T}$ and $\ket{H}$ states where $\ketbra{T}{T}= (I + (X+Y)/\sqrt{2})/2$ and $\ketbra{H}{H}= (I + (X+Y+Z)/\sqrt{3})/2$.
While $\ket{H}$ has maximum stabilizer work for a single qubit, the state $\ket{H}^{\otimes n}$ ceases to have maximum stabilizer work for $n\geq 2$.
}
\label{fig:SW}
\end{figure}

We now justify the properties listed above.
Faithfulness follows directly from the uniqueness of the cSPF for stabilizer states.
Specifically, if $\psi$ is a stabilizer state, then by construction $\mathcal{W}_{\beta}(\psi)=0$. Conversely, if $\mathcal{W}_{\beta}(\psi)=0$, then $\mathcal{Z}^{c}_{\beta}(\psi)$ must coincide with the unique stabilizer value, implying that $\psi$ is a stabilizer state. Clifford invariance follows immediately from Proposition~\ref{Prop:Clifford_invariance}.

\begin{theorem}[Sub-additivity of Stabilizer work]
    For arbitrary pure states $\psi_A$ and $\phi_B$, 
    \eqs{
    \mathcal{W}_{\beta}(\psi_A \otimes \phi_B)\leq \mathcal{W}_{\beta}(\psi_{A}) +  \mathcal{W}_{\beta}(\phi_{B})\,.
    }
\end{theorem}
\begin{proof}
The above result on the sub-additivity of stabilizer work follows directly from the supermultiplicativity of the cSPF established in  Lemma~\ref{lem:cSPF_supermultiplicativity}, together with the uniqueness of the cSPF for stabilizer states.
\end{proof}

Furthermore, adopting the notation used in the proof of Lemma~\ref{lem:cSPF_supermultiplicativity}, we have that the coefficients $b_k$ satisfy $b_k\leq d_B$.
Consequently, any probabilistic combination will yield $\sum_k p_k b_k\leq d_B$, which is sufficient to establish monotonicity under subsystem discarding. 
Equality is achieved only when the discarded subsystem is a stabilizer state, yielding invariance under the addition of stabilizer ancillas.

Next, we derive fundamental limitations on stabilizer work given a particular value of $\beta$.
\begin{corollary}[Bounds on Stabilizer Work]
    For any $\beta \in \mbR$ and any $n$-qubit state $\psi_{_A}\in \mH_A$ with dimension $d_A = 2^n$,
  \begin{equation}
  0\leq \mathcal{W}_{\beta}(\psi)\leq \log\left(\frac{d_A \cosh(\beta)-d_A}{\cosh(\beta)+(d_A^2-1)\cosh\left(\frac{\beta}{\sqrt{d_A + 1}}\right)-d^2_A} \right)     
  \end{equation}
\end{corollary}
The proof of the above corollary follows directly from Theorem~\ref{theorem:SPF_state_independent_bounds}.

We now examine the behavior of stabilizer work in two limiting regimes.
In the limit $\beta \to \infty$ (low-temperature limit), the stabilizer work converges to stabilizer-nullity, recovering the support-based characterization of magic.
Conversely, for very-small $\beta$, i.e., in the limit $\beta \to 0$ (high-temperature limit), the cSPF becomes highly insensitive to higher moments of the Pauli spectrum, and the stabilizer work reduces, at leading order, to a function of the second-order stabilizer R\'enyi entropy alone.
Hence, stabilizer work provides a continuous and monotonic interpolation between the two well-established measures of non-stabilizerness that have hitherto been treated independently.

Finally, by inheriting the properties of SPF, stabilizer work serves not only as a monotone but also as a diagnostic tool for classifying different classes of magical states. We next discuss some of SW's properties and its classification abilities.

\subsection{Stabilizer Work Properties}

In the above section, we introduced Stabilizer work as the canonical magic monotone arising from the SPF.
We now investigate its robustness and typical behavior for the same state ensembles considered in Section~\ref{sec:SPF_across_Hilbert_space}. We begin by establishing a continuity bound analogous to that satisfied by the SPF.
\begin{lemma}[Continuity of the Stabilizer Work]
\label{lemma:stab_work_lipshizts}
For every $n$-qubit state $\psi$, the stabilizer work, $\mathcal{W}_{\beta}(\psi)$, is a Lipschitz-continuous function of the state i.e., 
\eqs{|\mW_{\beta}(\psi)-\mW_{\beta}(\phi)|\leq K_n(\beta)\|\psi-\phi\|_1
}
where the Lipschitz constant is $K_{n}(\beta)=2\beta^{-1}\sinh{\beta}$.
\end{lemma}
\begin{proof}
Let $\psi$ and $\phi$ be two $n$-qubit pure states. From the definition of
the stabilizer work and the inequality
$|\log(a)-\log(b)|\leq |a-b|/\min\{a,b\}$ for positive $a,b$, we have
\eqs{\label{eq:tighter_ub_SW_diff}
\left|\mathcal{W}_{\beta}(\psi)-\mathcal{W}_{\beta}(\phi)\right|
\leq\frac{
d_{n}e^{-\beta}\beta \sinh{\beta}}{\min\left\{
\mathcal{Z}_{\beta}^{c}(\psi),
\mathcal{Z}_{\beta}^{c}(\phi)\right\}}\left||\phi-\psi|\right|_{1},
}
where we used the Lipschitz property of the SPF ( Theorem~\ref{Theorem:Lipshitzs_property}). 
Finally, by applying the elementary inequality $\cosh(x)-1\geq x^2/2$ in the SPF lower bound from Theorem~\ref{theorem:SPF_state_independent_bounds}, we conclude the proof.
\end{proof}

We desire to point out that, unlike the corresponding Lipschitz constant for the SPF, the Lipschitz constant of the stabilizer work does not grow exponentially with the number of qubits. 
A direct consequence of this feature is that any concentration statement about the Stabilizer work are stronger than those obtained from the SPF.
This improvement arises because stabilizer work is defined through a logarithmic ratio of partition functions and is therefore insensitive to the extensive normalization factor present in the SPF.

There is, however, a fundamental difference between the average SPF of an ensemble and its average stabilizer work.
Recall that Jensen's inequality implies that the average of $-\log(X)$ is not equal to the $-\log$ of the average of $X$ where $X$ is a random variable, since $-\log$ is a convex function.
In our case, this translates to saying that for an ensemble of states $\mathcal{E}$, $\mathbb{E}_{\psi \sim \mathcal{E}}[\mathcal{W}_{\beta}(\psi)]\neq \log (\mathbb{E}_{\psi \sim \mathcal{E}}[\mathcal{Z}^{c}_{\beta}(\psi)])$.
Therefore, although the concentration statement derived from Lemma~\ref{lemma:stab_work_lipshizts} guarantees that typical states' values lie exponentially close to the ensemble average value of stabilizer work, determining this average still requires an independent analysis.
The following lemma provides a quantitative bound on these fluctuations.

\begin{lemma}[Small gap]
\label{Lemma:Small_Gaps}

Let $\mathcal{E}$ be an ensemble of $n$-qubit states. Denote by $\mathcal{Z}_{\beta}(\mathcal{E})$ the ensemble-averaged SPF, and $\mathcal{W}_{\beta}(\mathcal{E})=-\log\mathcal{Z}^{c}_{\beta}(\mathcal{E})+\log{\mathcal{Z}^{c}_{\beta}({\rm Stab}_{n})}$.
Then, with negligible failure probability, every $\psi \in \mathcal{E}$ satisfies
\begin{equation*}
    \mathcal{W}_{\beta}(\psi)-\mathcal{W}_{\beta}(\mathcal{E}) \in [-{\rm logPoly}(n), \mathcal{O}(\beta^{4}) + {\rm LogPoly}(n)].
\end{equation*}
\end{lemma}
\begin{proof}
    The proof follows by controlling the tails of the Stabilizer work using Markov's inequality (see Supplemental Material for details).
\end{proof}

We can now apply the previous lemma to SPS vs $\nu$-compressible states. The trick is the following as the upper bound is controlled by an $n$-independent $\mathcal{O}(\beta^{4})$ growth. The SW distinguish between both ensembles whenever
\begin{equation}
\left|
\mathcal{W}_{\beta}({\rm SPS})
-\mathcal{W}_{\beta}({\rm Comp}_{\nu})
\right|> {\rm LogPoly}(n)
\end{equation}
Indeed, under this condition the two confidence intervals are disjoint, and any threshold chosen between them classifies the two ensembles with
negligible failure probability. Therefore, the exact average SPF formulas provide a stabilizer-work quantification of pseudomagic whenever their
induced work gap exceeds the combined confidence correction.

\section{Applications \label{sec:Applications}}

\subsection{Applications for low-rank stabilizer simulation}\label{sec:low-rank-simulation}

A direct application of the SPF (concretely Theorem~\ref{Theorem:Lipshitzs_property}) is in the context of quantum simulations via low-rank stabilizer decompositions~\cite{PhysRevX.6.021043}.
The main idea is that, given a (possibly magical) state $\ket{\psi}$, we can always find a minimal set of stabilizer states $\{\ket{\phi_l}\}$ such that $\ket{\psi}=\sum_{l=1}^{\chi(\psi)}c_{l}\ket{\phi_l}$. The size of this minimal set $\chi(\psi)$ is called the stabilizer rank of the state $\ket{\psi}$~\cite{PhysRevX.6.021043}.
In practical applications, one is interested rather in the approximated stabilizer rank $\chi_{\delta}(\psi)$, i.e., the smallest integer $k$ such that $||\psi -\psi^{\prime}|| \leq \delta$ for some reference state $\ket{\psi^{\prime}}$ with known exact stabilizer rank $\chi(\psi^{\prime})=k$. 
Directly computing or providing tight bounds for the exact stabilizer rank is difficult, and for this reason one is interested in more amenable quantities such as the stabilizer fidelity or the stabilizer extent.

\begin{definition}[Stabilizer Fidelity \cite{Bravyi/stab_fidelity/2019}]
The stabilizer fidelity, $F_{\rm{Stab}}(\psi)$, of an $n$-qubit state $\ket{\psi}$ is defined as $F_{\rm{Stab}}(\psi) := \max_{\ket{\phi} \in \rm{STAB}}|\braket{\phi|\psi}|^{2}$, where the maximization is over all $n$-qubit stabilizer states $\ket{\phi}$.
\end{definition}

The stabilizer fidelity $F_{\text{Stab}}(\psi)$ measures the overlap between a target state $\ket{\psi}$ and its closest stabilizer state. 
Closely related with the stabilizer fidelity is the stabilizer extent $\xi(\psi)$ defined as the minimum $\|c_{l}\|_{1}$ over all stabilizer decompositions $\ket{\psi}=\sum_{l \geq 1} c_{l} \ket{\phi_l}$, with $\{\ket{\phi_l}\}_l$ being the set of $n$-qubit stabilizer states. 
In \cite{Bravyi/stab_fidelity/2019}, it was shown, by a convex duality result, that $\xi(\psi)\geq F_{\text{Stab}}^{-1}(\psi)$. 
Furthermore $\chi_{\delta}(\psi) \leq 1 + \delta^{-2}\xi(\psi)$, making the stabilizer fidelity/extent the figures of merit to determine the stabilizer rank.

For us, the key observation resides in that we can make direct use of Lemma.~\ref{lemma:stab_work_lipshizts} to provide a family of upper bounds to the stabilizer fidelity (thereof lower-bound the extent) via the stabilizer Work.

\begin{proposition}[Stabilizer fidelity upper bounds]
\label{proposition:stab_fidellity_upper_bound}
Given a pure state $\ket{\psi}$, its stabilizer fidelity $F_{\rm{Stab}}(\psi)$ follows the family of upper bounds,
\begin{equation*}
    F_{\rm{Stab}}(\psi) \leq 1- \frac{1}{2}\left(\frac{ \cosh(\beta)-1}{\beta\sinh{\beta}}2^{-\mathcal{W}_{\beta}(\psi)}\mathcal{W}_{\beta}(\psi)\right)^{2} \;\; \forall\, \beta\geq 0.
\end{equation*}
Here the tightest possible of such bounds is obtained by optimizing over $\beta$. 
\end{proposition}

\begin{proof}
For any pure stabilizer state $\phi$, faithfulness plus the Lipschitz property of stabilizer work (Lemma~\ref{lemma:stab_work_lipshizts}) gives $\mathcal{W}_\beta(\psi) \leq 4\beta^{-1}\sinh\beta \sqrt{1-|\langle\phi|\psi\rangle|^2}$. 
The tight upper bound in the above result is obtained by using Eq.~\eqref{eq:tighter_ub_SW_diff}.
The right-hand side is independent of which stabilizer state \(\phi\) is used, so maximizing over all pure stabilizer states proves the claim. 
\end{proof}

In the supplementary materials Section~\ref{sm:fidelity_ub}, we also provide an illustration of these upper bounds on fidelity for the $\ket{T}$ and $\ket{H}$ state families.

The above bounds exploit the uniqueness of the SPF for stabilizer states, thereby reducing the problem of maximizing the fidelity over the set of stabilizer states to that of finding the maxima, in the case of the tightest bound, of a real function. These upper bound is complemented by known \cite{Tobi/Stab_Fidelity/2023,Gu/pseudomagic/2024} lower bound provided by the stabilizer R\'enyi entropy $\frac{\alpha-1}{2\alpha} M_{\alpha}(\psi) \leq F_{\text{Stab}}(\psi)$.

Remarkably, proposition~\ref{proposition:stab_fidellity_upper_bound} provides an operational interpretation to the SPF thereof SW in the context of stabilizer approximations:

\begin{task}
Given an $n$-qubit pure state $\psi$, determine whether there exists a stabilizer state that is $\epsilon$-close in fidelity to $\psi$; that is, whether there exists $\phi\in\mathrm{STAB}_{n}$ such that $F(\psi,\phi)\ge 1-\epsilon$. 
\end{task}

This question can be answered by checking whether the required fidelity exceeds the upper bound provided by Proposition~\ref{proposition:stab_fidellity_upper_bound}. This task is an special instance of the more general problem of tolerant stabilizer property testing \cite{Dutt/tolerant_stab_testing_poly/2024,Bao/tolerant_stab_testing_improved/2024}, ${\rm e.g.}$ given copies of a state with the promise that it is either $\epsilon_{1}$-close to a stabilizer state or $\epsilon_{2}$-far away, decide which one is the case. The procedure is captured by the following simple algorithm.
\vspace{8pt}
\begin{algorithm}
\caption{Fidelity testing}\label{alg:fidelity_testing}
\begin{algorithmic}[1]
\Procedure{FidTesting}{}
\Require{$\epsilon \ge 0$ \Comment{tolerance}}
\Ensure Void()
\State Compute $\sqrt{\epsilon_{\star}}\gets \frac{ \cosh(\beta)-1}{\sqrt{2}\beta\sinh{\beta}}2^{-\mathcal{W}_{\beta}(\psi)}\mathcal{W}_{\beta}(\psi)$ 
      \Comment{maximize over $\mathbb{R}^{+}$ if required}
\If{$\epsilon^{\star} \ge \epsilon$}
    \State Print( No stabilizer state is $\epsilon$-close to $\ket{\psi}$)
\EndIf
\EndProcedure
\end{algorithmic}
\end{algorithm}
\vspace{8pt}
${\rm Example}$ --- To ground the previous discussion, consider the particular case of $n$ copies of the $\ket{T}$ state. In this setting, due the discarding property, the stabilizer work admits the lower bound $ \mathcal{W}_{\beta}(\ket{T}^{\otimes n}) \;\ge\; \mathcal{W_{\beta}}(\ket{T})$. Then, we obtain an explicit upper  bound on the stabilizer fidelity: 
\begin{equation*}
    F_{\mathrm{Stab}}(\ket{T}^{\otimes n}) \;\le\;1 - \frac{1}{2}\left( \frac{ \cosh(\beta)-1}{\beta\sinh{\beta}}2^{-\mathcal{W}_{\beta}(\ket{T})}\mathcal{W}_{\beta}(\ket{T})\right)^{2}
\end{equation*}
for all $\beta >0$. In particular, by an straightforward numerical maximization, no pure stabilizer state can be $\epsilon$-close in fidelity to $\ket{T}^{\otimes n}$ for any $\epsilon \lesssim 2.27\times 10^{-3} $. Notice that this bound is independent of the number of copies and therefore also holds in the asymptotic $n\to \infty$ limit.

\subsection{Resource Interconversion}\label{subsec:interconversion}

The family $\{\mathcal{W}_\beta\}_{\beta\neq0}$ converts directly into interconversion bounds.
Every value of $\beta$ furnishes an independent conversion witness, and optimizing over the family yields the strongest statement. 
Throughout this subsection, the set of free operations $\mf$ refers to the protocol class of Sec.~\ref{sec:SW} under which $\mathcal{W}_\beta$ is a monotone.

\begin{theorem}[One-shot interconversion bound]
\label{thm:interconversion}
Let $\psi_A\in \mH_A$ and $\phi_B\in \mH_B$, and suppose that $n$ copies of $\psi$ can be converted into $m$ copies of $\phi$ by free operations. 
Then, for every $\beta\neq0$,
\begin{equation}
\mathcal{W}_\beta\!\big(\phi^{\otimes m}\big)\;\le\;
\mathcal{W}_\beta\!\big(\psi^{\otimes n}\big)\;\le\;
n\,\mathcal{W}_\beta(\psi),
\label{eq:interconv}
\end{equation}
where both tensor-power quantities can be expressed in closed form from the Pauli moments of a single copy.
\end{theorem}

\begin{proof}
The first inequality in Eq.~\eqref{eq:interconv} is monotonicity of
$\mathcal{W}_\beta$ under free operations, and the second is
subadditivity.
\end{proof}

We remark here that Theorem~\ref{thm:interconversion} is strictly stronger than any bound expressed through single-copy values alone whenever subadditivity is strict, that is, whenever neither state is a stabilizer state, since the exact tensor-power evaluation then improves on the corresponding single-copy ratio bound.

The family also supplies no-go statements for conversion.
If $\mathcal{W}_\beta(\psi)<\mathcal{W}_\beta(\phi)$ for all $\beta\neq0$, then $\psi$ cannot be converted into $\phi$.
More strongly, because each $\beta$ is an independent witness, the family certifies incomparability.
In other words, if the two curves cross, i.e., $\mathcal{W}_{\beta_1}(\psi)>\mathcal{W}_{\beta_1}(\phi)$ while $\mathcal{W}_{\beta_2}(\psi)<\mathcal{W}_{\beta_2}(\phi)$ for some $\beta_1,\beta_2$, then no free protocol converts either state into the other, which is a two-directional no-go that no single monotone can express.
Thus, these conditions are necessary but not sufficient.
As an example, take $\ketbra{H}{H}=(I+(X+Y+Z)/\sqrt{3})/2$ and $\ketbra{T}{T}=(I+(X+Y)/\sqrt{2})/2$.
Here $\mathcal{W}_{\beta}(\ketbra{H}{H})\ge\mathcal{W}_{\beta}(\ketbra{T}{T})$ for all $\beta\neq0$, so Eq.~\eqref{eq:interconv} does not obstruct $\ket{H}\to\ket{T}$; nevertheless, no deterministic stabilizer protocol in $\mF$ realizes this conversion.
Equation~\eqref{eq:interconv} is therefore a necessary but not sufficient condition for convertibility.

We can, however, use this necessary condition to lower bound the magic cost and distillation.
\begin{corollary}[Magic cost]
\label{cor:Tcost}
The $T$-cost of a state $\phi$ is the least number $t$ of
$\ket{T}$ states convertible to $\phi$ by free operations. Setting
$\psi=\ket{T}$, $n=t$, and $m=1$ in Eq.~\eqref{eq:interconv},
\begin{equation}
t\;\ge\;\sup_{\beta\neq0}\;
\min\Big\{t'\,:\,\mathcal{W}_\beta\!\big(T^{\otimes t'}\big)\ge
\mathcal{W}_\beta(\phi)\Big\},
\label{eq:tcost-exact}
\end{equation}
The weaker but simpler bound $t\ge\sup_\beta\mathcal{W}_\beta(\phi)/\mathcal{W}_\beta(T)$ follows from Eq.~\eqref{eq:tcost-exact} by subadditivity.
\end{corollary}

\begin{corollary}[Distillation Converse]
\label{cor:distill}
If $n$ copies of $\psi$ can be converted into $m$ copies of $\ket{T}$
by free operations, then for every $\beta\neq0$,
\begin{equation}
\mathcal{W}_\beta\!\big(T^{\otimes m}\big)\;\le\;n\,
\mathcal{W}_\beta(\psi)\,,
\end{equation}
a few-copy converse bound on distillation.
\end{corollary}

\subsection{Electronic structure of $\mathrm{H}_2$ dissociation \label{subsec:h2_dissoc}}

\begin{figure}
    \centering
  \includegraphics[width=\linewidth]{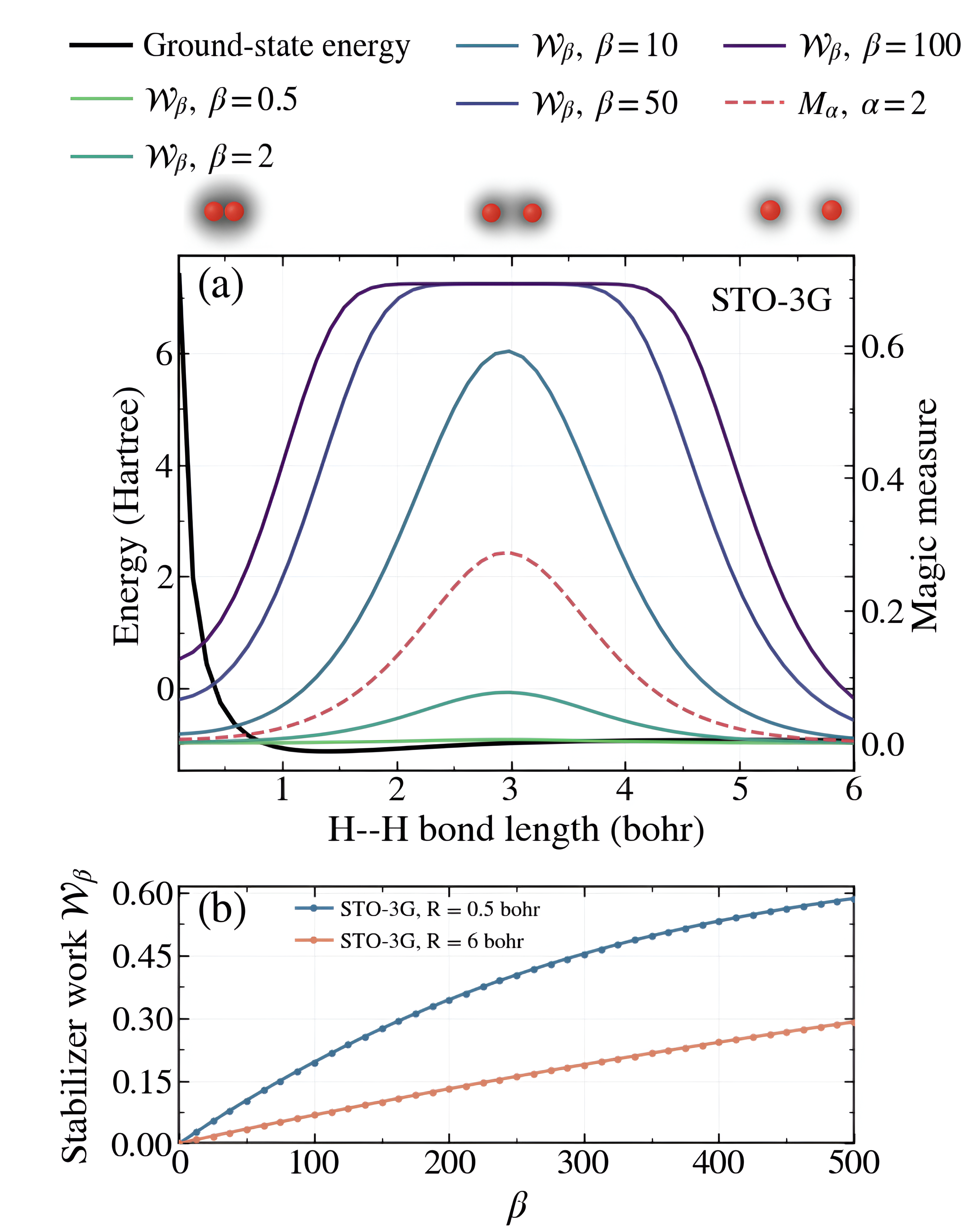}  
  \caption{Stabilizer work and 2-SRE along the $\mathrm{H}_2$ dissociation curve in the minimal STO-3G basis (four spin-orbitals, four qubits). (a)~Ground-state energy (left axis, black solid line) and the magic measures (right axis) as functions of the H--H bond length. The solid colored lines show the stabilizer work $\mathcal{W}_{\beta}$ for $\beta\in\{0.5,\,2,\,10,\,50,\,100\}$, while the red dashed line corresponds to the $\alpha=2$ stabilizer R\'{e}nyi entropy $M_{\alpha=2}$. All magic quantities are computed from the exact ground-state statevector obtained via full diagonalization of the Jordan--Wigner qubit Hamiltonian. 
  The pronounced peak in nonstabilizerness in the intermediate, strongly correlated region of the dissociation coordinate reflects the static electron correlation accompanying bond breaking, while within this minimal basis the magic decreases toward the dissociation limit, where the wavefunction approaches an essentially two-configuration covalent singlet.
  Increasing $\beta$ sharpens the features of the stabilizer work and resolves finer spectral structure of the Pauli spectrum. (b)~Stabilizer work $\mathcal{W}_{\beta}$ as a function of $\beta$ for two representative bond lengths ($R=0.5$ and $R=6$\,bohr). The monotonic growth and smooth convergence toward the asymptotic $\beta\to\infty$ limit are evident; within STO-3G the compressed geometry ($R=0.5$\,bohr) carries higher magic and approaches its asymptotic value more slowly than the near-dissociated geometry ($R=6$\,bohr), while the ordering between the two configurations is preserved at all values of $\beta$. As discussed in the main text and in Sec.~\ref{supp:h2_631G}, both this dissociation-limit decay and this geometry ordering are basis-set sensitive.}
  \label{fig:h2_dissoc}
\end{figure}

The dissociation of molecular hydrogen provides one of the simplest yet most physically revealing testbeds for studying electron correlation in quantum chemistry, and it has recently served as a convenient setting in which to examine the nonstabilizerness of molecular ground states. At the equilibrium bond length, the electronic ground state of $\mathrm{H}_2$ is well described by a single Slater determinant, and standard mean-field approaches such as Hartree--Fock theory yield accurate results. As the internuclear separation increases, however, static correlation effects become increasingly important: the ground-state wavefunction evolves into a superposition of nearly degenerate configurations whose relative weights must be carefully balanced to avoid unphysical spin contamination and to correctly recover the covalent bond-breaking limit~\cite{sarkis2025moleculesmagicalnonstabilizernessmolecular}. This crossover from a weakly to a strongly correlated regime is precisely the setting in which classical simulation methods encounter their greatest difficulties, and it is natural to ask whether the nonstabilizerness of the electronic ground state reflects this change in computational complexity. Indeed, Ref.~\cite{sarkis2025moleculesmagicalnonstabilizernessmolecular} recently reported that magic, as quantified by stabilizer R\'{e}nyi entropies, develops a pronounced peak along the bonding curve that tracks the correlation structure of the molecule. Here we revisit this problem through the lens of the stabilizer partition function and the stabilizer work introduced in the preceding sections, with the modest aim of illustrating that these quantities furnish a flexible diagnostic of nonstabilizerness. We deliberately refrain from drawing strong physical conclusions, for reasons that the basis-set analysis below and in the Supplemental Material make clear.

We consider the electronic structure of $\mathrm{H}_2$ within the minimal STO-3G Gaussian basis set, which yields four spin-orbitals and hence a four-qubit problem under the Jordan--Wigner transformation. The second-quantized electronic Hamiltonian is mapped onto a qubit Hamiltonian, and the ground-state wavefunction at each bond length is obtained by exact diagonalization (equivalently, full configuration interaction within the chosen basis). This procedure yields the statevector $\ket{\psi(R)}$ as a function of the internuclear distance $R$, from which both the stabilizer partition function $\mathcal{Z}_{\beta}(\psi(R))$ and the stabilizer work $\mathcal{W}_{\beta}(\psi(R))$ can be evaluated directly for arbitrary values of the inverse-temperature-like parameter $\beta$. We stress at the outset that the minimal basis lies very far from the complete-basis-set limit, and that the nonstabilizerness of a molecular ground state is not an intrinsic, basis-independent property: it depends on the one-particle basis and on the fermion-to-qubit encoding. The results below should therefore be understood as a demonstration of the stabilizer partition function framework rather than as converged statements about the hydrogen molecule. To make this caveat concrete, we repeat the entire analysis in the larger split-valence 6-31G basis in the Supplemental Material (Sec.~\ref{supp:h2_631G}), where several of the features seen here prove to be basis-set sensitive.

Figure~\ref{fig:h2_dissoc}(a) displays the ground-state energy alongside the stabilizer work $\mathcal{W}_{\beta}$ for several representative values of $\beta$, together with the $\alpha=2$ stabilizer R\'{e}nyi entropy $M_{\alpha=2}$, all plotted as functions of the H--H bond length. Within this minimal basis the magic measures remain modest near the equilibrium geometry, where the near single-determinant character of the wavefunction keeps the nonstabilizerness low, rise to a clear and pronounced peak in the intermediate, strongly correlated region of the dissociation coordinate, and then decrease toward the dissociation limit. This behavior is consistent with the intuitive picture that breaking the chemical bond requires a delicate superposition of nearly degenerate electronic configurations, and therefore enhanced nonstabilizerness, whereas in the STO-3G dissociation limit the wavefunction reduces to an essentially two-configuration covalent singlet of correspondingly lower magic. We caution, however, that the apparent decay of magic at large $R$ is a feature of the minimal basis and does not survive enlargement of the one-particle space, as we show in the Supplemental Material.

A notable feature visible in \figref{fig:h2_dissoc}(a) is the dependence of the stabilizer work on the parameter $\beta$. At small $\beta$ (for instance $\beta=0.5$), the stabilizer work captures primarily the lowest-order moments of the Pauli spectrum and consequently provides a relatively coarse probe of the nonstabilizerness, resulting in a smooth and comparatively featureless curve. As $\beta$ increases, higher-order moments contribute more prominently and the stabilizer work develops increasingly sharp features that resolve the fine structure of the magic landscape along the dissociation coordinate. At large $\beta$ (for instance $\beta=100$), the stabilizer work saturates to a profile that tracks the overall shape of $M_{\alpha=2}$, but with quantitative differences arising from the fact that $\mathcal{W}_{\beta}$ encodes the full moment hierarchy rather than a single R\'{e}nyi index. This tunability through $\beta$ constitutes a genuine and basis-independent advantage of the stabilizer partition function framework: by scanning over $\beta$, one can interpolate continuously between a coarse global characterization of magic and a fine-grained resolution of its spectral structure, all within a single unified quantity.

Figure~\ref{fig:h2_dissoc}(b) examines the convergence of the stabilizer work as a function of $\beta$ for two representative bond lengths, $R=0.5$ and $R=6$\,bohr. For both configurations, $\mathcal{W}_{\beta}$ grows monotonically with $\beta$ and converges to a well-defined asymptotic value, consistent with the low-temperature limit $\lim_{\beta\to\infty}\mathcal{Z}_{\beta}(\psi)=|\mathrm{Stab}(\psi)|$ discussed in the main text. In the minimal basis the compressed geometry at $R=0.5$\,bohr carries a markedly higher stabilizer work than the near-dissociated geometry at $R=6$\,bohr, with convergence toward the asymptotic regime correspondingly faster for the lower-magic ($R=6$) configuration; the ordering between the two geometries is preserved across the entire range of $\beta$. As we emphasize in the Supplemental Material, however, this particular ordering is itself basis-set dependent and in fact reverses in 6-31G, a further indication that the dissociation-limit behavior cannot yet be trusted at the quantitative level.

Taken together, these results show that the stabilizer partition function and the stabilizer work provide a flexible and physically transparent diagnostic of nonstabilizerness across the full dissociation pathway, complementing and extending the insights afforded by the stabilizer R\'{e}nyi entropies alone. 

\begin{figure*}[t]
    \centering
  \includegraphics[width=\textwidth]{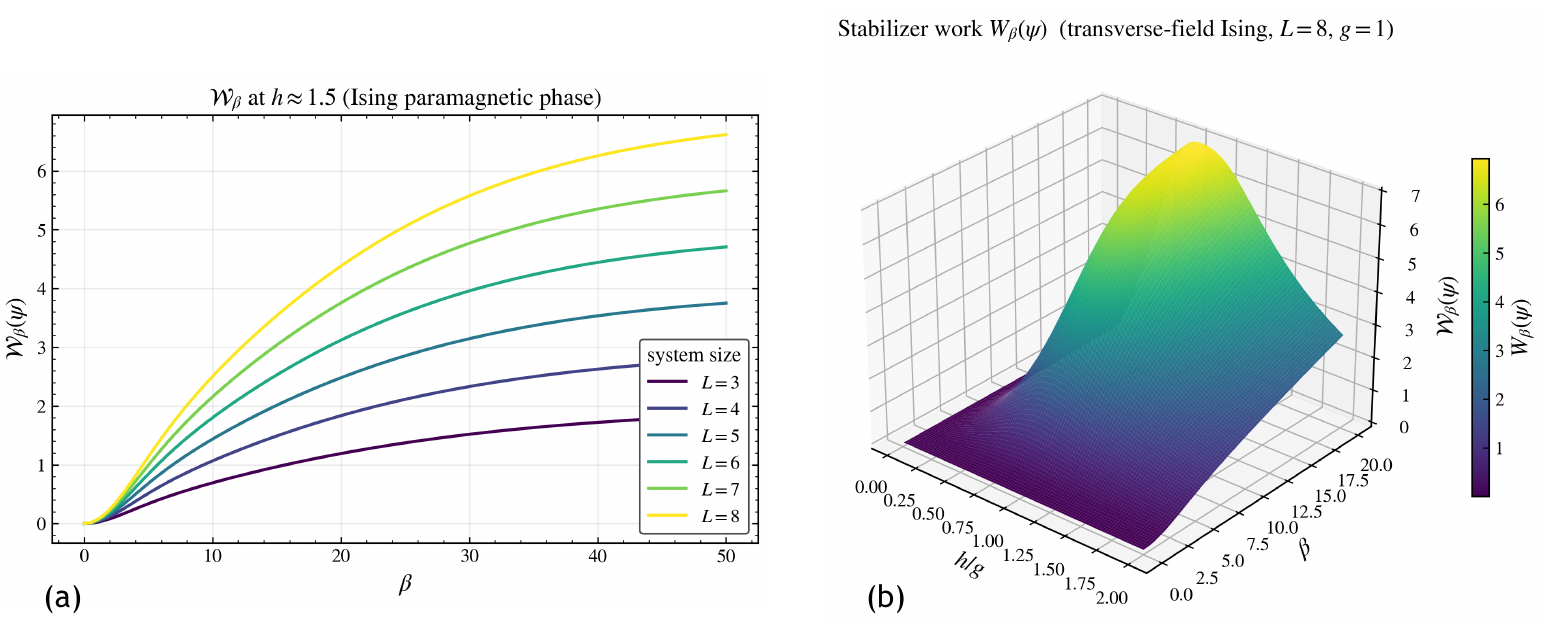}  
  \caption{Stabilizer work through the lens of the quantum Ising model. (a) Stabilizer work is plotted against tunable parameter $\beta$ at the Ising paramagnetic phase $(h=1.5, g=1)$ for varying system sizes. By construction, $W_\beta \rightarrow 0$ when $\beta \approx 0$ while it plateaus when $\beta$ is large. (b) Surface of stabilizer work is plotted against $h/g$ and $\beta$ for $L=8$.}
  \label{fig:IsingSPF1}
\end{figure*}

\subsection{Quantum Ising model}

The quantum Ising model is one of the most important theoretical models in condensed matter physics \cite{sachdev2011quantum}, statistical mechanics \cite{anonquantum,kyaw2020dynamical}, and quantum information science \cite{bastidas2018floquet,kyaw2018cluster}. Although mathematically simple, it captures profound physical phenomena such as symmetry breaking, quantum fluctuations, criticality, quantum phase transitions and dynamical quantum phase transitions. Its central value lies in showing how collective behavior can emerge from many interacting microscopic degrees of freedom.

In its simplest one-dimensional form, the transverse-field quantum Ising Hamiltonian is written as
$
H=-g\sum_i \sigma_i^z\sigma_{i+1}^z-h\sum_i \sigma_i^x,
$
where $(\sigma_i^z)$ and $(\sigma_i^x)$ are Pauli operators acting on spin (i). The parameter $J$ describes the interaction between neighboring spins, while $h$ represents an external transverse magnetic field. The first term favors alignment of spins along the $z$-direction, producing ferromagnetic order when $g>0$. The second term introduces quantum fluctuations by attempting to rotate the spins along the $x$-direction.
The competition between these two terms gives rise to a quantum phase transition. When $g\gg h$, the interaction term dominates and the ground state is ordered, i.e., spins tend to align either mostly up or mostly down along the $z$-axis. This phase breaks a discrete $\mathbb{Z}_2$ symmetry. In contrast, when $h\gg g$, the transverse field dominates, and the spins align along the $x$-direction, forming a quantum paramagnetic phase. Between these limits lies a critical point where the system changes phase even at zero temperature. In one dimension, this critical point occurs at $h=g$.

Unlike classical phase transitions, which are driven by thermal fluctuations, quantum phase transitions are driven by quantum fluctuations. 
At the critical point, the energy gap between the ground state and the first excited state closes in the thermodynamic limit. This produces long-range correlations and scale-invariant behavior. The quantum Ising model is therefore a central example for studying universality, conformal field theory, and critical exponents.

As seen in Fig.~\ref{fig:IsingSPF1}, the variable $\beta$ should be read as a moment-resolving parameter for the
state's Pauli spectrum rather than a physical temperature. 
A small-$\beta$ expansion of the SPF gives
$\mZ_\beta \simeq e^{-\beta}( 4^L + \tfrac{\beta^2}{2}\sum_P\langle P\rangle^2
+ \tfrac{\beta^4}{24}\sum_P\langle P\rangle^4+\cdots)$ where $L$ denotes the system size.
The $O(\beta^2)$ term is fixed by purity
($\sum_P\langle P\rangle^2=2^L$) and carries no magic; the first genuinely
nonstabilizer information appears at $O(\beta^4)$ through
$\sum_P\langle P\rangle^4$, which is precisely the participation entering the
second stabilizer R\'enyi entropy $M_2$. 
Thus the stabilizer work $\mW_\beta$ is a generating function for the entire family of Pauli-weight moments where the well-studied second-order stabilizer R\'enyi entropies are recovered as its leading Taylor coefficients, while increasing $\beta$ progressively reweights the sum toward the largest-magnitude Pauli expectations. 

Panel~(a) makes this concrete in the paramagnetic phase ($h/g\approx1.5$):
$\mW_\beta$ rises monotonically and saturates, with the saturation set by the
large-$\beta$ dominance just described.
We observe a clean ordering with system size that indicates that this quantity is extensive in the paramagnetic phase, as expected for a generic (non-Clifford) ground state. 
The apparent increment is close to $1$ per
site, which is suggestive that a pure stabilizer state has $2^L$ unit-weight Paulis
and hence a saturation of order $(L-1)$.

Panel~(b) extends this to the full $(h/g,\beta)$ landscape and contains a
built-in consistency check that is worth stating explicitly in the text. The two
field extremes are exactly solvable stabilizer points.
At $h/g\to0$ the ground
state is the ferromagnetic GHZ state and at $h/g\to\infty$ it is
$\ket{+}^{\otimes L}$, both of which have vanishing magic. The surface should
therefore approach its stabilizer floor along the $h/g\to0$ edge for every
$\beta$ and decrease again as $h/g$ grows beyond the window shown, and the
visible ridge of nonstabilizerness lives in the interior of the phase diagram,
in the vicinity of the order--disorder critical point at $h/g=1$. Along the
$\beta$ direction the same growth-and-saturation pattern persists at every field,
so the ridge is a robust feature of the converged ($\beta\gg1$) free energy and
not an artifact of a particular moment. 
Thus, we see that nonstabilizerness is generated by the competition between the Clifford-diagonalizable Ising and field terms and is maximal where neither limit dominates, vanishing at the two integrable endpoints where the Hamiltonian is effectively free.
In supplementary materials Section~\ref{sm:Ising}, we also provide a comparison of the finite-size behaviour of the stabilizer work and second-order SRE across the transverse-field Ising transition.

\section{Discussion and future outlook}\label{sec:discussion}

Quantifying magic has long forced a choice between measures one can trust and measures one can compute. In this work we dissolved that choice by changing the object of study.
Instead of assigning a state another scalar, we assigned it a statistical ensemble. 
Interpreting the Pauli spectrum as the energy spectrum of a fictitious many-body system, the Pauli gas, converts nonstabilizerness from a microscopic property of a state into a macroscopic property of an exponentially large ensemble, and hands the problem to statistical mechanics. 
Every subsequent result of this paper is an entry in the resulting dictionary. The canonical partition function of the gas is the stabilizer partition function, an analytic, Clifford-invariant, Lipschitz-continuous compression of the full Pauli spectrum, estimable to additive precision by Bell sampling at uniform cost even deep in the high-magic regime where stabilizer-entropy estimation becomes prohibitive. Its free energy, referenced against the unique stabilizer value, is the stabilizer work, a faithful, subadditive magic monotone at every temperature. Its zero-temperature limit is a third-law statement recovering the stabilizer nullity, its high-temperature limit recovers the stabilizer 2-R\'enyi entropy, and the full temperature profile between them constitutes a magic spectrum carrying strictly more information than either endpoint. Ensemble averaging, meanwhile, produced an analytic catalogue with no counterpart among competing monotones. We derived exact closed forms for Haar-random states, $\nu$-compressible states, and pseudomagic subset phase states, each accompanied by concentration guarantees that make the SPF a classifier of magic classes.

The classifier is sharpest precisely where classification is hardest.
Pseudomagic ensembles are engineered so that no efficient observer can detect their missing magic, and the SPF respects this cryptographic boundary, as it must.
But the framework sees structure that survives the boundary.
Circuit-doped states hide their magic in the ground level of the Pauli gas, lowering the high-temperature plateau of the stabilizer work, while subset phase states hide it in the fine structure of excited levels, preserving the plateau but altering the approach to it.
Distinguishing these two modes of concealment requires comparing coarse- and fine-grained magic simultaneously, a comparison no single number can make; the temperature profile makes it within one object and one measurement protocol.
We illustrated the reach of the framework through applications to low-rank stabilizer simulation, one-shot resource interconversion and magic cost, and physically motivated states in quantum chemistry and many-body physics~\cite{sarkis2025moleculesmagicalnonstabilizernessmolecular,LiuWinter2022ManyBodyMagic,Oliviero/IsingModelMagic/2022}.
The stabilizer work upper bounds stabilizer fidelity and thereby lower-bounds simulation cost.
Each temperature furnishes an independent witness for interconversion and magic cost, and in molecular dissociation and the transverse-field Ising model the temperature scan resolved features of the magic landscape that any fixed R\'enyi index averages away.

The framework as developed here concerns pure states, but it extends beyond them in two natural steps
The SPF admits a convex-roof extension to mixed states.
\begin{definition}[Mixed-state SPF]\label{def:mixed_state_SPF}
    The stabilizer partition function of a mixed state $\rho_{_{A}}\in \md_A$ is defined as a real function $\mathcal{Z}_{\beta}(\rho):\mbR\times \md_A\to \mbR$,
    \eqs{
    \mZ_{\beta}(\rho):= \sup \sum_{i\in [m]} p_i \mZ_{\beta}(\phi_i)
    }
    where the supremum is over all pure state decompositions $\{p_i, \phi_i\}_{i \in [m]}$.
\end{definition}
The supremum is the natural choice here because the SPF is maximized by stabilizer states, so that the roof extension preserves faithfulness, i.e., $\mZ_{\beta}(\rho)$ attains the stabilizer value if and only if $\rho$ admits a decomposition into pure stabilizer states, that is, if and only if $\rho$ is a free state.

We can further generalize the stabilizer partition function to quantum channels or completely positive and trace-preserving (CPTP) maps via the Choi–Jamio{\l}kowski isomorphism.
\begin{definition}[SPF of a quantum channel]
    The stabilizer partition function of a quantum channel $\mN\in \cptp(A\to B)$ is defined as a real function $\mathcal{Z}_{\beta}(\mN):\mbR\times \cptp \to \mbR$,
    \eqs{
    \mZ_{\beta}(\mN):= \mZ_{\beta}(\tilde{J}_{AB}^{\mN})
    }
    where $\tilde{J}_{AB}^{\mN}:=\id\otimes\mN(\Phi^+_{AA'})$ is the normalized Choi matrix of the quantum channel $\mN$ and $\Phi^+_{AA'}$ is the normalized maximally entangled state of the $AA'$ system with $|A|=|A'|$.
\end{definition}
This definition assigns a partition function, and thereby a stabilizer work, to quantum operations themselves, opening the door to a thermodynamic accounting of the magic generated or consumed by non-Clifford gates. Establishing the monotonicity properties of these extended quantities is left for future work.

The extension to channels invites a closer look at the free operations themselves. Throughout this work, the protocol class $\mf$ discussed in Sec.~\ref{sec:SW} was specified by what its elements do to stabilizer states; since the SPF framework works at the level of the Pauli spectrum, it is natural to ask what these operations look like in the same language. Explicitly, a quantum channel $\mO\in\mf$ must satisfy the following properties:
\begin{enumerate}
    \item $\mO_{A\to B}\in \cptp(A\to B)$,
    \item $\tr[\mO(\phi_A)P] = 0\ {\rm or} \pm1$ for any $P\in\mP_B$ and for all $\phi_A\in {\rm STAB}_A$,
    \item $\sum_{P\in \mP_B} |\tr[\mO(\phi_A)P]| = d_B$ for all $\phi_A\in {\rm STAB}_A$.
\end{enumerate}
These conditions constrain the Pauli expectation values of channel outputs, and are therefore most naturally cast in terms of the Pauli transfer matrix (PTM). Recall that the PTM of $\mO\in\cptp(A\to B)$ is given by
\eqs{
\left(R^{\mO}_{AB}\right)_{i,j} =\frac{1}{d_A} \tr\left[ P^B_i \mO(P^A_j)\right]\,,
}
where $P^B_i\in \mP_B$ and $P^A_j\in \mP_A$. For $\mO\in \mf$, the conditions above become, for all $\phi\in {\rm STAB}_A$,
\begin{enumerate}
    \item $\sum_j (R^{\mO}_{AB})_{i,j}p^{\phi}_j = 0 \ {\rm or}\ \pm1$,
    \item $\sum_i \left| \sum_j (R^{\mO}_{AB})_{i,j}p^{\phi}_j \right| = d_B$,
\end{enumerate}
where $p^{\phi}_j = \tr[\phi P_j]$. Simply put, the PTM of a deterministic pure-state stabilizer protocol takes one of three forms: a signed permutation (monomial $\pm1$) matrix, with exactly one $\pm1$ per row and column, representing a Clifford unitary; a matrix whose first column is the Pauli spectrum of a pure stabilizer state with all other entries zero, representing a discard-and-prepare channel; or a matrix with exactly one $1$ in every column, representing the discarding of a subsystem. The free operations of our resource theory are thus exactly the channels whose PTMs preserve the rigid, integer-valued structure of stabilizer Pauli spectra, the same rigidity that, in the Pauli-gas picture, characterizes the free states.
This algebraic characterization gives a concrete handle on the protocol class and may prove useful for extending the monotonicity of the stabilizer work to broader classes of operations.

The PTM characterization also suggests a route to relaxations of the resource theory itself.
The current conditions carve out a sharp, discrete set with signed permutations and stabilizer-spectrum columns.
One may instead consider channels whose PTMs satisfy these conditions only approximately, 
and ask how much stabilizer work such $\epsilon$-close operations can generate.
A quantitative answer would bound the magic injected by imperfect Clifford hardware directly from measurable PTM data. This is a practically relevant direction, since the PTM is precisely the object reconstructed by gate-set tomography, so the distance of an experimental gate from the free set defined above is an accessible quantity on current devices.

Stepping back from the specific protocol class, the correspondence at the heart of this work raises a broader structural question of whether the same technique extends to other resource theories.
Given a resource theory $\mathcal{R}$ with an associated spectrum analog ${\rm Spec}_{\mathcal{R}}$, such as the Pauli spectrum for magic or the entanglement spectrum, accessible through the $\alpha$-purities $\tr_{A}(\tr_{B}\,\psi_{AB})^{\alpha}$, for entanglement, one may ask whether the corresponding free energy is always a valid resource monotone. 
A positive answer would elevate the correspondence developed here from a tool for magic to a general design principle for resource measures, with the entanglement spectrum~\cite{Li_Haldane/Entanglement_Spectrum/2008} as the most immediate test case.

Within magic itself, a complementary direction concerns the role of the SPF and stabilizer work in structured models of many-body chaotic dynamics beyond Haar-random unitaries~\cite{WES/C_ensemble/2025}. It is known that a single local many-body Hamiltonian cannot generate a Haar-random unitary at late times~\cite{Roberts/chaos_design/2017}.
We therefore expect the late-time magic of such dynamics to deviate measurably from the Haar value computed in Theorem~\ref{theorem: Pauli_T_Haar_random}, making the SPF a natural probe of how physical chaotic dynamics falls short of true randomness. 
It is also worth exploring the interplay between magic and entanglement in holographic and field-theoretic settings~\cite{White/magicalCFT/2021,  Cao/magic_gravity/2025}, where the Pauli spectrum is again the natural object, and where the ensemble-averaging tools developed here apply directly.

We close with three open problems. First, this work focused on pure states; with the mixed-state and channel SPF now defined, the outstanding task is to establish monotonicity of the corresponding stabilizer work under the appropriate free operations, for which the PTM characterization above provides a starting point. 
 Second, the stabilizer work encapsulates strictly more information than any single SRE, and it remains to be seen what this additional information reveals about physical systems where SREs have already been profitably applied, such as ground states, quench dynamics, and random circuits. 
Our Ising and molecular results are first steps in this direction. 
Third, our interconversion bounds are necessary but not sufficient, and the $H$-to-$T$ example shows that even dominance of the stabilizer-work curve at every temperature does not guarantee convertibility; whether the full profile $\beta\mapsto\mathcal{W}_{\beta}$ characterizes convertibility in an asymptotic or catalytic setting, where single-monotone obstructions are known to soften, remains an open question.

\section*{Acknowledgements}
We extend our thanks to Sean Kim of LG Electronics, Next-gen Computing Lab (Seoul), Jacob Song and Jung-Yeol Kim of LG Electronics, Emerging Tech Lab (Silicon Valley), Paria Nejat of LG Electronics Toronto AI Lab, for their administrative support. W.E.S is supported by the CQT/NQSS PhD scholarship.
\section*{AI Acknowledgement Statement}
We used Claude to refine and sharpen our Introduction and Discussion section. For instance, we learnt about the Li and Haldane's result (Ref.~\cite{Li_Haldane/Entanglement_Spectrum/2008}) where they proposed the Entanglement spectrum as a generalization of entanglement entropy, which carries a similar flavor as our work.

\bibliography{draft_ref}

\pagebreak

\onecolumngrid
\newpage

\noindent\textbf{\large Supplemental Material}\\ \\
\noindent\textbf{Stabilizer Statistical Mechanics: A Framework for Efficient Quantification and Classification of Magic States}\\ \\
 William E. Salazar$^{1,2}$,
 Gaurav Saxena$^{1}$,
 Jack S. Baker$^{1}$,
 Leong Chuan Kwek$^{2,3,4}$,
 Thi Ha Kyaw$^{1}$\\
$^{1}$LG Electronics Toronto AI Lab, Toronto, Ontario M5V 1M3, Canada\\
$^{2}$Centre for Quantum Technologies, National University of Singapore, 3 Science Drive 2, Singapore 117543, Singapore\\
$^{3}$MajuLab, CNRS-UNS-NUS-NTU International Joint Research Unit, UMI 3654, Singapore 117543, Singapore\\
$^{4}$National Institute of Education, Nanyang Technological University, 1 Nanyang Walk, Singapore 637616, Singapore
\vspace{6pt}


\supplementtableofcontents
\medskip
\supplementtoctrue

\subsection{Additional definitions and supporting lemmas}

\subsubsection{Further details about Stabilizer Hamiltonian}

In this section we provide further details about the Stabilizer Hamiltonian construction presented in the main text. To be self consistent we recall the definition here
\begin{equation}
    \hat{H}(\psi)=\sum_{s \in \pm 1}\sum_{P\in \mathcal{P}_{n}}(1-s\,\tr(P\psi))\left|P,s \right) \left(P, s \right|.
\end{equation}
with $\left|P,s \right):=(\mathbb{I}\otimes P) \ket{\Phi^{+}} \otimes \ket{s}$ the state associated to the POVM $\Pi_{P}^{s}=\frac{1}{2}(\mathbb{I} +s P)$. The first observation is that as the $\left|P,s \right)$ states are orthogonal, the stabilizer Hamiltonian decomposes as the direct sum between different $s$-sectors, namely  $\hat{H}(\psi)=\hat{H}_{+}(\psi) \oplus \hat{H}_{-}(\psi)$, where $\hat{H}_{\pm}(\psi)$ denotes the Hamiltonian restricted to the specific $\pm$-subspace, ${\rm i.e.}$ by fixing $\ket{s}$. This observation provides an easy characterization of the Stabilizer Hamiltonian ground state manifold in terms of the following lemma.

\begin{lemma}[Stabilizer Hamiltonian zero-energy state manifold]
The zero state manifold of each $\hat{H}_{s}(\psi)$-sector of the stabilizer Hamiltonian is given by the all the Pauli strings $P \in \mathcal{P}_{n}$ such that $\tr{P\psi}=s$.~\remark{While for the ($+$) sector the ground state automatically has zero energy, ${\rm e.g.}$ for $P=\mathbb{I}$. This isn't necessarily true for the ($-$) sector. The reason is that there can be states $\psi$ such that $\tr P\psi$ is strictly great than $-1$.}
\end{lemma}

The previous Lemma implies the asymptotic large-$\beta$ behavior of the SPF. Concretely, and in a direct analogy with the low temperature limit of a canonical partition function, we have
\begin{equation}
    \mathcal{Z}_{\beta}(\psi) \xrightarrow[\beta {\,\rm large} ]{} \Omega(E_{G}) e^{-\beta E_{G}}(1+\dots )
\end{equation}
with $E_{G}$ the ground state energy and $\Omega(E_{G})$ the degeneracy of the ground state manifold. For the SPF, the decomposition of the stabilizer Hamiltonian into subspaces implies 
\begin{equation}
    \mathcal{Z}_{\beta}(\psi)=\frac{1}{2}\tr e^{-\beta \hat{H}(\psi)}=\frac{1}{2}\tr_{+}e^{-\beta \hat{H}_{+}(\psi)}+\frac{1}{2}\tr_{-}e^{-\beta \hat{H}_{-}(\psi)} \xrightarrow[\beta {\,\rm large} ]{} \frac{1}{2}|{\rm STAB}(\psi)| 
\end{equation}
with $\tr_{\pm}$ the trace over each subspace.

The previous construction allows us to extent the stabilizer Hamiltonian to ensembles of states. For example, the average stabilizer Hamiltonian for Haar random states (or any 1-design)  reads
\begin{equation}
    \mathbb{E}_{\psi \sim {\rm Haar}_{n}}[ \hat{H}_{s}(\psi)]=\mathbb{I}^{\otimes 2}-s\ket{\Phi^{+}}\bra{\Phi^{+}}.
\end{equation}
We desire to emphasize that $\mathcal{Z}_{\beta}({\rm Haar_{n}}) \neq \tr  e^{-\beta \,\mathbb{E}_{\psi \sim {\rm Haar}_{n}}[ \hat{H}(\psi)]}$, and in the main text we presented the exact $\mathcal{Z}_{\beta}({\rm Haar_{n}})$. However, one can still get qualitative (coarse grained) features of the SPF from this approximation.

\subsubsection{Stabilizer R\'enyi entropies from the SPF }

The basic observation to relate the $\alpha$-stabilizer R\'enyi entropy $M_{\alpha}(\psi)$ to the SPF is that, for integer $\alpha$, and as pointed in the main text both rely on the knowledge of the Pauli-spectrum moments. And in the case of the SPF these moments can be recovered by differentiating,

\begin{proposition}[Stab R\'enyi from SPF]\label{Prop: Moment_GF}
Given the SPF $\mathcal{Z}_{\beta}(\psi)$ of the pure state $\psi$, the integer stabilizer R\'enyi entropies are obtained via.
\begin{align}
    M_{\alpha}(\psi)=\lim_{\beta \to 0}\frac{1}{1-\alpha}\log \left( \frac{1}{d_{n}} \frac{\partial^{2\alpha}}{\partial \beta^{2\alpha}}e^{\beta}\mathcal{Z}_{\beta}{(\psi)} \right)\quad \text{for} \quad \alpha \in \mathbb{N}.
\end{align}
\end{proposition}

\begin{proof}
    The SPF is the moment-generating function of the Pauli spectrum, therefore $\sum_{x\,  \in \,\text{Spec}(\psi) } x^{2m} =\frac{\partial^{2m}}{\partial \beta^{2m}}e^{\beta}\mathcal{Z}_{\beta}{(\psi)} \bigg\vert_{\beta \to 0}$. The exchange of the limits comes because the positivity of the SPF derivatives.
\end{proof}

\subsubsection{Pauli spectrum density}

As mentioned previously the Pauli spectrum defines the discrete probability distribution $\Xi(\ket{\psi})=\{ \frac{x^2}{d_{n}}, \; \text{for}\; x \in \text{Spec}(\ket{\psi})\}$. However, this isn't the only probability measure one can associate with $\text{Spec}(\ket{\psi})$. Concretely and from a Random matrix theory inspired approach, one can define a Pauli spectrum density given by the point measure

\begin{equation}
\label{eq:Pauli_spectrum_density}
    d\mu_{\ket{\psi}}(x)=\frac{1}{d^{2}_{n}}\sum_{x_{l} \,\in \,\text{Spec}(\psi)}\delta(x-x_{l})\, dx \quad \text{with} \quad x\in [-1,1].
\end{equation}
The relation between $\Xi(\ket{\psi})$ and $ d\mu_{\ket{\psi}}[x]$ is, at the moment level, given by 

\begin{equation}
    \mathbb{E}_{x\sim \Xi(\ket{\psi})}\left[x^{2\alpha}\right]=d_{n}\mathbb{E}_{x\sim d\mu_{\psi}}\left[x^{2(\alpha+1)}\right].
\end{equation}
Where $\mathbb{E}_{x\sim d\mu_{\psi}}[\cdots]=\int (\cdots) d\mu_{\psi}(x)$ denotes the average over the point measure. In terms of the Pauli-spectrum density, the SPF reads 

\begin{equation}
    \mathcal{Z}_{\beta}(\psi)=e^{-\beta}\,d_{n}^{2}\,\mathbb{E}_{x\sim d\mu_{\psi}}[\cosh{(\beta x)}].
\end{equation}
A fundamental distinction must be made between ensemble average and individual Pauli densities $d\mu_{\ket{\psi}}$. Analogously to the density of energy levels in RMT, the Pauli density for an individual state is a non-smooth sequence of $\delta$-peaks. However when the individual Pauli density is averaged over a (large enough) ensemble of states $\mathcal{E}_{\psi}$, ${\rm e.g.}$ $d\mu_{\ket{\psi}}(x) \to \mathbb{E}_{\mathcal{E}_{\psi}}[d\mu_{\ket{\psi}}(x)]$, the averaged density transforms into an smooth distribution. This is exactly what happens in the case of Haar random states.

\begin{proposition}[Pauli density for Haar random states \cite{Turkeshi/Pauli_spectrum/2025}]

The averaged Pauli density $d\mu_{{\rm Haar}_{n}}(x)=\mathbb{E}_{\ket{\psi} \sim {\rm Haar}_{n}}[d\mu_{\ket{\psi}}(x)]$ for Haar random states is given by,

\begin{equation}
    d\mu_{{\rm Haar}_{n}}[x]=\frac{1}{d_{n}^{2}}\delta(x-1)+\frac{d_{n}^{2}-1}{d_{n}^{2}}\frac{(1-x^{2})^{\frac{d_{n}}{2}-1}\,\Gamma\left(\frac{d_{n}+1}{2}\right)}{\sqrt{\pi}\,\Gamma\left(\frac{d_{n}}{2}\right)},
\end{equation}
with $\Gamma(x)$ the Gamma function.
\end{proposition}

Clearly both $\mathcal{Z}_{\beta}(\psi)$ and $d\mu_{\ket{\psi}}[x]$ have the same information content regarding the statistics of the Pauli spectrum, and given the knowledge of one, the other can be determined via an inverse transform.
However, working directly with the SPF provides two crucial advantages.
First, in contrast to $d\mu_{\ket{\psi}}$, the SPF is smooth for all states, and not only for ensembles of states. Second, from a pure operational perspective, $\mathcal{Z}_{\beta}(\psi)$ has simple expressions for different classes of states, converting the Pauli spectrum into a classifier. Even more, $\mathcal{Z}_{\beta}(\psi)$ as argued in the main text, the SPF admits a clean interpretation from the resource theory of magic point of view. 

\subsubsection{Monotone Reweighting lemma}
\begin{lemma}\label{lem:Chebyshev_lemma}
    If $a_0\geq a_1\geq\cdots\geq 0$ and two probability distributions $p,q$ are such that $p_k/q_k$ is non-increasing in $k$, then $\sum_kp_ka_k\geq \sum_k q_k a_k$.
\end{lemma}
\begin{proof}
    Since $\sum_k p_k = 1 = \sum_k q_k$ and the difference $p_k-q_k$ is positive for small $k$ and negative for large $k$ (one sign change, at some $k_0$. Using $\sum_k(p_k - q_k) = 0$, we get
    \eqs{
    \sum_k (p_k- q_k)a_k = \sum_k (p_k-q_k)(a_k-a_{k_0})\geq 0
    }
    as $p_k-q_k$ has the same sign as $a_k-a_{k_0}$.
\end{proof}

\begin{theorem}\label{sub-additivity_supporting_thm}
    Let $a=(a_0,a_1,a_2,\cdots)$ and $b=(b_0,b_1,b_2,\cdots)$ be two non-increasing vectors such that $a_0 = d_A^2, b_0 = d_B^2, a_1 = d_A, b_1 = d_B, a_{k>1}\leq a_1, b_{k>1}\leq b_1$.
    Let $w = (1, \frac{x^2}{2!}, \frac{x^4}{4!},\cdots)$ be a vector such that $\sum_k w_k = c$ where $c = \cosh(x)$ and $x\in \mbR$. Then it holds that
    \eqs{
    \frac{\sum_k w_k a_k b_k}{(\sum_k w_k a_k)(\sum_k w_k b_k)}\geq \frac{(d_Ad_B+c-1)}{(d_A+c-1)(d_B+c-1)}
    }
\end{theorem}
\begin{proof}
    Let $N(a,b) = \sum_k w_k a_k b_k$ and let $D_a = \sum_k w_k a_k$ and $D_b = \sum_k w_k b_k$.
    Also, let $p_k^a = \frac{w_ka_k}{D_a}$ be a probability distribution. Then
    \eqs{
    \frac{N(a,b)}{D_a} = \sum_k p^a_k b_k\,.
    }
    and \eqs{
    \frac{N(a,b)}{D_aD_b} = \frac{1}{D_b}\sum_k p^a_k b_k
    }
    If we define $a' = (d^2_A, d_A, d_A, d_A, \cdots)$, then $a_k/a'_k$ is non-increasing in $k$, and we have from Lemma~\ref{Chebyshev_lemma} that $\sum_k p^a_kb_k\geq \sum_k p^{a'}_k b_k$.
    So we get
    \eqs{
    \frac{N(a,b)}{D_a D_b}\geq \frac{N(a',b)}{D_a' D_b}\,.
    }
    Following similar steps for $b' = (d^2_B, d_B, d_B,\cdots)$, we get
    \eqs{
    \frac{N(a,b)}{D_a D_b}\geq \frac{N(a',b')}{D_a' D_b'} = \frac{(d_Ad_B+c-1)}{(d_A+c-1)(d_B+c-1)}
    }
    completing the proof.
\end{proof}
If we use SPF ($\mZ_{\beta}$) instead of the core SPF ($\mZ_{\beta}^c$) to define a function similar to stabilizer work, we get a function which has some properties in common with the stabilizer work. Let us call this function as \textit{unfiltered stabilizer work} and denote it by $\mathfrak{u}_{\beta}(\psi)$. The common properties include faithfulness, Clifford-invariance, and subadditivity. The proofs of faithfulness and Clifford-invariance are straightforward, and we provide the proof of subadditivity below. However, it turns out that the unfiltered stabilizer work has some desirable properties missing such as the unfiltered stabilizer work does not always decrease under partial trace for all values of $\beta$. 
\begin{theorem}[Sub-additivity of unfiltered stabilizer work]
    Unfiltered work is sub-additive, i.e.,
    \eqs{
    \mathfrak{u}_{\beta}(\psi_A \otimes \phi_B)\leq  \mathfrak{u}_{\beta}(\psi_{A}) + \mathfrak{u}_{\beta}(\phi_{B})
    }
\end{theorem}
\begin{proof}
    To prove the sub-additivity, we need to prove 
    \eqs{
    \frac{\mZ_{\beta}(\psi_A \otimes \phi_B)}{\mZ_{\beta}(\psi_A)\mZ_{\beta}(\phi_B)}  \geq \frac{\mZ_{\beta}({\rm Stab}_{AB})}{\mZ_{\beta}({\rm Stab}_{A})\mZ_{\beta}({\rm Stab}_{B})} 
    }
    which can be rewritten as 
    \eqs{
    \frac{\sum_{i,j}\cosh(\beta \tr[\psi_A P_i]\tr[\phi_B P_j])}{\mZ_{\beta}(\psi_A)\mZ_{\beta}(\phi_B)}\geq \frac{\mZ_{\beta}({\rm Stab}_{AB})}{\mZ_{\beta}({\rm Stab}_{A})\mZ_{\beta}({\rm Stab}_{B})} \,.
    }
    When we expand the hyperbolic cosine terms in the above equation, it can further be simplified to
    \eqs{
    \frac{\sum_{k\geq 0}w_k a_k b_k}{(\sum_k w_k a_k)(\sum_k w_k b_k)} \nonumber    \geq \frac{d_Ad_B-1+d_Ad_B \cosh(\beta)}{(d_A-1+d_A\cosh(\beta))(d_B-1+d_B\cosh(\beta))} \nonumber 
    }
    where $w_k = \frac{\beta^{2k}}{(2k)!}, a_k = \sum_i\tr[\psi_A P_i]^{2k}, b_k = \sum_j\tr[\phi_B P_j]^{2k}$.
    The proof of this inequality is precisely the result of Theorem~\ref{sub-additivity_supporting_thm}, thus completing the proof of the sub-additivity of the unfiltered stabilizer work.

\end{proof}

\subsubsection{Small Gap (Lemma~\ref{Lemma:Small_Gaps})}

By Markov inequality with $\exp(-\mathcal{W}_{\beta}(\psi))$ we have
\begin{equation}
    {\rm Pr}_{\psi \sim \mathcal{E}}(\exp(-\mathcal{W}_{\beta}(\psi)) \geq \epsilon) \leq \frac{\mathcal{Z}_{\beta}^{c}(\mathcal{E})}{\mathcal{Z}^{c}_{\beta}({\rm Stab}_{n})\epsilon}
\end{equation}
with $\mathcal{Z}_{\beta}^{c}(\mathcal{E})$ the ensemble averaged SPF. By choosing $\epsilon = \frac{\mathcal{Z}^{c}_{\beta}(\mathcal{E})}{\mathcal{Z}^{c}_{\beta}({\rm Stab}_{n})}{\rm poly(n)}$ and using the convexity of $-log$ we obtain the lower bound, namely
\begin{align}
    {\rm Pr}_{\psi \sim \mathcal{E}}(\mathcal{W}_{\beta}(\psi) \geq \mathcal{W}_{\beta}(\mathcal{E}) - \log {\rm Poly}(n)) \geq 1-\frac{1}{{\rm poly(n)}}.
\end{align}
For the upper bound, we use Markov inequality with $\exp(\mathcal{W}_{\beta}(\psi))$ instead, \emph{i.e.},
\begin{equation}
    {\rm Pr}_{\psi \sim \mathcal{E}}(\exp(\mathcal{W}_{\beta}(\psi)) \geq \epsilon) \leq \frac{\mathcal{Z}^{c}_{\beta}({\rm Stab}_{n})}{\epsilon} \mathbb{E}_{\psi\sim \mathcal{E}}[\mathcal{Z}_{\beta}(\psi)^{-1}].
\end{equation}
Where in order to bound the reciprocal average we recall the following result:
\begin{theorem}[Kantorovich inequality \cite{Ptk/Kantorovich/1995}]
Let $x_{1}\leq x_{2} \leq \dots, \leq x_{n}$ a list of positive reals. And let $\{p_{l} \geq 0\}_{l=1}^{n}$, with $\sum_{l=1}^{n}p_{l}=1$, then 
\begin{equation}   \left(\sum_{l=1}^{n}p_{l}\,x_{l}\right)\left(\sum_{l=1}^{n}p_{l}\,x^{-1}_{l}\right) \leq \left(\frac{{\rm AM}}{{\rm GM}}\right)^{2}
\end{equation}
with ${\rm AM}=\frac{1}{2}(x_{1}+x_{n})$ the arithmetic and ${\rm GM}=\sqrt{x_{1}x_{n}}$ geometric means.
\end{theorem}
By using Kantorovich inequality with the SPF upper/lower bounds from Theorem.~\ref{theorem:SPF_state_independent_bounds} and setting $\epsilon=\frac{\mathcal{Z}_{\beta}^{c}(\rm Stab_{n})}{\mathcal{Z}_{\beta}^{c}(\mathcal{E})} {\rm Poly}(n) \left(\frac{{\rm AM}}{{\rm GM}}\right)^{2}$ , we proof the claim for the upper bound, namely
\begin{align}
    {\rm Pr}_{\psi \sim \mathcal{E}}(\mathcal{W}_{\beta}(\psi) \leq \mathcal{W}_{\beta}(\mathcal{E})+{\rm Small}_{\beta}+\log {\rm Poly}(n)) \geq 1-\frac{1}{\rm Poly(n)}.
\end{align}
with,
\begin{equation}
    {\rm Small}_{\beta}=2\log ({\rm AM/GM})= 2\log\left(\frac{\mathcal{Z}^{c}_{\beta}({\rm Stab}_{n})+\mathcal{Z}_{\beta}^{c}({\rm Min})}{2\sqrt{\mathcal{Z}^{c}_{\beta}({\rm Stab}_{n})\mathcal{Z}^{c}_{\beta}({\rm Min})} }\right).
\end{equation}
Thus, ${\rm Small}_{\beta}$ quantifies the gap induced by the SPF bounds, with $\mathcal{Z}_{\beta}^{c}({\rm Min})$ denoting the SPF lower bound in
Theorem~\ref{theorem:SPF_state_independent_bounds}. The following proposition provides an upper bound on this difference.

\begin{proposition}[Small Poly bound]
The following ${\rm Small}_{\beta}\leq\mathcal{O}(\beta^{4})$ holds.
\end{proposition}

\begin{proof}
By the AM/GM inequality, ${\rm Small}_{\beta}\geq0$. From Theorem~\ref{theorem:SPF_state_independent_bounds}, the upper and lower bounds for the core SPF are, respectively,
\begin{align}
\mathcal{Z}_{\beta}^{c}({\rm Stab}_{n})
&=
e^{-\beta}d_{n}\bigl(\cosh(\beta)-1\bigr),\\
\mathcal{Z}_{\beta}^{c}({\rm Min})
&=
e^{-\beta}\left[
\cosh(\beta)-1
+(d_{n}^{2}-1)
\left(
\cosh\left(\frac{\beta}{\sqrt{d_{n}+1}}\right)-1
\right)
\right].
\end{align}
Using $\cosh(x)-1\geq x^{2}/2$, we obtain $\mathcal{Z}_{\beta}^{c}({\rm Min})\geq e^{-\beta}d_{n}\frac{\beta^{2}}{2}$. Then,
\begin{align}
1 \leq \frac{ \mathcal{Z}_{\beta}^{c}({\rm Stab}_{n}) }{ \mathcal{Z}_{\beta}^{c}({\rm Min}) } &\leq
\frac{2(\cosh(\beta)-1)}{\beta^{2}}= \left( \frac{\sinh(\beta/2)}{\beta/2} \right)^{2}.
\end{align}
Furthermore, using $\log\left(\frac{\sinh x}{x}\right)\leq\frac{x^{2}}{6}$, gives
\begin{equation}
0 \leq \log\left( \frac{ \mathcal{Z}_{\beta}^{c}({\rm Stab}_{n})}{
\mathcal{Z}_{\beta}^{c}({\rm Min})}\right)\leq \frac{\beta^{2}}{12}.
\end{equation}
Finally, writing the AM/GM ratio in its equivalent hyperbolic form and
using $\log(\cosh x)\leq x^{2}/2$, we have
\begin{align}
{\rm Small}_{\beta}
&=
2\log\cosh\left[
\frac{1}{2}
\log\left(
\frac{
\mathcal{Z}_{\beta}^{c}({\rm Stab}_{n})
}{
\mathcal{Z}_{\beta}^{c}({\rm Min})
}
\right)
\right] \leq \frac{1}{4} \left[ \log\left( \frac{
\mathcal{Z}_{\beta}^{c}({\rm Stab}_{n}) }{
\mathcal{Z}_{\beta}^{c}({\rm Min}) } \right) \right]^{2}
\leq \frac{\beta^{4}}{576}.
\end{align}
\end{proof}
Finally, in order to finish the proof we use the union bound on the Upper and lower bound events. 

\subsection{Proofs}

\subsubsection{Lipschitz bound (Theorem~\ref{Theorem:Lipshitzs_property})}

In this section, we present two complementary proofs of the Lipschitz property of the SPF. The first proof uses the Clifford-orbit representation
and gives a Lipschitz constant with a quadratic overhead in the Hilbert-space dimension. We then give a complementary proof, directly from the Pauli
representation of the SPF, which refines this overhead to a linear one and therefore establishes the bound stated in the main text.

\begin{itemize}
    \item By using the Clifford action representation \eqref{eq:Clifford_representation_P_transform}, the norm of the difference between the Stabilizer partition function of the states $\ket{\psi}$, $\ket{\phi}$ is upper bounded as
\eqs{
    |\mathcal{Z}_{\beta}(\psi)-\mathcal{Z}_{\beta}({\phi})| &\leq \frac{(4^{n}-1)}{e^{\beta}}\sum_{l=0}^{\infty}\frac{\beta^{2l}}{2l!}\left\|(P^{\otimes 2l} \Phi_{\mathcal{C}l_{n}}^{2l}\left(\psi^{\otimes 2l}-\phi^{\otimes 2l}\right)\right\|_{1},\\
    &\leq \frac{(4^{n}-1)}{e^{\beta}}\sum_{l=0}^{\infty}\frac{\beta^{2l}}{2l!}\left\|\left(\psi^{\otimes 2l}-\phi^{\otimes 2l}\right)\right\|_{1} \label{eq:ub_using_isometry_invariance},
}
where for the first inequality, we have used the triangle inequality  to set the upper bound in terms of the trace norm over the $2l$-th fold Clifford channel action.
Then using the fact that the trace norm is invariant under isometries, and using the data-processing inequality, we get Eq.~\eqref{eq:ub_using_isometry_invariance}.
Lastly, by using the telescoping property of the trace norm, i.e., $\|\psi^{\otimes 2l}-\phi^{\otimes 2l}\|_{1}\leq 2l ||\psi -\phi||_{1}$, one arrives (after a straightforward series summation) to the result presented in the main text, i.e.,
\begin{equation}
    |\mathcal{Z}_{\beta}(\psi)-\mathcal{Z}_{\beta}(\phi)|  \leq \frac{(4^{n}-1)}{e^{\beta}}\,\beta\,\sinh{\left(\beta\right)}\,||\psi-\phi||_{1}.
\end{equation}

\item  The previous proof induces a \(d_n^2\) overhead in the Lipschitz constant. We now present a complementary proof that reduces this quadratic
overhead to a linear one. From the Pauli
representation of the SPF,
\begin{align}
    |\mathcal{Z}_{\beta}(\psi)-\mathcal{Z}_{\beta}(\phi)|
    &\leq e^{-\beta}\sum_{P\in\mathcal{P}_n}
    \left|\cosh(\beta \tr(\psi P))-\cosh(\beta \tr(\phi P))\right|
    \nonumber\\
    &\leq e^{-\beta}\sum_{P\in\mathcal{P}_n}
    \sum_{l\geq 1}\frac{\beta^{2l}}{(2l)!}
    \left|\tr(\psi P)^{2l}-\tr(\phi P)^{2l}\right|.
    \label{eq:spf_lipschitz_pauli_expansion}
\end{align}
Using the telescoping property for every \(l\geq 1\),
\begin{align}
    \left|\tr(\psi P)^{2l}-\tr(\phi P)^{2l}\right|
    &\leq l\left|\tr(\psi P)^2-\tr(\phi P)^2\right|
    \nonumber\\
    &\leq l|\tr(\psi P)-\tr(\phi P)|(|\tr(\psi P)|+|\tr(\phi P)|).
    \label{eq:spf_even_power_difference}
\end{align}
Furthermore, by Cauchy-Schwarz,
\begin{align}
    &\sum_{P\in\mathcal{P}_n}
    |\tr(\psi P)-\tr(\phi P)|(|\tr(\psi P)|+|\tr(\phi P)|)
    \nonumber\\
    &\quad\leq
    \left(\sum_{P\in\mathcal{P}_n}|\tr(\psi P)-\tr(\phi P)|^2\right)^{1/2}
    \left(\sum_{P\in\mathcal{P}_n}(|\tr(\psi P)|+|\tr(\phi P)|)^2\right)^{1/2}.
    \label{eq:spf_pauli_cauchy}
\end{align}
Where the first term is bounded by Parseval's identity, namely
\begin{equation}
    \sum_{P\in\mathcal{P}_n}|\tr(\psi P)-\tr(\phi P)|^2
    =
    d_n\|\psi-\phi\|_2^2
    \leq
    d_n\|\psi-\phi\|_1^2,
    \label{eq:spf_pauli_difference_parseval}
\end{equation}
and the second term by using $(a+b)^2\leq2(a^2+b^2)$, and Parseval once again, \emph{i.e.},
\begin{align}
    \sum_{P\in\mathcal{P}_n}(|\tr(\psi P)|+|\tr(\phi P)|)^2
    &\leq
    2\sum_{P\in\mathcal{P}_n}(\tr(\psi P)^2+\tr(\phi P)^2)
    \nonumber\\
    &=
    2d_n\left(\tr(\psi^2)+\tr(\phi^2)\right)
    =
    4d_n.
    \label{eq:spf_pauli_sum_parseval}
\end{align}
Finally, combining \eqref{eq:spf_pauli_cauchy}, \eqref{eq:spf_pauli_sum_parseval}, and after a straightforward summation yields the result stated in the main text, namely
\begin{equation}
    |\mathcal{Z}_{\beta}(\psi)-\mathcal{Z}_{\beta}(\phi)|
    \leq
    d_ne^{-\beta}\beta\sinh(\beta)\|\psi-\phi\|_1.
\end{equation}
Therefore, the complementary proof reduces the dimensional overhead in the Lipschitz constant from $d_{n}^2$ to $d_{n}$.

\end{itemize}

\subsubsection{SPF sampling problem (Theorem \ref{theorem:SPF_Complexity})}

In this section we construct an unbiased estimator $\hat{\mathcal{Z}}_{\beta}^{(k)}(\psi)$ for the $k$th truncation of the stabilizer partition function (SPF). The proof uses the Bell-sampling algorithm presented in \cite{Tobias/Pauli_sampling/2024}, and shows that the truncated SPF can be estimated with the same primitive and comparable asymptotic cost as the stabilizer R\'enyi entropies. We then use the truncation remainder to control the approximation to the full SPF.

Let $\mathcal{Z}^{(k)}_{\beta}(\psi)$ denote the $2k$th polynomial truncation of the SPF in the Pauli moments, namely
\begin{equation}
    \mathcal{Z}^{(k)}_{\beta}(\psi)=e^{-\beta}\sum_{x \in \text{Spec}(\psi)}\sum_{l = 0}^{k} \frac{(\beta \,x)^{2l}}{(2l)!}
\end{equation}
As pointed out in the main letter, the complexity of computing the SPF is directly related to the complexity of obtaining the Pauli-spectrum moments. At first sight, one may expect that computing the $k$th truncation requires running the algorithm in \cite{Tobias/Pauli_sampling/2024} separately for each moment, leading to a classical runtime of $\mathcal{O}(k^{2}n)$. Here we show that a single run with additional classical memory is enough to estimate all even moments up to order $2k$, which reduces the runtime to $\mathcal{O}(kn)$. The argument has two steps: first we construct an estimator for $\mathcal{Z}^{(k)}_{\beta}(\psi)$, and then we control the truncation error $\mathcal{Z}_{\beta}(\psi)-\mathcal{Z}^{(k)}_{\beta}(\psi)$.

To make the bookkeeping explicit, we separate the constant contribution and encode the non-trivial even moments in the vector
\begin{equation}
 \mathbf{m}_{k}(\psi)=\left(\frac{1}{d_{n}}\sum_{x \in \text{Spec}(\psi)}x^{2}, \dots, \frac{1}{d_{n}}\sum_{x \in \text{Spec}(\psi)}x^{2k} \right),
\end{equation}
and introduce the coefficient vector
\begin{equation}
    \mathbf{c}^{(k)}(\beta)=\left(\frac{\beta^{2}}{2!}, \dots, \frac{\beta^{2k}}{(2k)!}\right).
\end{equation}
With this notation the $k$th truncation is reconstructed as
\begin{equation}   
\label{appendix_eqn:k_SPF_truncation}
{\mathcal{Z}}^{(k)}_{\beta}(\psi)=e^{-\beta}\left(d_{n}^{2}+d_{n}\,\mathbf{c}^{(k)}(\beta) \cdot \mathbf{m}_{k}(\psi)\right),
\end{equation}
where the constant term has been written separately. Therefore, given an estimator $\hat{\mathbf{m}}_{k}(\psi)$ for $\mathbf{m}_{k}(\psi)$, we define
\begin{equation}
\hat{\mathcal{Z}}^{(k)}_{\beta}(\psi)=e^{-\beta}\left(d_{n}^{2}+d_{n}\,\mathbf{c}^{(k)}(\beta) \cdot \hat{\mathbf{m}}_{k}(\psi)\right).
\end{equation}
The key point is that the algorithm in \cite{Tobias/Pauli_sampling/2024} can be implemented with additional classical memory so as to estimate $\hat{\mathbf{m}}_{k}(\psi)$ in a single run. The details are given in the following algorithm.

\begin{algorithm}[h!]
\label{algo:Pauli_sampling}
\caption{Pauli Moments \cite{Tobias/Pauli_sampling/2024} plus classical memory}
\begin{algorithmic}[1]

\Require Integer $k > 1$;
$\mathcal{N}$ (Samples)
\Ensure $\mathbf{m}_{k}$ \Comment{Vector with the even Pauli-spectrum moments up to order $2k$}

\State $\mathbf{m}_{k}\gets [0,\dots,0]$

\For{$ \alpha = 0, \dots, \mathcal{N}-1$}
    \State Prepare $\lvert \eta \rangle = U_{\text{Bell}}^{\otimes n} \lvert \psi^{*} \rangle \otimes   \lvert \psi\rangle$
    \State Sample $P \sim \lvert 
    \langle P \mid \eta \rangle \rvert^2$ 
    \State $b \gets 1$
    \State $\mathbf{B}_{k} \gets [0,\dots,0]$
    \For{$\ell = 1, \dots, 2k$}
        \State Prepare $\lvert \psi \rangle$ and measure in eigenbasis of 
        Pauli String $P$ for eigenvalue $\lambda \in \{+1,-1\}$
        \State $b \gets b \cdot \lambda$
        \If{$\ell=$ Even }
        \State $\mathbf{B}^{\frac{\ell}{2}-1}_{k} \gets b$ \Comment{Keep only even moments}
        \EndIf
    \EndFor
    \State $\mathbf{m}_{k}\gets \mathbf{m}_{k} + \frac{1}{\mathcal{N}}\mathbf{B}_{k}$
\EndFor \\
\Return $\mathbf{m}_{k}$
\end{algorithmic}
\end{algorithm}

The Bell-sampling primitive itself is independent of $\beta$: once $\hat{\mathbf{m}}_{k}(\psi)$ has been learned, the truncation for any fixed $\beta$ is obtained by contracting with $\mathbf{c}^{(k)}(\beta)$.

\begin{lemma}[Estimator bounds]\label{app_lemma:k_order_stimator}
The algorithm in \ref{algo:Pauli_sampling} learns the normalized $k$th-order truncation $d_n^{-1}\mathcal{Z}^{(k)}_{\beta}(\psi)$ with failure probability $\delta$ and additive error $\epsilon$ in $\mathcal{N}=\mathcal{O}\left(\frac{\cosh^{2}{\beta}}{e^{2\beta}\epsilon^{2}}\log\left(\frac{1}{\delta}\right)\right)$ samples.
\end{lemma}

\begin{proof}
For each outer-loop sample $\alpha$, let $X_{\alpha}:=\mathbf{c}^{(k)}(\beta)\cdot\mathbf{B}_{k}^{(\alpha)}$, where $\mathbf{B}_{k}^{(\alpha)}$ denotes the sample contribution accumulated during that run of the algorithm. Since each entry of $\mathbf{B}_{k}^{(\alpha)}$ is $\pm 1$, we have
\begin{equation}
    |X_{\alpha}|\leq \sum_{l=1}^{k}\frac{\beta^{2l}}{(2l)!}\leq \sum_{l=0}^{\infty}\frac{\beta^{2l}}{(2l)!}=\cosh{\beta}.
\end{equation}
Therefore Hoeffding's inequality gives
\begin{equation}
    {\rm Pr}(|\mathbf{c}^{(k)}(\beta)\cdot\hat{\mathbf{m}}_{k}- \mathbf{c}^{(k)}(\beta)\cdot\mathbf{m}_{k}| \geq \epsilon) \leq 2\exp\left(-\frac{\mathcal{N} \epsilon^{2}}{2\cosh^{2}{\beta}}\right).
\end{equation}
Using \eqref{appendix_eqn:k_SPF_truncation}, the normalized estimator satisfies
\begin{equation}
    \label{app_eqn:k_th_concentration}
    {\rm Pr}\left(\frac{\left|\hat{\mathcal{Z}}^{(k)}_{\beta}(\psi)- \mathcal{Z}^{(k)}_{\beta}(\psi)\right|}{d_{n}} \geq \epsilon\right) \leq 2\exp\left(-\frac{\mathcal{N} \epsilon^{2} e^{2\beta}}{2\cosh^{2}{\beta}}\right).
\end{equation}
Then, by setting the failure probability at most $\delta=2\exp\left(-\frac{\mathcal{N} \epsilon^{2} e^{2\beta}}{2\cosh^{2}{\beta}}\right)$ in \eqref{app_eqn:k_th_concentration}, we get the desired sample complexity.

\end{proof}

The previous lemma shows that the truncated SPF can be estimated efficiently for fixed $k$ and $\beta$. We now show how this estimator approximates the full SPF

\begin{theorem}[Estimator Concentration]
Let
\begin{equation}
    \Delta(\beta,k)=d_{n}e^{-\beta}\left(\cosh{(\beta)}-\sum_{l=0}^{k} \frac{\beta^{2l}}{(2l)!}\right).
\end{equation}
Then, for every $\epsilon>\Delta(\beta,k)$, the estimator $\hat{\mathcal{Z}}^{(k)}_{\beta}(\psi)$ approximates the normalized SPF $d_n^{-1}\mathcal{Z}_{\beta}(\psi)$ with error bound

\begin{equation}
    {\rm Pr}\left(\frac{\mathcal{Z}_{\beta}(\psi)-\hat{\mathcal{Z}}^{(k)}_{\beta}(\psi)}{d_n} \geq \epsilon\right) \leq 2\exp\left(-\frac{\mathcal{N} (\epsilon-\Delta(\beta,k))^{2}e^{2\beta}}{2\cosh^{2}{\beta}}\right),
\end{equation}

and therefore requires
\begin{equation}
    \mathcal{N}=\mathcal{O}\left(\frac{\cosh^{2}{\beta}}{e^{2\beta}(\epsilon-\Delta(\beta,k))^{2}}\log\left(\frac{1}{\delta}\right)\right)
\end{equation}
samples.

\end{theorem}

\begin{proof}
    Let $\Delta \mathcal{Z}_{\beta}^{(k)}(\psi)=\mathcal{Z}_{\beta}(\psi)-\mathcal{Z}^{(k)}_{\beta}(\psi)\geq 0$ denote the truncation remainder. Since $|x|\leq 1$ for every $x \in \text{Spec}(\psi)$, we obtain
    \begin{equation}
        \frac{\Delta \mathcal{Z}_{\beta}^{(k)}(\psi)}{d_n}= \frac{e^{-\beta}}{d_n}\sum_{x \in \text{Spec}(\psi)}\left(\cosh{(\beta x)}-\sum_{l=0}^{k}\frac{(\beta x)^{2l}}{(2l)!}\right)\leq \Delta(\beta,k).
    \end{equation}
    Moreover, as $\mathcal{Z}_{\beta}(\psi)-\hat{\mathcal{Z}}^{(k)}_{\beta}(\psi)=\mathcal{Z}^{(k)}_{\beta}(\psi)-\hat{\mathcal{Z}}^{(k)}_{\beta}(\psi)+\Delta \mathcal{Z}_{\beta}^{(k)}(\psi)$, then if $ \mathcal{Z}_{\beta}(\psi)-\hat{\mathcal{Z}}^{(k)}_{\beta}(\psi)\geq d_{n} \epsilon $
    \begin{equation}
     \mathcal{Z}^{(k)}_{\beta}(\psi)-\hat{\mathcal{Z}}^{(k)}_{\beta}(\psi) \geq (\epsilon-\Delta(\beta,k))d_{n}.
    \end{equation}
    Therefore, by applying Lemma~\ref{app_lemma:k_order_stimator} with error threshold $\epsilon-\Delta(\beta,k)>0$ gives the stated concentration bound and sample complexity.
\end{proof}

\subsubsection{Variational lower bounds (Theorem~\ref{theorem:SPF_state_dependent_bounds}) }

Now we present a refinement to the lower bound based on the Chernoff bound.

\begin{lemma}[variational SPF lower bound]\label{lemma:variatonal_SPF_lower_bound}
Given $\beta \geq 0$, and a control parameter $\epsilon\in[0,1)$, the SPF
of an $n$-qubit pure state $\psi$ follows the family of lower bounds
\begin{equation}
    \mathcal{Z}_{\beta}(\psi) \geq e^{-\beta}d_n^2
    \left[
        \cosh(|\beta|\epsilon)\left(1-f_{\epsilon}(\psi)\right)
        +
        f_{\epsilon}(\psi)
    \right],
\end{equation}
with $ f_{\epsilon}(\psi) = \frac{1-2^{-M_2(\psi)-n}}{1-\epsilon^4}$, and $M_2(\psi)$ the second stabilizer R\'enyi entropy. Moreover, the tightest of such bounds is obtained by optimizing over $\epsilon$, namely 
\begin{equation}
    \mathcal{Z}_{\beta}(\psi)
    \geq
    \sup_{\epsilon \in [0,1)}  e^{-\beta}d_n^2
    \left[  \cosh(|\beta|\epsilon)\left(1-f_{\epsilon}(\psi)\right)  +f_{\epsilon}(\psi) \right].
\end{equation}

\end{lemma}

\begin{proof}
Let $x\sim d\mu_{\ket{\psi}}$ be distributed according to the Pauli-spectrum density in \eqref{eq:Pauli_spectrum_density}. In terms of this measure, the SPF reads 
\begin{equation}
    \mathcal{Z}_{\beta}(\psi)= e^{-\beta}d_n^2\, \mathbb{E}_{x\sim d\mu_{\ket{\psi}}} \left[\cosh(\beta x)\right].
\end{equation}
Now, the trick is to split the Pauli spectrum into the two regions $|x|\geq\epsilon$ and $|x|<\epsilon$. Since the hyperbolic cosine is monotonic increasing on $[0,\infty)$, we have the lower bound
\begin{equation}
    \cosh(\beta x) = \cosh(\beta|x|) \geq \cosh(|\beta|\epsilon)\mathbf{1}_{\{|x|\geq\epsilon\}}+ \mathbf{1}_{\{|x|<\epsilon\}}.
\end{equation}
with $\mathbf{1}_{\mathcal{R}}$ the indicator function in the region $\mathcal{R}$. Therefore, by denoting $p_{\epsilon}(\psi) := {\rm Prob}_{x\sim d\mu_{\ket{\psi}}}(|x|<\epsilon)$, we have the lower bound
\begin{align}
    \mathbb{E}_{x\sim d\mu_{\ket{\psi}}}
    \left[\cosh(\beta x)\right] &\geq \cosh(|\beta|\epsilon) {\rm Prob}(|x|\geq\epsilon) +  {\rm Prob}(|x|<\epsilon)\nonumber \\
    &= \cosh(|\beta|\epsilon) \left(1-p_{\epsilon}(\psi)\right)  +  p_{\epsilon}(\psi).
\end{align}
Moreover, by Markov's inequality,
\begin{align}
    p_{\epsilon}(\psi) \leq
    {\rm Prob}(|x|\leq\epsilon) \nonumber =
    {\rm Prob}(1-x^4\geq1-\epsilon^4) \nonumber \leq
    \frac{1- 2^{-M_{2}(\psi)-n}}{1-\epsilon^4 }=
    f_{\epsilon}(\psi).
\end{align}
Where we further used $\mathbb{E}_{x\sim d\mu_{\ket{\psi}}}[x^4]=2^{-M_{2}(\psi)-n}$. Finally, since $\cosh(\beta\epsilon)\geq 1$, the function $p\mapsto \cosh(\beta\epsilon)(1-p)+p$ is decreasing in $p$. Then, by replacing $p_{\epsilon}(\psi)$ by its upper bound $f_{\epsilon}(\psi)$ preserves the lower bound, and we arrive at the presented result
\begin{equation}
    \mathcal{Z}_{\beta}(\psi) \geq e^{-\beta}d_n^2
    \left[ \cosh(|\beta|\epsilon)\left(1-f_{\epsilon}(\psi)\right) + f_{\epsilon}(\psi) \right].
\end{equation}
\end{proof}

The this variational form of the lower bound has a similar interpretation as the one presented in the main text. The idea is that the SPF is controlled a limiting case pauli spectrum.

\subsubsection{SPF for product states (Theorem~\ref{theorem: separable_states_P_transform})}

The main idea behind the proof is to split the Pauli spectrum 
as the union between a stabilizer contribution, containing all the $\pm 1$ in the spectrum, and its complement. 

\begin{proposition}[Canonical Splitting]
    Given a state $\ket{\psi}$, its Pauli spectrum $\text{Spec}(\ket{\psi})$ can be always brought into the form $\text{Spec}(\ket{\psi})=\text{Spec}_{S}(\ket{\psi})\, \cup \, \text{Spec}_{R}(\ket{\psi})$. With $\text{Spec}_{S}(\ket{\psi}) = \{\tr(\psi P), \,\, \text{for all } \,\, P\ket{\psi}=\pm \ket{\psi}\}$ the stabilizer part of the spectrum, and $\text{Spec}_{R}(\ket{\psi}) = \text{Spec}(\ket{\psi})\backslash \text{Spec}_{S}(\ket{\psi})$ the residual contribution to the spectrum containing all Pauli's such that $|\tr{P\psi}|$ is strictly less than one. 
\end{proposition}
By using the previous Proposition, the Pauli spectrum of the separable state $\ket{\psi_{AB}}=\ket{\psi_{A}}\otimes \ket{\psi_{B}}$ have the following two equivalent splittings,
\[   
\text{Spec}(\ket{\psi_{AB}}) = 
     \begin{cases}
        &i)\quad\text{Spec}_{S} (\ket{\psi_{A}})\times \text{Spec}(\ket{\psi_{B}}) \cup\, \text{Spec}_{R} (\ket{\psi_{A}})\times \text{Spec}(\ket{\psi_{B}})\\
       &ii) \quad\text{Spec} (\ket{\psi_{A}})\times \text{Spec}_{S}(\ket{\psi_{B}}) \cup\, \text{Spec} (\ket{\psi_{A}})\times \text{Spec}_{R}(\ket{\psi_{B}})  \\ 
     \end{cases}
\]
Consequently, the SPF $\mathcal{Z}_{\beta}(\psi_{AB})$ for the joint state can be expressed as the average SPF between the two equivalent splittings $i)$ and $ii)$, ${\rm e.g.}$

\begin{align}
 \mathcal{Z}_{\beta}(\psi_{AB}) &=\frac{1}{2}\left[2^{|A|-\nu_{A}}\mathcal{Z}_{\beta}(\psi_{B})\,+\,2^{|B|-\nu_{B}}\mathcal{Z}_{\beta}(\psi_{A})\,+\, e^{-\beta}\sum_{t\, \in \, \text{Spec}_{R}(\ket{\psi_{A}})\times \text{Spec}_{R}(\ket{\psi_{B}})}2\cosh(\beta\,t) \right.\\
 & \left. + e^{-\beta}\sum_{t \,\in\, \text{Spec}_{R}(\ket{\psi_{B}}) }2^{|A|-\nu_{A}} \cosh{(\beta\,t)}+ e^{-\beta}\sum_{t \,\in\, \text{Spec}_{R}(\ket{\psi_{A}}) }2^{|B|-\nu_{B}}\cosh{(\beta\,t)}\right].
\end{align}
Where $2^{|A|-\nu_{A}}$ (respectively $2^{|B|-\nu_{B}}$) is the cardinality of the stabilizer part of the Pauli spectrum $\text{Spec}_{S}(\ket{\psi_{A}})$ (respectively $\text{Spec}_{S}(\ket{\psi_{B}})$), and we have further used the SPF invariance under the $\mathbb{Z}_{2}$ action over the Pauli spectrum, ${\rm i.e.}$ $\text{Spec}(\ket{\psi})$ and  $-\text{Spec}(\ket{\psi})$ lead to the same SPF. 
To conclude the proof notice that for every state $\ket{\phi}$. $
 e^{-\beta}\sum_{t\, \in \,\text{Spec}_{R}(\ket{\phi})}\cosh{(\beta\,t)}=\mathcal{Z}_{\beta}(\ket{\phi})-e^{-\beta}|\text{Spec}_{S}(\phi)|\cosh{(\beta)}$ leading to the expression presented in the main text.

\subsubsection{SPF for pseudomagic states (Theorem~\ref{Theorem:Pseudorandom_SPF})}

Subset phase states (SPS) are defined via
\begin{equation}
    \ket{\psi_{f,S}}=\frac{1}{\sqrt{|S|}}\sum_{x \in S}(-1)^{f(x)}\ket{x},
\end{equation}
with $f:\{0,1\}^{n}\to \{0,1\}$ a pseudorandom function, and $S \subseteq \{0,1\}^{n}$ a pseudorandom subset of fixed length  $|S|=2^{k}$ ($k\leq n$). SPS are efficiently preparable and indistinguishable from maximal magic (Haar random) states for computationally bounded observers \cite{Aaronson/pseudoentanglement/2022,Gu/pseudomagic/2024}, namely

\begin{align}
    \left \vert  \left \vert \mathbb{E}_{\psi_{f,S} \sim \text{SPS}}\left(\ket{\psi_{f,S}}\bra{\psi_{f,S}}^{\otimes t}\right)- \mathbb{E}_{\psi \sim \text{Haar}_{n}}\left(\ket{\psi}\bra{\psi}^{\otimes t}\right) \right \vert \right \vert_{1} \leq \frac{1}{\text{Poly}(n)}.
\end{align}

By simplicity we are going to focus our analysis on the $|S|=2^{n}$ case, where there is only a single subset $S=\{0,1\}^{n}$, and the only average one has to care about is the one over the pseudo-random functions, ${\rm e.g.}$ $\ket{\psi_{f,S}} \to \ket{\psi_{f}}=\frac{1}{\sqrt{d_{n}}}\sum_{ x\in \{0,1\}^{n}} (-1)^{f(x)}\ket{x}$. We are interested in the Laplace transform of the average Pauli-spectrum density \eqref{eq:Pauli_spectrum_density}, namely 
\begin{equation}
\label{app_eqn:lap_transform}
    \mathcal{Z}_{\beta}^{\pm}({\rm SPS})\equiv d_{n}^{2}\int e^{\pm\beta x} \,\mathbb{E}_{\psi \sim \psi_{f}}[d\mu_{\ket{\psi}}(x)]=\sum_{P \in \mathcal{P}_{n}}\mathbb{E}_{\psi \sim \psi_{f}}\left[  \exp(\pm\beta \braket{\psi|P\,|\psi}) \right].
\end{equation}
Here the integration is done over the support of the average Pauli-spectrum density. The average SPF is obtained from \eqref{app_eqn:lap_transform} like 
\begin{equation}
    \mathcal{Z}_{\beta}({\rm SPS}) =\mathbb{E}_{\psi \sim \psi_{f}}[\mathcal{Z}_{\beta}(\psi)]=\frac{e^{-\beta}}{2}\left(\mathcal{Z}_{\beta}^{+}({\rm SPS})+\mathcal{Z}_{\beta}^{-}({\rm SPS})\right).
\end{equation}
By the Tableau representation of Pauli strings, ${\rm i.e.}$ $\mathcal{P}_{n} \cong \mathbb{F}_{2}^{(n|n)}$ \cite{Aronsson_gottesman/Tableau/2004}, each string $P \in \mathcal{P}_{n}$ can be uniquely decomposed as $P=i^{\#_{Y}(P)} X Z$ with $X \in \mathcal{P}_{X,n}$, and $Z \in \mathcal{P}_{Z,n}$ the Abelian sub-groups generated by single Pauli $X$ and $Z$ respectively, and $\#_{Y}(P)$ an indicator function counting the number of single Pauli $Y$ in the string $P$.

\begin{proposition}[Pauli expectation for SPS]\label{app_prop:SPS_matel}
    Given a Pauli string $P \in \mathcal{P}_{n}$, its expectation over SPS is given by 

    \begin{equation}
    \braket{\psi_{f}|P\,|\psi_{f}}=\frac{1}{d_{n}}i^{Z\cdot X}\sum_{x \in \{0,1\}^{n}}(-1)^{x \cdot Z}(-1)^{f(x)+f(x\oplus X) },
\end{equation}
with $x \cdot Z= \sum_{i}^{n} Z_{i}\,x_{i}$ the dot product between the tableau representation of the Pauli string $P_{Z}$, and the bit string $x$, and $x\oplus X=\sum_{i=1}^{n}x_{i}\oplus X_{i}$ the element-wise sum between the tableau representation of the $X$ string with $x$.
\end{proposition}

\begin{proof}
By using the tableau representation, $P=i^{\#_{Y}(P)} X Z$, where $\#_{Y}(P)=Z \cdot X$. The matrix element between a pair of bit-strings $x$, and $\bar{x}$ reads
 \begin{equation}
    \braket{\bar{x}|P|x}= i^{Z \cdot X}\braket{\bar{x}|X Z|x}=i^{Z\cdot X}(-1)^{Z \cdot x}\braket{\bar{x}|X|x}=i^{Z\cdot X}(-1)^{Z \cdot x}\braket{\bar{x}|X\oplus x}.
    \end{equation}
where the last inner product enforces the $\delta_{\bar{x},X\oplus x}$ constraint finishing the proof. 
\end{proof}
By Proposition~\ref{app_prop:SPS_matel}, the SPS Laplace transform \eqref{app_eqn:lap_transform} reads
\begin{equation}
\label{app:eq_exact_laplace}
    \mathcal{Z}^{-}_{\beta}({\rm SPS}) =\sum_{X\in \mathcal{P}_{X,n}}\sum_{Z \in \mathcal{P}_{Z,n} }\mathbb{E}_{f}\left[\exp\left(-\frac{\beta}{d_{n}}\,i^{\,Z\cdot X}\sum_{x \in \{0,1\}^{n}}(-1)^{Z\cdot x}(-1)^{f(x)+f(x\oplus X) }\right)\right],
\end{equation}
with $\mathbb{E}_{f}[\cdots]$ the Boolean function average. The main strategy to perform the average in \eqref{app:eq_exact_laplace} is to realize that the expectation corresponds to the moment generating function of the $2^{n}$  random vector $ \textbf{r}_{s}= \{r_{s}(x)=(-1)^{f(x)+f(x \oplus s)}\}$, ${\rm e.g.}$
\begin{equation}
    M(\beta; {\textbf{J}})=\mathbb{E}_{f} \left[\exp{\left({\textbf J}(X;Z) \cdot {\textbf r}_{X} \right)}\right] \quad  \text{with} \quad {\textbf J}(X;Z)_{x}=-\frac{\beta}{d_{n}} \, i^{Z \cdot X}(-1)^{Z \cdot x},
\end{equation}
Then, $ \mathcal{Z}^{-}_{\beta}({\rm SPS}) =\sum_{X\in \mathcal{P}_{X,n}}\sum_{Z \in \mathcal{P}_{Z,n}}  M(\beta; {\textbf{J}})$. 
The MGF can be accessed via its connected correlations, ${\rm e.g.}$
\begin{equation}
\label{app_eq:f_commulant_expansion}
    M(\beta; {\textbf{J}})=\exp{\left(\sum_{x}\mathbb{E}_{f}[{r}_{X}(x)]J(X;Z)_{x} +\frac{1}{2!}\sum_{x,y} \mathbb{E}_{f}[r_{X}(x)r_{X}(y)]_{c}\,J(X;Z)_{x} J(X;Z)_{y} \dots \right)}
\end{equation}
where the further dots represent higher terms in the cumulant expansion. The main technical tool for the characterization of the cumulants is the following lemma.

\begin{lemma}[connected correlations]
\label{app_lemma:connected_redemacher}
Let \(f:\{0,1\}^{n}\to \{0,1\}\) be a uniformly random function, and define $r_s(x)=(-1)^{f(x)+f(x\oplus s)}$, with \(s\neq 0^n\) a fixed bit-string. Then the \(2k\)-th order cumulant is given by
\begin{equation}
    \mathbb{E}_{f}
    \left[
        r_s(x_1)r_s(x_2)\cdots r_s(x_{2k})
    \right]_{c}
    =
    \kappa_{2k}
    \prod_{j=2}^{2k}\Delta_s(x_1,x_j),
\end{equation}
where $\Delta_s(x,y):= \mathbb{E}_{f}\left[r_s(x)r_s(y)\right]_{c} = \delta_{x,y}+\delta_{x\oplus s,y}$, and \(\kappa_{2k}\) is the \(2k\)-th order cumulant of a Rademacher random variable (${\rm e.g.}$ \(R=\pm 1\) with equal probability) , namely
\[
    \kappa_{2k}
    =
    \left.
    \frac{d^{2k}}{dt^{2k}}\log\cosh t
    \right|_{t=0}.
\]
\end{lemma}

\begin{proof}
For \(f\) uniformly random, the variables $\eta_x:=(-1)^{f(x)}$ are independent Rademacher variables. ${\rm i.e.}$,
\begin{equation}
    \mathbb{E}_{f}[\eta_x]=0, \quad  \mathbb{E}_{f}[\eta_x\eta_y]=\delta_{x,y}.
\end{equation}
Moreover, $r_s(x)=\eta_x\eta_{x\oplus s}$. Since \(s\neq 0^n\), the map \(x\mapsto x\oplus s\) has no fixed points and partitions \(\{0,1\}^n\) into disjoint two-point orbits
\begin{equation*}
        [x]_s:=\{x,x\oplus s\}.
\end{equation*}
Because \(s\oplus s=0^n\), we have $ r_s(x\oplus s) = \eta_{x\oplus s}\eta_{x\oplus s\oplus s} = r_s(x)$. Thus \(r_s(x)\) depends only on the \(s\)-edge \([x]_s\). Define $ R_{[x]_s}:=r_s(x)$. If \([x]_s\neq [y]_s\), then the pairs
\[
    \{x,x\oplus s\}
    \qquad\text{and}\qquad
    \{y,y\oplus s\}
\]
are disjoint. Hence \(R_{[x]_s}\) and \(R_{[y]_s}\) are functions of disjoint
sets of independent Rademacher variables, and are therefore independent.
Each \(R_{[x]_s}\) is itself a Rademacher variable. Now consider the cumulant generating function
\begin{equation*}
K(J_1,\dots,J_{2k})=\log\mathbb{E}_{f}\left[\exp\left(\sum_{i=1}^{2k}J_i r_s(x_i)\right) \right]=\sum_{e}
    \log\cosh
    \left(
        \sum_{i:\,[x_i]_s=e}J_i
    \right).
\end{equation*}
Here we have grouped the the sources according to the \(s\)-edge to which \(x_i\) belongs, and the sum runs over all the \(s\)-edges \(e\). Then,

\[
    \mathbb{E}_{f}
    \left[
        r_s(x_1)\cdots r_s(x_{2k})
    \right]_c
    =
    \left.
    \partial_{J_1}\cdots \partial_{J_{2k}}
    K(J_1,\dots,J_{2k})
    \right|_{J=0}.
\]
A mixed derivative involving all sources $J_1,\dots,J_{2k}$ can be nonzero only if all these sources appear in the same term of the sum over $s$-edges. This happens if and only if
\begin{equation}
    x_j\in \{x_1,x_1\oplus s\} \qquad \text{for all } j=2,\dots,2k.
\end{equation}
Since \(s\neq 0^n\), this condition is exactly encoded by
\[
    \prod_{j=2}^{2k}\Delta_s(x_1,x_j),
\]
with $\Delta_s(x_1,x_j)=\delta_{x_1,x_j}+\delta_{x_1\oplus s,x_j}$. On the support of this product, all  $r_s(x_1),r_s(x_2),\dots,r_s(x_{2k})$ coincide with a single Rademacher random variable \(R\). Hence the connected correlator reduces to the \(2k\)-th cumulant of \(R\):
\begin{equation}
     \mathbb{E}_{f}
    \left[
        r_s(x_1)\cdots r_s(x_{2k})
    \right]_c
    =
    \kappa_{2k}.
\end{equation}
Since $\log \mathbb{E}[e^{tR}] =\log\cosh t,$ we have $\kappa_{2k} =\left.\frac{d^{2k}}{dt^{2k}}\log\cosh t \right|_{t=0}.$ Combining the support condition with the value of the cumulant gives
\begin{equation}
     \mathbb{E}_{f}
    \left[
        r_s(x_1)r_s(x_2)\cdots r_s(x_{2k})
    \right]_c
    =
    \kappa_{2k}
    \prod_{j=2}^{2k}\Delta_s(x_1,x_j),
\end{equation}
which proves the claim. Notice also that $r_{x\oplus s}(s) =r_{x}(s)$, so the connected correlation is invariant under $s$-shifts. Thus, the connected correlator should be the product of the fundamental invariants $\mathbb{E}_{f}[r_{s}(x)r_{s}(y)]_{c}$.
\end{proof}

Once with the result in Lemma~\ref{app_lemma:connected_redemacher} we are in place to prove the main theorem, namely the SPF for SPS. First we introduce the following intermediate result 

\begin{lemma}[MGF for SPS]
For every $X\neq 0^{n}$, the MGF in the cumulant expansion reads
\begin{equation}
\label{app_eqn:exact_SPS_MGF}
    M(\beta; \mathbf{J}) = \exp\left(d_n\omega(\beta)(1+(-1)^{Z\cdot X})\right), \qquad X\neq0^n,
\end{equation}
with $\omega(\beta) = \sum_{k\geq1}\frac{1}{(2k)!} \left(\frac{\beta}{d_n}\right)^{2k} \kappa_{2k}\frac{2^{2k-1}}{2} = \frac{1}{4}\log\cosh\left(\frac{2\beta}{d_n}\right)$.
\end{lemma}

\begin{proof}
As odd cumulants vanish, only the even cumulants contribute. Using
the connected-correlation lemma (Lemma~\ref{app_lemma:connected_redemacher}), the $2k$th contribution to
the cumulant exponent is
\begin{align}
    &\frac{1}{(2k)!}
    \sum_{x_1,\dots,x_{2k}}
    \mathbb{E}_{f}
    [r_X(x_1)\cdots r_X(x_{2k})]_{c}
    J(X;Z)_{x_1}\cdots J(X;Z)_{x_{2k}}
    \nonumber\\
    &=
    \frac{1}{(2k)!}
    \left(\frac{\beta}{d_n}\right)^{2k}
    (-1)^{k Z\cdot X}\kappa_{2k}
    \sum_{x_1,\dots,x_{2k}}
    (-1)^{Z\cdot(x_1+\cdots+x_{2k})}
    \prod_{j=2}^{2k}\Delta_X(x_1,x_j).
\end{align}
The product of deltas enforces
$x_j\in\{x_1,x_1\oplus X\}$ for every $j\geq2$. Writing $a$ for the
number of indices $j\geq2$ equal to $x_1$ and $b$ for the number equal
to $x_1\oplus X$, with $a+b=2k-1$, the simplifies to
\begin{equation}
    (-1)^{Z\cdot(x_1+a x_1+b(x_1\oplus X))}.
\end{equation}
Since $a+b$ is odd, the two possible parity cases give
\begin{equation}
    (-1)^{Z\cdot(x_1+a x_1+b(x_1\oplus X))}
    =
    \begin{cases}
    1, & a+1 \text{ and } b \text{ even},\\
    (-1)^{Z\cdot X}, & a+1 \text{ and } b \text{ odd}.
    \end{cases}
\end{equation}
Both parity classes contain $2^{2k-2}$ binomial terms, and the remaining
$x_1$ sum gives a factor $d_n$. Thus the $2k$th contribution is
\begin{equation}
    \frac{1}{(2k)!}
    \left(\frac{\beta}{d_n}\right)^{2k}
    (-1)^{k Z\cdot X}\kappa_{2k}
    \frac{2^{2k-1}}{2}d_n(1+(-1)^{Z\cdot X}).
\end{equation}
If $Z\cdot X=1$, this contribution vanishes. whereas for $Z\cdot X=0$, it becomes $ d_n \frac{1}{(2k)!}\left(\frac{\beta}{d_n}\right)^{2k} \kappa_{2k}2^{2k-1}$. Resumming over $k$ gives
\begin{equation}
    \log M(\beta;\mathbf{J})
    =
    d_n\omega(\beta)(1+(-1)^{Z\cdot X}),
\end{equation}
with $\omega(\beta) =\frac{1}{4}\log\cosh\left(\frac{2\beta}{d_n}\right)$ which proves the claim. 
\end{proof}

Now, to evaluate the SPF for SPS, we split the Pauli sum into three sectors. First, the identity Pauli $X=0^n$, $Z=0^n$ gives $\exp(-\beta)$. Second, the non-identity diagonal Paulis, $X=0^n$ and $Z\neq0^n$, have vanishing expectation and contribute $d_n-1$. Finally, for $X\neq0^n$, we split the $Z$ sum according to the parity of $Z\cdot X$. For each fixed nonzero $X$, exactly $d_n/2$ choices of $Z$ satisfy $Z\cdot X=0$, and $d_n/2$ choices satisfy $Z\cdot X=1$. Therefore
\begin{align}
    \mathcal{Z}^{-}_{\beta}({\rm SPS})
    &=
    e^{-\beta}
    +(d_n-1)
    +
    \frac{(d_n-1)d_n}{2}
    \left[\cosh\left(\frac{2\beta}{d_n}\right)\right]^{d_n/2}
    +
    \frac{(d_n-1)d_n}{2}.
\end{align}

The expression for $\mathcal{Z}^{+}_{\beta}({\rm SPS})$ is obtained by
$\beta\mapsto-\beta$. Since $\omega(\beta)=\omega(-\beta)$, the only
term changed by this replacement is the identity contribution:
$e^{-\beta}\mapsto e^{\beta}$. Hence
\begin{align}
    \mathcal{Z}_{\beta}({\rm SPS})
    &=
    \frac{e^{-\beta}}{2}
    \left(
        \mathcal{Z}^{+}_{\beta}({\rm SPS})
        +
        \mathcal{Z}^{-}_{\beta}({\rm SPS})
    \right) \nonumber\\
    &=
    e^{-\beta}
    \left[ \frac{(d_{n}-1)(d_{n}+2)}{2}+\cosh\beta
    +\frac{(d_n-1)d_n}{2}
        \cosh{\left(\frac{2\beta}{d_{n}}\right)^{\frac{d_{n}}{2}}}
    \right].
\end{align}
A direct consequence of the previous theorem is the exact Pauli-spectrum density for subset phase states.

\begin{corollary}[Average SPS Pauli-spectrum density]
\label{App_corollary:SPS_Pauli_spectrum}
For $\ket{\psi}\sim{\rm SPS}$ with $|S|=2^n$, the normalized average Pauli-spectrum density $d\mu_{\rm SPS}(x)=\mathbb{E}_{\ket{\psi}\sim{\rm SPS}}
[d\mu_{\ket{\psi}}(x)]$ is given by
\begin{align}
    d\mu_{\rm SPS}(x)=\frac{1}{d_n^2}\delta(x-1) +
    \frac{(d_n-1)(1+d_n/2)}{d_n^2}\delta(x)+
    \frac{(d_n-1)d_n}{2d_n^2\,2^{d_n/2}}
    \sum_{k=0}^{d_n/2}\binom{d_n/2}{k}
\delta\left(x-\left(\frac{4k}{d_n}-1\right)\right).
\end{align}
\end{corollary}

\begin{proof}
The proof follows directly from taking the inverse Laplace transform of $e^{\beta}\mathcal{Z}_{\beta}^{-}(\rm SPS)$.
\end{proof}

Notice that opposite to the $d\mu_{{\rm Haar}_{n}}$ case, the SPS density doesn't reduce to an smooth function after the ensemble average. However, both Haar- and SPS- SPF's are smooth. This is an advantage of the SPF over the Pauli spectrum point measure. Lastly, we present numerical results supporting the exact SPS Pauli-spectrum density and its corresponding SPF (see Fig~\ref{fig:SPS_numerical_validation})

\begin{figure*}[t]
    \centering
    \begin{minipage}[t]{0.495\textwidth}
        \centering
        \makebox[\linewidth][l]{\textbf{(a)}}\\[0.1em]
        \includegraphics[width=\linewidth]{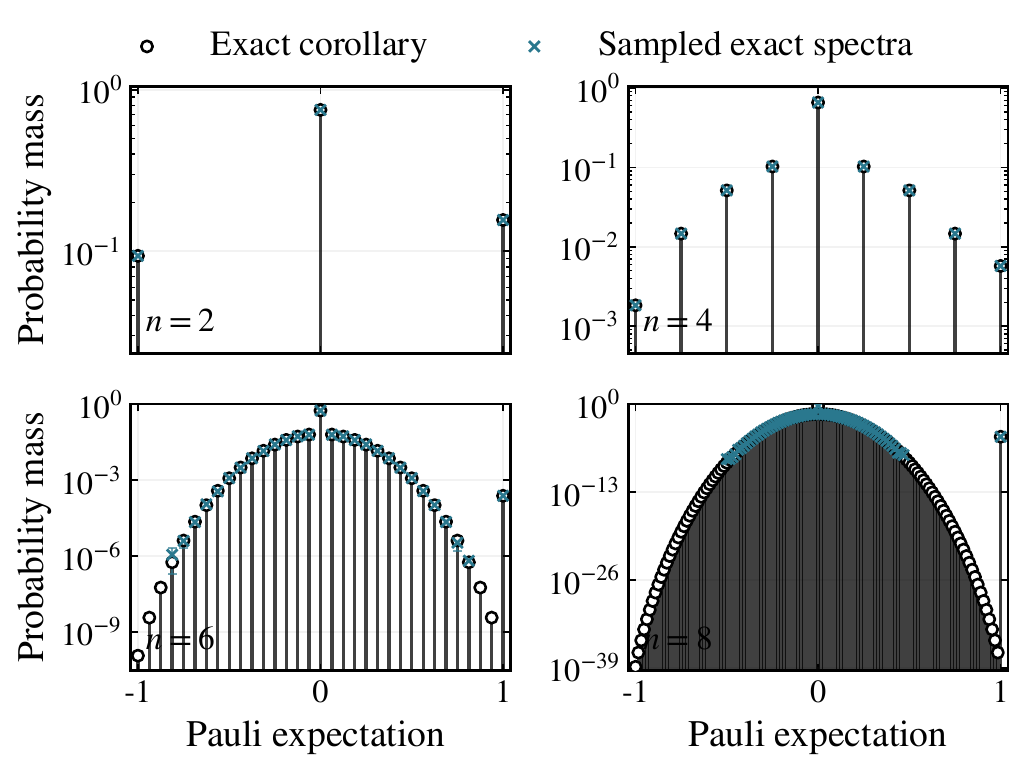}
    \end{minipage}\hfill
    \begin{minipage}[t]{0.495\textwidth}
        \centering
        \makebox[\linewidth][l]{\textbf{(b)}}\\[0.1em]
        \includegraphics[width=\linewidth]{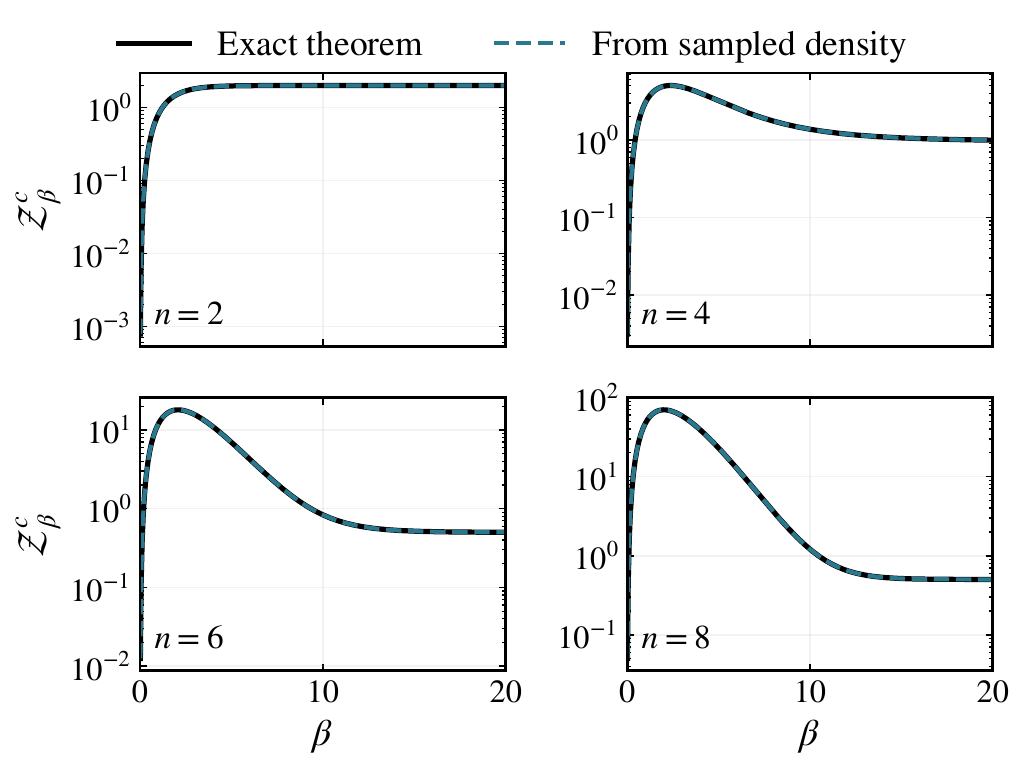}
    \end{minipage}
    \caption{Numerical validation of the exact Pauli-spectrum density and stabilizer partition function (SPF) for full-support subset phase states (SPS), with $|S|=d_n=2^n$ and $n\in\{2,4,6,8\}$. (a)~Normalized Pauli-spectrum probability masses obtained from exact evaluation of all $d_n^2$ Pauli expectation values for each SPS instance (teal crosses), compared with the average density in Corollary~\ref{App_corollary:SPS_Pauli_spectrum} (black stems and open circles). The averages for $n=2$ and $n=4$ exhaust all Boolean functions, while those for $n=6$ and $n=8$ use $4096$ and $2048$ independent uniformly random Boolean functions, respectively. (b)~Core SPF reconstructed from the sampled Pauli-spectrum density (teal dashed curves), compared with the exact ensemble-average expression in Theorem~\ref{Theorem:Pseudorandom_SPF} (black solid curves).}
    \label{fig:SPS_numerical_validation}
\end{figure*}

\subsection{Additional numerical simulations \label{supp:numerical_sim}}

\subsubsection{Fidelity Upper bounds}\label{sm:fidelity_ub}

The upper bound on the stabilizer fidelity as obtained in Proposition~\ref{proposition:stab_fidellity_upper_bound} is given by
\begin{equation*}
    F_{\rm{Stab}}(\psi) \leq 1- \frac{1}{2}\left(\frac{ \cosh(\beta)-1}{\beta\sinh{\beta}}2^{-\mathcal{W}_{\beta}(\psi)}\mathcal{W}_{\beta}(\psi)\right)^{2} \;\; \forall\, \beta\geq 0.
\end{equation*}
For illustration of how the upper bound on the stabilizer fidelity changes with the parameter $\beta$, we now plot in Fig.~\figref{fig:fidelity_ub}, the upper bounds obtained from this result by taking the $\ket{T}$ state and the $\ket{H}$ state families for the one=, two- and four-qubit cases.
\begin{figure}
    \centering
    \includegraphics[width=0.7\linewidth]{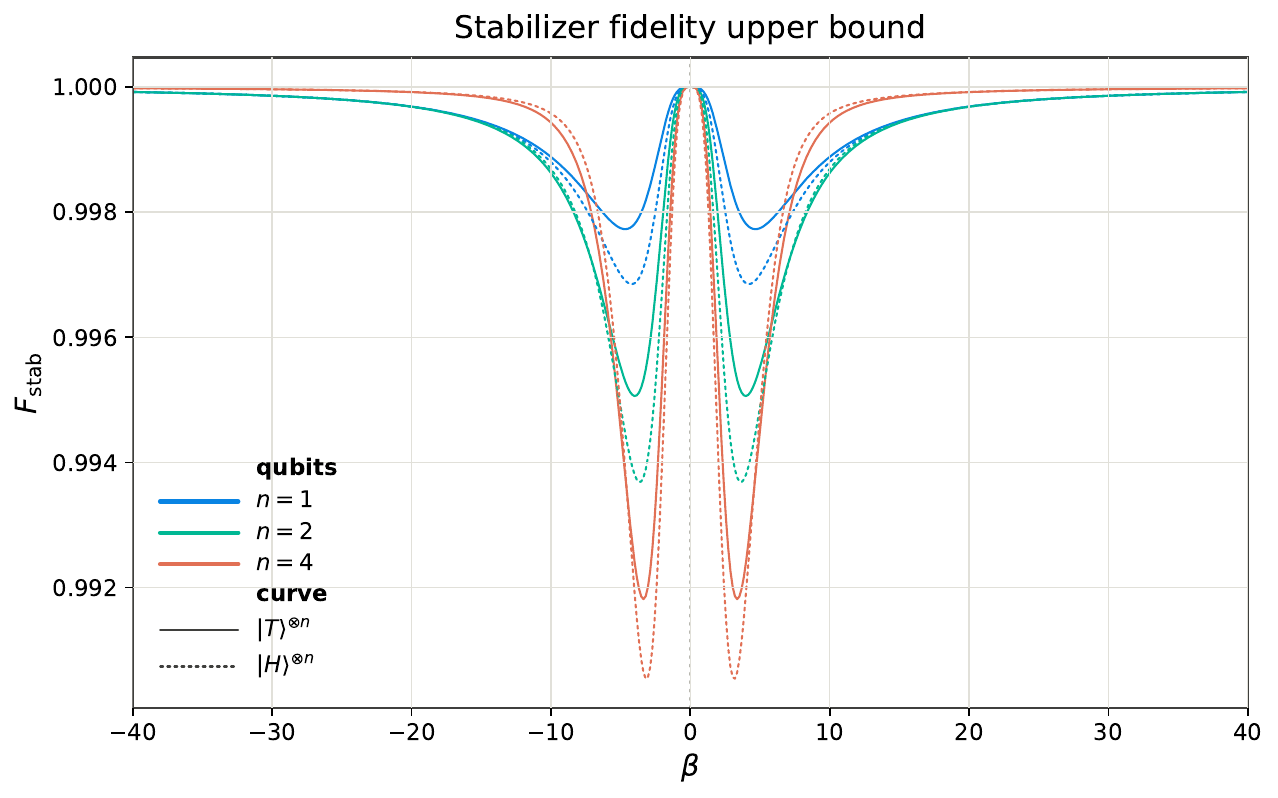}
    \caption{Stabilizer fidelity upper bounds as obtained from Proposition~\ref{proposition:stab_fidellity_upper_bound}.}
    \label{fig:fidelity_ub}
\end{figure}
The optimization over $\beta$ as is required in Section~\ref{sec:low-rank-simulation} then just requires looking for the visible minimum in the plot.

\subsubsection{Transverse Field Ising}\label{sm:Ising}
Fig. \ref{fig:IsingSRE2} compares the finite-size behavior of the stabilizer Rényi entropy (SRE) with our stabilizer-work construction, Fig. \ref{fig:IsingSPF2}, across the transverse-field Ising transition. The SRE results illustrate an intrinsic limitation of fixed-order Rényi diagnostics: for each chosen index $\alpha$, $M_\alpha(\psi)$ probes one specific moment of the underlying stabilizer distribution. Accordingly, changing $\alpha$ from $2$ to $5$ to $10$ changes both the magnitude and detailed shape of the response. Although all three indices display enhanced non-stabilizerness near the transition and systematic growth with system size, each $\alpha$ provides only a single projection of the information contained in the state. No individual SRE therefore characterizes the full hierarchy of stabilizer fluctuations.
On the other hand, the stabilizer work provides a substantially richer description as it incorporates contributions from all even moments simultaneously. Consequently, varying $\beta$ systematically changes their relative weighting rather than selecting an isolated moment.

The numerical results make this distinction explicit. For $\beta=0.5,2$ and $10$, the stabilizer work develops a pronounced maximum on the ordered side approaching the critical point, whose amplitude increases systematically with the system size $L$. At $\beta=100$, the response changes qualitatively into a broad, nearly saturated plateau spanning the critical region, followed by only a weak decrease at larger $h/g$. Thus different $\beta$ regimes expose different sectors of the same underlying stabilizer distribution.
These results establish a clear conceptual distinction: fixed-$\alpha$ SRE is \textit{moment-resolved}, whereas stabilizer work is \textit{moment-generating and distribution-sensitive}. The $\beta$-dependent family therefore retains substantially more information about the reorganization of magic across a quantum phase transition than any single Rényi index.

\begin{figure*}[ht]
    \centering
  \includegraphics[width=\textwidth]{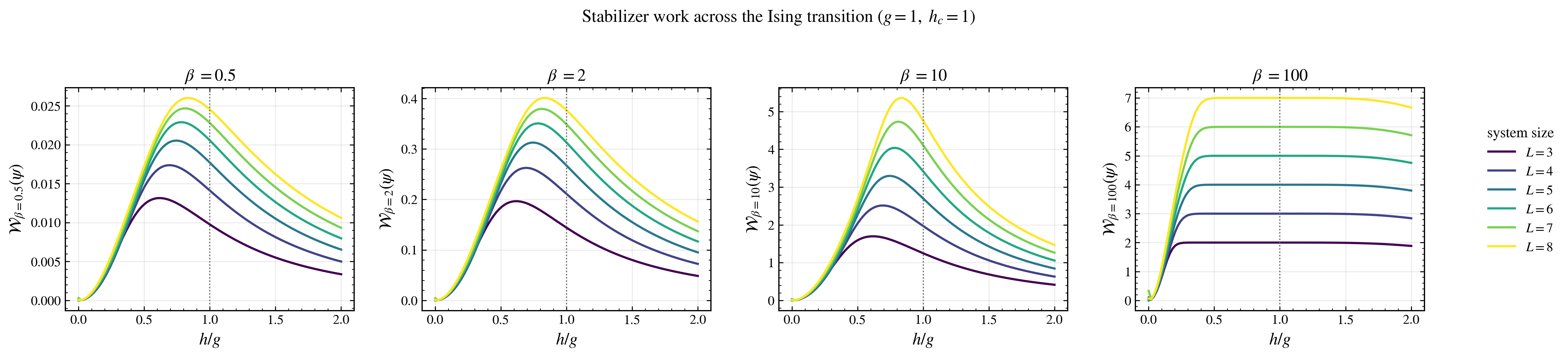}  
  \caption{Stabilizer work through the lens of quantum Ising model. Stabilizer work is plotted against physical parameter $h/g$ for four different tunable parameters $\beta=0.5,2,10,100$. }
  \label{fig:IsingSPF2}
\end{figure*}

\begin{figure*}[hb]
    \centering
  \includegraphics[width=\textwidth]{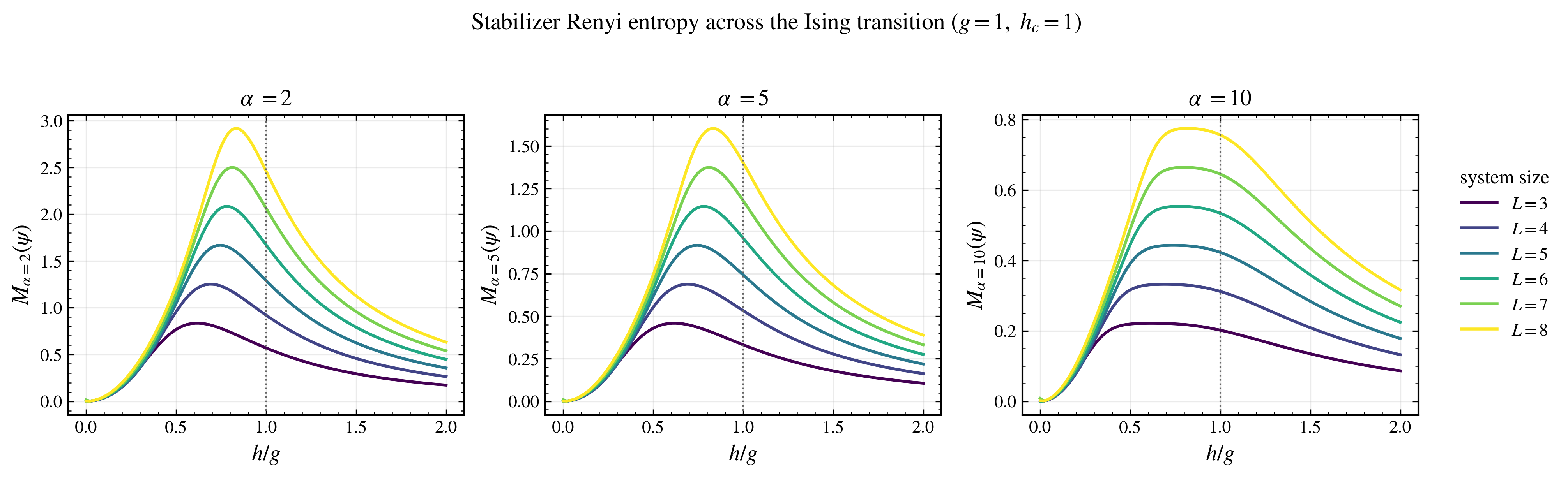}  
  \caption{Stabilizer R\'enyi Entropy through the lens of quantum Ising model. Stabilizer R\'enyi entropy ($M_\alpha$) is plotted against physical parameter $h/g$ for three different $\alpha$ parameters $\alpha=2,5,10$.}
  \label{fig:IsingSRE2}
\end{figure*}

\subsubsection{$\mathrm{H}_2$ dissociation in the 6-31G basis \label{supp:h2_631G}}

\begin{figure*}[t]
    \centering
    \includegraphics[width=0.6\textwidth]{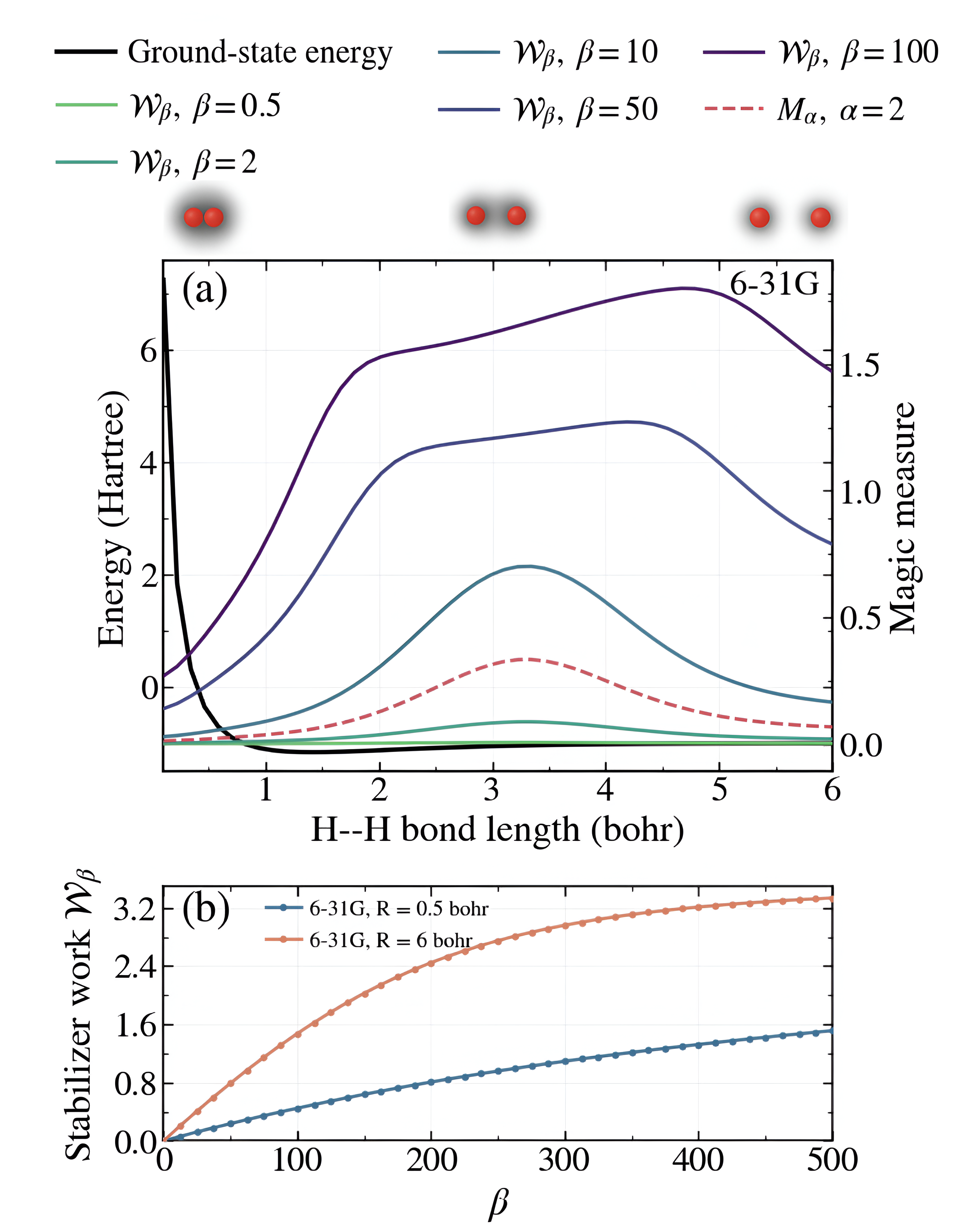}
    \caption{Stabilizer work and nonstabilizerness along the $\mathrm{H}_2$ dissociation curve in the split-valence 6-31G basis (eight spin-orbitals, eight qubits), to be compared with the minimal STO-3G results of \figref{fig:h2_dissoc}. (a)~Ground-state energy (left axis, black solid line) and the magic measures (right axis) as functions of the H--H bond length: stabilizer work $\mathcal{W}_{\beta}$ for $\beta\in\{0.5,2,10,50,100\}$ (solid colored lines) and the $\alpha=2$ stabilizer R\'{e}nyi entropy $M_{\alpha=2}$ (red dashed line), all computed from the exact ground-state statevector obtained by full diagonalization of the Jordan--Wigner qubit Hamiltonian. As in the minimal basis, the nonstabilizerness is small near equilibrium and develops a pronounced peak in the strongly correlated region; in contrast to STO-3G, however, the peak is displaced to a slightly larger bond length and an appreciable amount of magic is retained toward the dissociation limit, most visibly in the large-$\beta$ curves. (b)~Stabilizer work $\mathcal{W}_{\beta}$ as a function of $\beta$ for the two representative bond lengths $R=0.5$ and $R=6$\,bohr. The monotonic growth and smooth convergence toward the $\beta\to\infty$ asymptote mirror the STO-3G case, but the ordering of the two geometries is reversed relative to \figref{fig:h2_dissoc}(b): in 6-31G the near-dissociated geometry ($R=6$\,bohr) carries the larger stabilizer work.}
    \label{fig:h2_dissoc_631G}
\end{figure*}

In the main text (Sec.~\ref{subsec:h2_dissoc}, \figref{fig:h2_dissoc}) we analyzed the nonstabilizerness of the $\mathrm{H}_2$ ground state along its dissociation coordinate within the minimal STO-3G basis, and we cautioned that several of the observed features may be artifacts of the incompleteness of that basis. Here we substantiate that caution by repeating the analysis in the larger split-valence 6-31G basis. For hydrogen, 6-31G assigns two contracted basis functions per atom, so that $\mathrm{H}_2$ is described by four spatial molecular orbitals, that is, eight spin-orbitals and hence an eight-qubit problem ($d_{n}=2^{8}=256$) under the Jordan--Wigner transformation, as compared with the four-qubit problem of the minimal basis. As before, the ground-state statevector $\ket{\psi(R)}$ is obtained at each bond length by exact diagonalization of the Jordan--Wigner qubit Hamiltonian (full configuration interaction within the basis), and both $\mathcal{Z}_{\beta}(\psi(R))$ and $\mathcal{W}_{\beta}(\psi(R))$ are evaluated from the full Pauli spectrum. The results are shown in \figref{fig:h2_dissoc_631G}.

Several qualitative features carry over unchanged from the minimal basis. As seen in \figref{fig:h2_dissoc_631G}(a), the magic measures are small near the equilibrium geometry, where the wavefunction retains a dominant single-determinant character, and rise to a pronounced peak in the intermediate, strongly correlated region of the dissociation coordinate. The role of $\beta$ as a moment-resolving parameter is likewise unchanged: small $\beta$ produces a smooth, coarse probe dominated by the lowest Pauli-spectrum moments, while increasing $\beta$ sharpens the features and, at large $\beta$, tracks the overall shape of $M_{\alpha=2}$. Panel (b) confirms that, for each fixed geometry, $\mathcal{W}_{\beta}$ grows monotonically with $\beta$ and saturates toward its $\beta\to\infty$ asymptote, with the ordering between the two representative geometries preserved across the entire range of $\beta$. In this sense the stabilizer-work diagnostic behaves consistently across the two bases.

The differences, however, are equally instructive. First, the peak in nonstabilizerness is displaced to a somewhat larger bond length in 6-31G than in STO-3G. Second, and more strikingly, the magic no longer collapses toward the dissociation limit: the high-$\beta$ stabilizer work rises and remains elevated at large $R$ rather than decaying, so that an appreciable amount of nonstabilizerness is retained well into the separated-atom regime [\figref{fig:h2_dissoc_631G}(a)]. This reversal is made explicit in panel (b): whereas in the minimal basis the compressed geometry $R=0.5$\,bohr carried the larger stabilizer work, in 6-31G the near-dissociated geometry $R=6$\,bohr carries the larger value, so that the ordering of magic between the two fixed geometries is inverted relative to STO-3G.

A plausible origin of this behavior is that the enlarged one-particle space leaves residual configuration mixing in the bond-breaking and separated-atom regions that is absent in the minimal basis: in STO-3G the dissociation limit reduces to an essentially two-configuration covalent singlet of low magic, whereas the additional radial flexibility of 6-31G prevents such a simple reduction. We emphasize, however, that the nonstabilizerness extracted in this way is not invariant under the choice of one-particle basis or fermion-to-qubit encoding, so that part of the basis dependence (and in particular the magnitude of the retained magic at large $R$) may reflect these representational choices rather than an intrinsic property of the molecule. We therefore refrain from attaching physical significance to the dissociation-limit magic at this stage.

In summary, the comparison between STO-3G and 6-31G confirms that the stabilizer partition function delivers a consistent qualitative diagnostic of nonstabilizerness in both bases, while simultaneously demonstrating that the finer features, such as the peak position, the dissociation-limit magic, and even the ordering of magic between two fixed geometries are basis-set sensitive. Since neither basis approaches the complete-basis-set limit, these observations should be regarded as motivating, rather than settling, the question of how much nonstabilizerness the physical $\mathrm{H}_2$ molecule retains upon dissociation. A converged answer awaits systematic basis-set extrapolation, which we leave for future work.

\end{document}